\documentclass{article}
\usepackage{amsfonts, amsmath, amssymb, amsgen, amscd}
\usepackage{cite}
\usepackage{tikz}
\usepackage{setspace}
\usepackage{amsmath,amssymb}
\usepackage{adjustbox}
\usepackage[flushleft]{threeparttable}
\usepackage{array,booktabs,makecell}
\usepackage[font=small]{caption}
\usepackage{verbatim}
\usepackage[top=2cm, bottom=2cm, left=2cm, right=2cm]{geometry}
\def\ad{\mathop{\rm ad}\nolimits}
\usepackage[colorlinks]{hyperref}
\usepackage{bibentry}
\newtheorem{theorem}{Theorem}

\newtheorem{lemma}[theorem]{Lemma}

\newtheorem{proposition}[theorem]{Proposition}
\newtheorem{remark}[theorem]{Remark}

\newenvironment{proof}[1][Proof]{\noindent\textbf{#1.} }{\ \rule{0.5em}{0.5em}}
\usepackage{booktabs}
\usepackage[all]{xy}
\usepackage{siunitx}
\newcommand{\G}[1]{\mathfrak{#1}}
\def\ot{\otimes}

\newcommand{\C}[1]{\mathcal{#1}}
\def\rrbiprod{{\cdot\kern-.33em\triangleright}}
\def\lrbiprod{{\cdot\kern-.44em\triangleleft}}
\begin{document}

\title{On Extensions, Lie-Poisson Systems, and Dissipations}
\maketitle
\begin{center}
\begin{large}
\author{Oğul Esen*\footnote{oesen@gtu.edu.tr}, Gökhan Özcan*\footnote{gokhanozcan@gtu.edu.tr} and Serkan Sütlü**\footnote{serkan.sutlu@isikun.edu.tr}
\bigskip \\
*Department of Mathematics, Gebze Technical University, \\ 41400 Gebze-Kocaeli, Turkey \bigskip \\ 
** Department of Mathematics, I\c{s}ık University, \\ 34980 \c{S}ile-\.{I}stanbul, Turkey\\ 
}
\end{large}
\end{center}

\begin{abstract} On the dual space of \textit{extended structure}, equations governing the collective motion of two mutually interacting Lie-Poisson systems are derived. By including a twisted 2-cocycle term, this novel construction is providing the most general realization of (de)coupling of Lie-Poisson systems. A double cross sum (matched pair) of 2-cocycle extensions are constructed. The conditions are determined for this double cross sum to be a 2-cocycle extension by itself. On the dual spaces, Lie-Poisson equations are computed. We complement the discussion by presenting a double cross sum of some symmetric brackets, such as double bracket, Cartan-Killing bracket, Casimir dissipation bracket, and Hamilton dissipation bracket. Accordingly, the collective motion of two mutually interacting irreversible dynamics, as well as mutually interacting metriplectic flows, are obtained. The theoretical results are illustrated in three examples. As an infinite-dimensional physical model, decompositions of the BBGKY hierarchy are presented. As finite-dimensional examples, the coupling of two Heisenberg algebras and coupling of two copies of $3D$ dynamics are studied. 


\end{abstract}

\tableofcontents

\setlength{\parskip}{4mm}
\onehalfspacing

\section{Introduction}

It is a well-known fact that two mutually interacting dynamical/mechanical systems, when coupled, cannot preserve their individual motions in the collective system. This is manifested in the equation of motion of the collective system as the existence of additional terms, to those belonging to the individual subsystems. It is the Hamiltonian realization of this collective motion, of two mutually interacting physical systems, that we study in the present paper. In order to be able to present the mutual interactions as Lie group / Lie algebra actions, we shall consider the systems whose configuration spaces are Lie groups, \cite{AzIz98,baurle1997,OLV,SaWe13}.  

If the configuration space of a physical system admits a Lie group structure, then the reduced Hamiltonian dynamics can be achieved on the dual space of the Lie algebra \cite{AbMa78,Av, LiMa12, MaRa13} since it carries a natural Poisson structure, called the Lie-Poisson bracket. Many physical systems fit into this geometry, such as the rigid body, fluid and plasma theories \cite{Ho08,HoScSt09}. Coupling two different characters of a physical system, such as the fluid motion under the EM field or the rigid body motion under the gravity \cite{Rati80}, are particular instances of coupled systems where only one-sided interaction is allowed. Such systems have been studied, in literature, by the semidirect product theory, which has been successfully established both in Lagrangian dynamics, see for example \cite{Cendra98,marsden1998symplectic}, and Hamiltonian dynamics, see for example \cite{marsden1984semidirect,marsden1984reduction}, see also \cite{VaPaEs20}.

\textbf{Double cross sum (matched pair) Lie groups.} Semidirect product Lie group constructions allow only one-sided action. As such, mutual actions are beyond the scope of the semidirect product theory. On the other hand, the double cross sum (matched pair) theory of \cite{Ma90, Ma902, Ma20} rises over two Lie groups with mutual actions, subject to certain compatibility conditions. As a result, the Cartesian product of two such Lie groups becomes a Lie group by itself, through a suitable multiplication built over the mutual actions, called the ``double cross sum'' of these Lie groups \cite{Ma90, Ma902, Ma20}. A double cross product Lie group, then, contains the constitutive Lie groups as trivially intersecting Lie subgroups. Historically, generalizations of the semidirect product theory may be traced back to \cite{Sz50,Sz58,szep1962sulle,Za42} under the name of the Zappa-Sz\'ep product, see also \cite{Br05}. 
The idea was revived later in \cite{Mack72}, as the ``product of subgroups'' (and as the ``double Lie group'' in \cite{LuWe90}), from the representation theory point of view. However, the name ``matched pair'' was first used in \cite{Sing70,Sing72} for the Hopf algebra extensions. A pair of groups were referred to as ``matched pairs'', for the first time, in \cite{Ta81} as yet another incarnation of the intimate relation between Hopf algebras and groups. We add that double cross sum (matched pair) Lie algebras are studied in \cite{KoMa88} under the name of ``twilled extension''. See \cite{OS21} for more details and a list of applications of the theory both in mathematics and physics. 

\textbf{Dynamics on double cross sums (matched pairs).} Double cross sum (matched pair) of Lie groups (algebras) provides a promising geometrical framework for studying the collective motion of two mutually interacting physical systems. However, the Hamiltonian (Lie-Poisson) theory over double cross sum groups (and their double cross sum Lie algebras) is developed only recently in \cite{OS17}. This Hamiltonian matched pair theory has immediately found applications in kinetic theory \cite{EsPaGr17}, and fluid theories \cite{EsGRGuPa19}. The Lagrangian counterpart of the double cross sums has been developed for the first order Euler-Poincar\'{e} theories in \cite{OS16}, and for the higher order theories in \cite{EsKuSu20}. For discrete Lagrangian dynamics in the framework of Lie groupoids, the matched pair approach has found applications in \cite{OS21}. 

Evidently, one may use the matched pair strategy to decouple a physical system into two of its interacting subsystems as well. So, one may study the constitutive subsystems in order to examine the whole system. In \cite{OS20}, it has been shown that non-relativistic collisionless Vlasov's plasma can be written as a non-trivial double cross sum (matched pair) of its two subdynamics. One of the constitutive subdynamics is the compressible Euler fluid, while the other is the kinetic moments of order $>2$ of the plasma density function. In this exhibition, the double cross sum (matched pair) decomposition is in harmony with the nature of physics as well. 

One of the motivations of this study is to find a proper algebraic/geometric framework in order to decompose the dynamics of Vlasov's kinetic moments at any order (see \cite{GiHoTr,GiHoTr08} for the geometry of kinetic moments). Such a trial is hopeless even in the realm of matched pair Hamiltonian dynamics since the graded character of the underlying Lie algebra manifests only two Lie subalgebra decompositions. One is to cut after the zeroth moment, and the other is to cut after the first moment. Any other cut yields a Lie subalgebra with a complementary space, which is merely a subspace. 
Even though the higher order cuts attract deep attention, the literature seems to contain a gap in this issue. Regarding the cut from the second moment, we refer the reader to \cite{EsGRGuPa19} for the $10$-moment kinetic theory which paves the way towards the whole Grad hierarchy \cite{grad1965boltzmann} including the entropic moments \cite{grmela2017hamiltonian}. 
We also refer to an incomplete list \cite{Le96,PeChMoTa15,Pe90} for the related works on the kinetic moments. We shall address this issue as the first goal of this paper. Since our approach is pure algebraic/geometric, we shall need to determine a proper Lie algebra extension before we dive into the details of this goal.

\textbf{Extended structures.} Lie algebra actions are not the only sources to couple two Lie algebras. More generally, Lie algebra extensions may also be formed by the weak-actions and (twisted) 2-cocycles as well. An ``extended structure'', introduced in  \cite{Agormili14, AgoreMili-book}, provides a method to couple a Lie algebra and a vector space in a universal way. More precisely, a Lie algebra and a vector space are coupled using the action of the Lie algebra on the vector space, and the corresponding weak \textit{action}  (encoded by a twisted 2-cocycle) of the linear space on the Lie algebra. 
From the decomposition point of view, this corresponds to the decomposition of a Lie algebra into a Lie subalgebra and its vector space complement. In particular, if the complementary space happens to be a Lie subalgebra, then the extended structure reduces to a double cross sum. Postponing the details into Section \ref{Sec-MP}, we shall for now be content with the following diagram.
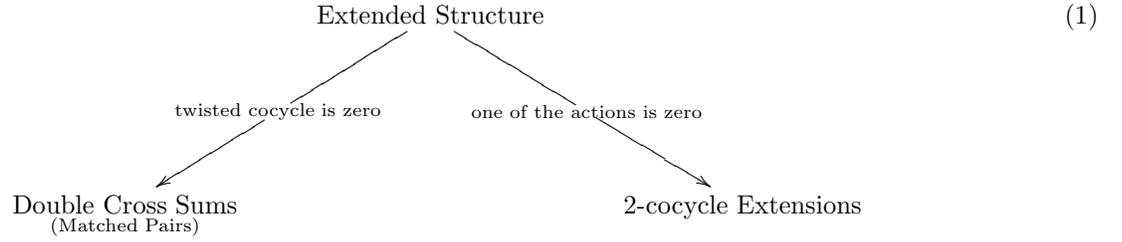
\begin{equation}
\xymatrix{
&\text{Extended Structure}  \ar[ddl]|{\text{twisted cocycle is zero}}
\ar[ddr]|{\text{one of the actions is zero}}
\\ \\
\underset{\text{(Matched Pairs)}}{\text{Double Cross Sums}}  && \text{2-cocycle Extensions}}
\end{equation}
That is, the theory of extended structures accommodates both the matched pair theory and the theory of 2-cocycle extensions as particular instances.

\textbf{The first goal of the present work.}  
For the semidirect product Hamiltonian theory, introductions of the central extensions have already been studied, see for example, in \cite{MaMiPe98}. But such an extension is missing for the matched pair Hamiltonian theory. In accordance with this observation, the first goal of the present work is to fill this gap with a theory of Hamiltonian (Lie-Poisson) dynamics on the dual of \textit{extended structures}. This will be achieved in Section \ref{LP-Ext-Sec}. The rich geometry of extended structures enables us to couple a Lie-Poisson bracket with some external variables interacting in all possible ways. From the decomposition point of view, this corresponds to the decomposition of Lie-Poisson dynamics to one of its Lie-Poisson subdynamics, and a complementary subsystem, which is not necessarily Lie-Poisson. So, the Lie-Poisson bracket and the Lie-Poisson equations provided in Subsection \ref{2-LP-Coec} evidently applicable to any Lie-Poisson model. In other words, they determine the most general form of the decomposition. 
We remark that, in the matched pair Hamiltonian dynamics, the constitutive systems must be Lie-Poisson by itself. In other words, the matched pair Hamiltonian dynamics turns out to be now a particular instance of extended  Hamiltonian dynamics in Subsection \ref{2-LP-Coec}. 

Extended Hamiltonian framework in Subsection \ref{2-LP-Coec} fits very well with the cuts of order $\geq 2$ kinetic moments of the density function of the Vlasov's plasma. So it provides an algebraic/geometric answer to the problem motivating the first goal. This decomposition is now more or less direct in the realm of the work \cite{OS20}. We shall, on the other hand, postpone it to a future work where we plan to discuss the geometry with further physical intuitions. 

As an illustration of 
the model in Subsection \ref{2-LP-Coec}, in the present work, we address another phenomenon in plasma dynamics, namely; all possible decompositions of (3 particles) BBGKY (Bogoliubov-Born-Green-Kirkwood-Yvon) hierarchy \cite{Ha}. Accordingly, in Section \ref{examples}, we shall examine this concrete model in full detail in order to make the novel Lie-Poisson structures introduced in this paper more clear. In \cite{marsden1984hamiltonian}, it was proved that BBGKY hierarchy can be recasted as a Lie-Poisson equation for $n>3$ particles. Focusing on the case $n=3$,  which is missing in \cite{marsden1984hamiltonian}, we shall present two decompositions of the BBGKY dynamics; the matched pair (double cross sum) decomposition, and the extended structure decomposition. This way, we shall also be able to compare these two approaches. These decompositions will precisely determine the relationship between the moments of the $3$-particle plasma density function. Without the extended Hamiltonian dynamics, on the other hand, the latter decomposition wouldn't be possible.  

\textbf{The second goal of the present work.} 
Being Lie algebras, 2-cocycle extensions of Lie algebras may form matched pairs along with the properly coupled 2-cocycle terms. We shall derive the conditions for the matched pair of 2-cocycle extensions to be a 2-cocycle in Section \ref{Sec-Cop-co} as the second goal of this note. We record the result in Proposition \ref{Prop-Gokhan}. In addition, Lie-Poisson dynamics on the coupled system will be presented. 
The geometric framework, and the dynamical equations, will be illustrated in Section \ref{Sec-3D} via two copies of the Heisenberg algebra. It is well-known that the symplectic two-form on the two-dimensional Euclidean space can be used to determine an extended Lie algebra structure on the three-dimensional Euclidean space. These results with the Heisenberg algebra in $3D$ \cite{Ha15}. Being a nilpotent Lie algebra of class two, the Heisenberg algebra may be matched by itself, \cite{Ma20}. This provides an application of the extended Hamiltonian theory as well as Proposition \ref{Prop-Gokhan} and results with the physically interesting equations on the dual spaces. The coadjoint orbits of the Heisenberg Lie algebras read the canonical Hamiltonian formalism. So that coupling of two Heisenberg algebras provides a mutual coupling of two canonical Hamiltonian dynamics under mutual interaction, which was also missing in the literature.
 
\textbf{Coupling of dissipative systems.}
If a dynamical system is in the Hamiltonian form, then as a result of the skew-symmetry of the Poisson bracket, the Hamiltonian function is a conserved quantity \cite{MaRa13}. This geometric fact corresponds to the conservation of energy when applied to some physical problems. The time-reversal character of the Hamiltonian dynamics depends basically on this observation. Those systems violating the time-reversal property, therefore, could not be put into the Hamiltonian formulation. Nevertheless, one may achieve to add a (Rayleigh type) dissipative term to the Lie-Poisson dynamics by means of a linear operator from the dual space to the Lie algebra \cite{bloch96}. This naive strategy works very well for many physical problems. More generally, and perhaps a more geometric approach is to add an additional feature to the manifold. There are methods, in the literature, to achieve this. 

\textbf{GENERIC - Metriplectic systems.} In the early '80s, some extensions of the Poisson geometry are introduced independently in order to add dissipative terms into Hamiltonian formulations (see Subsection \ref{Sec-Generic}). These geometries are known today under the name as metriplectic dynamics \cite{Ka84,Ka85,Mo84,Mo86} or as GENERIC (General Equation for Non-Equilibrium Reversible Irreversible Coupling) \cite{Gr84,GrOt97,Ot05}. In metriplectic systems, the geometry is determined by two compatible brackets, namely a Poisson bracket and a (possibly semi-Riemannian) symmetric bracket. In GENERIC, which is more general  \cite{Mielke2011}, a dissipation potential is employed in order to arrive at the irreversible part of the dynamics. Accordingly, the Legendre transformation of dissipation potential determines the time irreversible part of the dynamics  \cite{PaGr18,Gr2018}. If the potential is quadratic, then one arrives at a symmetry bracket as a particular case. 

One of the problems in this coupling is to determine a proper symmetric bracket, or a dissipation potential compatible with the Poisson geometry (see Subsection \ref{Sec-Sym-Bra}). In the present work, we shall refer to geometric ways to obtain symmetric brackets; such as the double bracket \cite{Br88,Br93,Br94}, Cartan-Killing bracket \cite{Mo09}, and Casimir dissipation bracket \cite{GB13}. We are interested in these geometries since, in the Lie-Poisson framework, they may be defined by the Lie algebra bracket directly in an algorithmic way. In the present work, we shall study extensions/couplings of the symmetric brackets as well as the dissipative systems. The latter will be the third goal here.

\textbf{The third goal of the present work.} 
In order to present a complete picture, we shall study the couplings of the dissipative terms which are added to the Lie-Poisson dynamics as the third goal of this present work in Section \ref{Sec-C-DS}. In other words, our third goal is to couple two metriplectic systems under mutual actions. First, we aim to provide a way to couple two mutually interacting systems involving Rayleigh type dissipative terms. Later, we present couplings of the Double brackets, the Cartan-Killing brackets, and the Casimir dissipation brackets. 

As for the coupling problem, our emphasis shall be on $3D$ systems. Accordingly, we present two illustrations. One, given in Section \ref{Diss-Gen-Exp}, will be on the rigid body dynamics. We shall provide couplings of both reversible and irreversible rigid body dynamics under mutual interactions. From the decomposition point of view, this corresponds to the Iwasava decomposition of $SL(2,\mathbb{C)}$. The other, in Section \ref{Sec-3D}, will be to continue with the Heisenberg algebra, endowing the geometry with dissipations.

\textbf{Contents.} In the following section, we shall be exhibiting a brief summary of the preliminaries for the sake of the completeness of the work. These are including Hamiltonian dynamics and metriplectic dynamics. Section \ref{Sec-MP} is reserved for
the presentation of the extended structure, and its particular instances such as a matched pair of Lie algebras and 2-cocycle extensions. In Section \ref{LP-Ext-Sec}, Lie-Poisson dynamics are studied in the dual space of the extended structure. In Section \ref{examples}, we shall decompose BBGKY dynamics as an illustration of the theoretical results obtained in the previous sections. In Section \ref{Sec-Cop-co}, 
we shall determine the conditions for a matched pair 2-cocycle to be a 2-cocycle of a matched pair. Couplings of the symmetric brackets are given in Section \ref{Sec-C-DS}. In Sections \ref{Sec-3D} and \ref{Diss-Gen-Exp}, $3D$ examples are provided. 

\section{Fundamentals: Lie-Poisson Dynamics, and Dissipations}

\subsection{Hamiltonian Dynamics}

Consider a Poisson manifold $\left( P,\{\bullet , \bullet \}\right)$ \cite{LaPi12,Va12}. On this geometry, Hamilton's equation generated by a Hamiltonian function(al) $\mathcal{H}$ is defined to be
\begin{equation}
\dot{\textbf{z}}=\{\textbf{z},\mathcal{H}\},
\end{equation}
where $\mathbf{z}$ is in $P$. Define Hamiltonian vector field $X_{\mathcal{H}}$ for a Hamiltonian function(al) $\mathcal{H}$, as follows 
\begin{equation}
X_{\mathcal{H}}(\mathcal{F})=\{\mathcal{F},\mathcal{H}\}. 
\end{equation}
A function(al) $\mathcal{C}$ is called a Casimir function(al) if it commutes with all other function(al)s that is, $\{\mathcal{F},\mathcal{C}\}=0$  for all $\mathcal{F}$. If there does not exist any non-constant Casimir function(al) for a Poisson bracket, then we say that the Poisson bracket is non-degenerate. It should be noted that the Hamiltonian vector field generated by a Casimir function(al) $\mathcal{C}$ is identically zero.  
The characteristic distribution, that is the image space of all Hamiltonian vector fields is integrable. This reads a foliation of $P$ into a collection of symplectic leaves \cite{wei83}. That is, on each leaf, the Poisson bracket turns out to be non-degenerate. If the bracket is already non-degenerate on $P$ then there exists only one leaf, and $P$ turns out to be a symplectic manifold.

Skew-symmetry of Poisson bracket verifies that Hamiltonian function(al) is preserved throughout the motion. Since Hamiltonian function(al) is taken as the total energy in classical systems, we may call this property as the conservation of energy. This manifests  the reversible character of Hamiltonian dynamics.

Referring to Poisson bracket, we define a bivector field $\Lambda$  as follows
\begin{equation} \label{bivec-PoissonBra}
\Lambda(d\mathcal{F},d\mathcal{H}):=\{\mathcal{F},\mathcal{H}\}
\end{equation}
for all $\mathcal{F}$ and $\mathcal{H}$, \cite{DuZu06}. Here, $d\mathcal{F}$ and $d\mathcal{H}$ denote the de-Rham exterior derivatives. So that, we may alternatively introduce a Poisson manifold by a tuple $(P,\Lambda)$ consisting of a manifold and a bivector field. There is Schouten-Nijenhuis algebra on bivector fields \cite{deAzPeBu96}. In this picture, the Jacobi identity turns out to be the commutation of $\Lambda$ with itself under the Schouten-Nijenhuis bracket that is,
\begin{equation} \label{Poisson-cond}
[\Lambda,\Lambda]=0.
\end{equation} 

\textbf{Lie-Poisson systems.} Consider a Lie algebra $\mathfrak{K}$ equipped with a Lie bracket $[\bullet ,\bullet ]$, \cite{Ja79}. Dual $\mathfrak{K}^{\ast }$ admits a Poisson bracket, called Lie-Poisson bracket \cite{Ho08,HoScSt09,demethods2011,LiMa12,MaRa13}. 
For two function(al)s $\mathcal{F}$ and $\mathcal{H}$, the (plus/minus) Lie-Poisson bracket is defined to be 
\begin{equation} \label{LP-Bra}
\{\mathcal{F},\mathcal{H}\} ( \textbf{z} ) = \pm \Big\langle \textbf{z} ,\left[ \frac{\delta \mathcal{F}}{
\delta \textbf{z}},\frac{\delta \mathcal{H}}{\delta \textbf{z} }\right]\Big\rangle
\end{equation} 
where $\delta \mathcal{F}/\delta \textbf{z} $ is the partial derivative (for infinite dimensional cases, the Fr\'{e}chet derivative) of the  function(al) $\mathcal{F}$. Here, the pairing on the right hand side  is the duality between $\mathfrak{K}^*$ and $\mathfrak{K}$ whereas the bracket is the Lie algebra bracket on $\mathfrak{K}$. Note that, we assume the reflexivity condition on $\mathfrak{K}$, that is the double dual $\mathfrak{K}^{**}=\mathfrak{K}$. 
The dynamics of an observable $\mathcal{F}$, governed by a Hamiltonian function(al) $\mathcal{H}$, is then computed to be
\begin{equation}\label{aaa}
\dot{\mathcal{F}}=\{\mathcal{F},\mathcal{H}\} ( \textbf{z}  )
= \pm \Big\langle \textbf{z} ,\left[ \frac{\delta \mathcal{F}}{\delta \textbf{z} },\frac{\delta \mathcal{H}}{
\delta \textbf{z}}\right] \Big\rangle =\pm \Big\langle \textbf{z} ,-\ad_{\delta \mathcal{H}/
\delta \textbf{z}}\frac{\delta \mathcal{F}}{\delta \textbf{z} }\Big\rangle = \pm \Big\langle
\ad_{\delta \mathcal{H}/
\delta \textbf{z}}^{\ast }\textbf{z} ,\frac{\delta \mathcal{F}}{\delta \textbf{z} }
\Big\rangle.
\end{equation}
Here, $\ad_\textbf{x} \textbf{x}':=[\textbf{x},\textbf{x}']$ for all $\textbf{x}$ and $\textbf{x}'$ in $\mathfrak{K}$ is the (left) adjoint action of the Lie algebra $\mathfrak{K}$ on itself whereas $\ad^*$ is the (left) coadjoint action of the Lie algebra $\mathfrak{K}$ on the dual space $\mathfrak{K}^*$.
 Notice that $\ad^*_\textbf{x}$ is defined to be minus of the linear algebraic dual of $\ad_\textbf{x}$. Then, 
 we obtain the equation of motion governed by a Hamiltonian function(al) $\mathcal{H}$ as
\begin{equation} \label{eqnofmotion}
\dot{\textbf{z}} \mp \ad_{{\delta \mathcal{H}}/{\delta \textbf{z} }}^{\ast }\textbf{z}=0.
\end{equation}

\begin{remark}
There are a plus/minus notations in \eqref{LP-Bra} and \eqref{aaa}. A plus sign appears if the reduction (Lie-Poisson reduction) is performed referring to a right symmetry whereas a minus sign appears if the reduction is performed referring to a left symmetry. 
For the plasma dynamics (see Section \ref{examples}), we shall refer plus Lie-Poisson bracket since Vlasov's plasma has a right (called relabelling) symmetry. For finite-dimensional rigid body motion (see Section  \ref{Sec-3D}), we 
shall employ minus Lie-Poisson bracket.
\end{remark}

\textbf{Coordinate realizations.} Assume a (local) coordinate chart $(z_i)$ (we prefer subscripts since we only focus on the dual spaces) around a point $\textbf{z}$ in $P$. Then the Poisson bivector can be represented by a set of coefficient functions $\Lambda_{ij}$ determining a Poisson bracket as
\begin{equation} \label{PB-local}
\{\mathcal{F},\mathcal{H}\}=\Lambda_{ij} \frac{\partial \mathcal{F}}{\partial z_i}\frac{\partial \mathcal{H}}{\partial z_i}. 
\end{equation}
Then the equation of motion generated by a Hamiltonian function $\mathcal{H}$ becomes
\begin{equation}
\dot{z}_i=\Lambda_{ij} \frac{\partial \mathcal{H}}{\partial z_j}.
\end{equation}
Let us examine the Lie-Poisson structure which is defined on the dual of a finite dimensional Lie algebra. For this, assume that an $K$ dimensional Lie algebra $\mathfrak{K}$ admitting a basis $\{\textbf{k}_{i}\}=\{\textbf{k}_{1},\dots, \textbf{k}_{K}\}$. The Lie algebra bracket on $\mathfrak{K}$ determine a set of scalars $C_{ij}^{l}$, called structure constants, satisfying
\begin{equation}\label{cerceve}
   [\textbf{k}_{i},\textbf{k}_{j}]=C_{ij}^{m}\textbf{k}_{m},
\end{equation}
where the summation convention is assumed over the repeated indices. Note that, after fixing a basis, the structure constants define a Lie bracket in a unique way. One has the dual basis $\{\textbf{k}^i\}=\{\textbf{k}^{1},\dots, \textbf{k}^{K}\}$ on the dual space $\mathfrak{K}^*$. We denote an element of $\mathfrak{K}^*$ by $\textbf{z}=z_i\textbf{k}^i$ where the coordinates $\{z_1, \dots,z_N\}$ are being real numbers. The (plus/minus) Lie-Poisson bracket \eqref{LP-Bra} can be computed in this picture as 
\begin{equation} \label{LP-loc}
\{\mathcal{F},\mathcal{H}\}= \pm C^{n}_{ij} z_{n}  \frac {\partial \mathcal{F}}{\delta z_{i}} \frac {\partial \mathcal{H}}{\delta z_{j}}.
\end{equation}
The calculation in (\ref{LP-loc}) reads that the coefficients $\Lambda_{ij}$  of the (plus/minus) Poisson bivector (\ref{PB-local}) are determined through the linear relations
\begin{equation} \label{SC-LP}
\Lambda_{ij}=\pm C_{ij}^mz_m.
 \end{equation}
In this case, Lie-Poisson equations \eqref{eqnofmotion} are computed to be 
\begin{equation} \label{loc-LP-Eqn}
\dot{z}_j\mp C^{n}_{ij} z_{n}\frac{\partial \mathcal{H}}{\delta z_i }=0.
\end{equation}

\textbf{Rayleigh dissipation.} Let us present a simply way to add dissipation to Lie-Poisson dynamics. Define a linear transformation $\Upsilon $ from $\mathfrak{K}^*$ to $\mathfrak{K}$. We equip a dissipative term to the right hand side of the Lie-Poisson system \eqref{eqnofmotion} by simply adding $\mp ad^*_{\Upsilon(\textbf{z})} \textbf{z}$ that is,
\begin{equation}\label{RD-Eqn}
\dot{z}\mp ad^*_{\delta \mathcal{H} / \delta \textbf{z}} \textbf{z} = \mp ad^*_{\Upsilon(\textbf{z})} \textbf{z}
\end{equation}
see \cite{bloch96}. We ask that $\Upsilon$ be a gradient relative to a
certain metric at least on adjoint orbits. 
In the upcoming subsection, we introduce a geometric framework for obtaining dissipation in the Lie-Poisson setting.

\subsection{GENERIC (Metriplectic) Systems}
\label{Sec-Generic}

In order to add dissipative terms to Hamiltonian dynamics, two geometric models are addressed in the literature, namely metriplectic systems \cite{Ka84,Ka85,Mo84,Mo86} and GENERIC (an acronym for General Equation for Non-Equilibrium Reversible-Irreversible Coupling), \cite{GrOt97,Gr2018}. 
In metriplectic systems, Poisson geometry is equipped with a proper symmetric bracket (possibly semi-Riemannian). In GENERIC, which is more general  \cite{Mielke2011}, a dissipation potential is employed. For convex potential functions, after the Legendre transformation, one arrives at the irreversible part of the dynamics  \cite{PaGr18,Gr2018}. If the potential is quadratic, then one arrives at a symmetric bracket as a particular case. In this work, we are interested in dissipative dynamics defined through symmetric brackets \cite{Otto2001}. Let us depict this geometry in detail. 

Consider a Poisson manifold $(P,\{\bullet,\bullet\})$ and assume, additionally, a symmetric bracket $(\bullet,\bullet)$ on the space of smooth functions on $P$. The metriplectic bracket $[\vert\bullet,\bullet\vert]$ on the manifold $P$ is defined by the addition of the Poisson bracket and the symmetric bracket that is, for two function(al)s $\mathcal{H}$ and $\mathcal{F}$,
\begin{equation}\label{MB}
\left[\vert \mathcal{H},\mathcal{F} \vert\right] = \{\mathcal{H},\mathcal{F}\} + a (\mathcal{H},\mathcal{F}),
\end{equation}
where $a$ being a scalar. Note that, a metriplectic bracket is an example of a Leibniz bracket
\cite{OrPl04}. There is no unique way to define a symmetric bracket. One way is to introduce a (maybe semi-)Riemannian metric $\mathcal{G}$ on $M$. After a bracket is established, the next task is to determine the generating function(al)s. In accordance with this, we determine two different kinds of the GENERIC (metriplectic) systems \cite{Gu07}. 

In the first kind of GENERIC (metriplectic) systems, one refers  a single function(al) $\mathcal{F}$ to generate the equations of motion. So that, the dynamics is given as
\begin{equation}\label{MD1}
\dot{\textbf{z}}=\left[\vert \textbf{z},\mathcal{F}\vert\right]=\{\textbf{z},\mathcal{F} \}+ a (\textbf{z},\mathcal{F}) 
\end{equation}
for $\textbf{z}\in P$.
By, particularly, choosing the metric $\mathcal{G}$ positive definite, and by letting $a$ be equal to $-1$, one arrives at the dissipation of the generating function $\mathcal{F}$ in time. In the second kind of GENERIC (metriplectic) systems, there exist two function(al)s, namely a Hamiltonian function(al) $\mathcal{H}$ and an entropy-type function(al) $\mathcal{S}$. In this time, the dynamics is written as
\begin{equation}\label{MD2}
\dot{\textbf{z}}=\{\textbf{z},\mathcal{H} \} + a (\textbf{z},\mathcal{S}).
\end{equation}
If the following identities
\begin{equation}\label{C1}
\{\mathcal{S},\mathcal{H}\}= 0, \qquad (\mathcal{H},\mathcal{S})= 0
\end{equation}
hold, the  metriplectic dynamics (\ref{MD2}) can be generated by a single function (free energy function) $\mathcal{F}=\mathcal{H}-\mathcal{S}$ defined as the difference of the Hamiltonian and entropy-type function(al)s. For such kind of systems, Hamiltonian function(al) $\mathcal{H}$ is a conserved quantity whereas the dissipative behavior
of the system is interpreted as the increase
of entropy along trajectories assuming that $a$ is positive. This case is possible if the Poisson structure is degenerate and the symmetric tensor $G$ is at most semi-definite. In the following subsection, we introduce some examples of symmetric brackets that can be attached to the Lie-Poisson bracket.

\subsection{Some Symmetric Brackets on the Duals of Lie algebras}
\label{Sec-Sym-Bra}
In this subsection, we list some symmetric brackets available on the dual $\mathfrak{K}^*$ of a Lie algebra $\mathfrak{K}$. After a symmetric bracket is determined, say $(\bullet,\bullet)$ the irreversible dynamics governed by a generating function, say $\mathcal{S}$, is computed to be 
\begin{equation} \label{diss-dyn-eq}
\dot{\textbf{z}}=a(\textbf{z},\mathcal{S}),
\end{equation}
where $a$ is a real number. 
Assuming a basis $\{\mathbf{k}_i\}$, and the according local coordinates $\{z^i\}$ on $\mathfrak{K}$, the primary goal in this subsection is to define a symmetric tensor field 
\begin{equation}
\mathcal{G}=\mathcal{G}_{ij} dz^i\otimes dz^j
\end{equation}
on $\mathfrak{K}$. In the dual space $\mathfrak{K}^*$, we employ the dual basis $\{\textbf{k}^{i}\}$ and the coordinates $\{z_i\}$. There are two distinguished functions on the Lie-Poisson picture. Here  is a list: 

\textbf{Double bracket.} Recall the structure of a Lie algebra $\mathfrak{K}$ given in \eqref{cerceve}. In the Lie-Poisson setting, the coefficients of the Poisson bivector are determined by the structure constants of the Lie algebra as \eqref{SC-LP} that is, $\Lambda_{ij}=C_{ij}^lz_l$ \cite{Mo09}. For two functions $\mathcal{F}$ and $\mathcal{H}$, we define a symmetric bracket, literarily called double bracket, 
\begin{equation}\label{doubledissi}
	(\mathcal{F},\mathcal{H})^{(D)}=
\sum_{j}	\Lambda_{ij} \Lambda_{lj}\frac{\partial \mathcal{F}}{\partial z_i} \frac{\partial \mathcal{H}}{\partial z_l} =
	\sum_{j}C_{ij}^rC_{lj}^s z_r  z_s \frac{\partial \mathcal{F}}{\partial z_i} \frac{\partial \mathcal{H}}{\partial z_l},
\end{equation} 
see \cite{Br93}. So that we can write the coefficients of the symmetric bracket as in terms of the structure constants of the Lie algebra as follows
\begin{equation}
\mathcal{G}_{ij}=\sum_{l}C_{il}^rC_{jl}^s z_r  z_s.
\end{equation}
Now we define a metriplectic bracket on $\mathfrak{K}^*$ by adding Lie-Poisson bracket \eqref{LP-loc} and double bracket \eqref{doubledissi} that is
\begin{equation}
\dot{\mathcal{F}}=[\vert\mathcal{F},\mathcal{H}\vert]^{(D)}=\{\mathcal{F},\mathcal{H}\}+a
(\mathcal{F},\mathcal{H})^{(D)}.
\end{equation}
So that according to the definition \eqref{MD1} we compute the equation of motion as
\begin{equation} \label{aaaa}
\dot{z}_j \mp C^{m}_{ij} z_{m}\frac{\partial \mathcal{H}}{\partial z_i }= a\sum_{i} C_{ji}^r C_{mi}^n z_r  z_n \frac{\partial \mathcal{H}}{\partial z_m} 
\end{equation}
where on the left hand side we have the reversible Hamiltonian dynamics whereas the dissipative term is located at the right hand side. 

\textbf{Cartan-Killing bracket.}
Consider a Lie algebra given in coordinates as \eqref{cerceve}. Referring to skew-symmetry of the structure constants the Cartan-Killing metric is defined as 
\begin{equation} \label{CK}
	\mathcal{G}_{i j}=C_{i m}^{n}C_{j n}^{m}.
\end{equation} 
It can be easily shown that the scalars  $\mathcal{G}_{i j}$ define a symmetric and bilinear covariant tensor \cite{Mo09}. We define a symmetric bracket for functions $\mathcal{F}$ and $\mathcal{H}$ in terms of the metric as follows
 \begin{equation} \label{loc-car}
(\mathcal{F},\mathcal{H})^{(CK)}= \frac{\partial \mathcal{F}}{\partial z_{i}}\mathcal{G}_{ij} \frac{\partial \mathcal{H}}{\partial z_{j}}=C_{i j}^{n}C_{m n}^{j}\frac{\partial \mathcal{F}}{\partial z_{i}}\frac{\partial \mathcal{H}}{\partial z_{m}}.
 \end{equation}
To arrive at a metriplectic bracket on $\mathfrak{K}^*$, we add the Lie-Poisson bracket \eqref{LP-loc} and the symmetric bracket (\ref{loc-car}) that is,
\begin{equation}
\dot{\mathcal{F}}=[\vert\mathcal{F},\mathcal{H}\vert]^{(CK)}=\{\mathcal{F},\mathcal{H}\}+a
(\mathcal{F},\mathcal{H})^{(CK)}.
\end{equation}
Accordingly, the metriplectic dynamics is computed to be
\begin{equation}
\begin{split}
\dot{z}_j \mp  C^{m}_{ij} z_{m}\frac{\partial \mathcal{H}}{\partial z_i }=aC_{j i}^{n}C_{m n}^{i}  \frac{\partial \mathcal{H}}{\partial z_{m}}
\end{split}
\end{equation}
where, on the left hand side, we have the reversible Hamiltonian dynamics whereas, on the right hand side, the dissipative term is exhibited. 

\textbf{Casimir dissipation bracket.} Define a symmetric bilinear operator $\psi $ on Lie algebra $\mathfrak{K}$. Referring to any Casimir function $\mathcal{C}$ of the Lie-Poisson bracket, define a symmetric bracket \cite{GB13} of  two functions $\mathcal{F}$ and $\mathcal{H}$ as
\begin{equation}\label{ccs}
(\mathcal{F},\mathcal{H})^{(CD)}=-\psi \left( \left[ \frac{\delta \mathcal{F}}{\delta \textbf{z} },\frac{\delta \mathcal{H}}{\delta \textbf{z} }\right] ,\left[ \frac{\delta C}{\delta \textbf{z}},\frac{%
\delta \mathcal{H}}{\delta \textbf{z}}\right] \right).  
\end{equation}
This bracket is suitable for the second type of metriplectic bracket.
Note that the change of Hamiltonian function over time is constant and the change of  Casimir function can be given as
\begin{equation*}
\dot{\mathcal{C}} =-\psi \left( \left[ \frac{
\delta \mathcal{C}}{\delta \textbf{z} },\frac{\delta \mathcal{H}}{\delta \textbf{z} }\right] ,\left[ \frac{
\delta \mathcal{C}}{\delta \textbf{z} },\frac{\delta \mathcal{H}}{\delta \textbf{z} }\right] \right)=-\left\Vert \left[ \frac{\delta \mathcal{C}}{\delta \textbf{z} },\frac{\delta \mathcal{H}}{\delta
\textbf{z} }\right] \right\Vert ^{2}<0.
\end{equation*} 
 The dynamics of an arbitrary observable $\mathcal{F}$ governed by a Hamiltonian function $\mathcal{H}$ is deduced by this bracket as
\begin{equation}
\begin{split}
\dot{\mathcal{F}}=-\psi \left( \left[ \frac{\delta \mathcal{F}}{\delta \textbf{z} },\frac{
	\delta \mathcal{H}}{\delta \textbf{z} }\right] ,\left[ \frac{\delta \mathcal{C}}{\delta \textbf{z} },\frac{
	\delta \mathcal{H}}{\delta \textbf{z} }\right] \right) = \left\langle    \left[ \frac{\delta \mathcal{C}}{\delta \textbf{z} },\frac{
	\delta \mathcal{H}}{\delta \textbf{z} }\right]^{\flat} ,ad_{{\delta H}/{\delta \textbf{z} }}\frac{
	\delta \mathcal{F}}{\delta \textbf{z} } \right\rangle  
&=\left\langle  -ad^*_{{\delta \mathcal{H}}/{\delta \textbf{z} }}\left[ \frac{\delta C}{\delta \textbf{z} },\frac{
	\delta \mathcal{H}}{\delta \textbf{z} }\right]^\flat , \frac{
	\delta \mathcal{F}}{\delta \textbf{z} } \right\rangle.
\end{split}
\end{equation}
Here, the musical mapping $\flat$, from $\mathfrak{K}$ to 
$\mathfrak{K}^{\ast }$, is defined through the symmetric operator $\psi $, satisfying the identity  $\left\langle \textbf{x}^{\flat},\textbf{x}' \right\rangle =\psi \left( \textbf{x},\textbf{x}' \right) $ for two elements $\textbf{x}$ and $\textbf{x}'$ in $\mathfrak{K}$. In this case, the dissipative equation of motion can be written as follows
\begin{equation} \label{DB-Dyn} 
\dot{\textbf{z}}=-ad_{\frac{\delta  \mathcal{H}}{\delta \textbf{z}}}^{\ast }\left[ \frac{\delta C}{%
\delta \textbf{z}},\frac{\delta  \mathcal{H}}{\delta \textbf{z}}\right] ^{\flat}.
\end{equation} 

We collect the Lie-Poisson bracket \eqref{LP-loc} and the Casimir Dissipation bracket \eqref{ccs} together to arrive at the following metriplectic bracket 
\begin{equation}
\dot{\mathcal{F}}=[\vert\mathcal{F},\mathcal{H}\vert]^{(CD)}=\{\mathcal{F},\mathcal{H}\}+a
(\mathcal{F},\mathcal{H})^{(CD)}.
\end{equation}
Then we compute the equation of motion as 
\begin{equation} \label{MB-DBD-Dyn}
\dot{\textbf{z}} \mp ad_{\frac{\delta \mathcal{H}}{\delta \textbf{z}}}^{\ast }z =(-a)ad_{\frac{\delta
H}{\delta \textbf{z} }}^{\ast }\left[ \frac{\delta \mathcal{C}}{\delta \textbf{z} },\frac{\delta \mathcal{H}}{
\delta \textbf{z} }\right] ^{\flat}.
\end{equation}

 \textbf{Hamilton dissipation  bracket.} We start with assuming a symmetric (semi-positive definite) bilinear operator $\psi $ defined on a Lie algebra $\mathfrak{K}$. We fix a Casimir function $\mathcal{C}$ of the Lie-Poisson bracket \eqref{LP-Bra}, and then introduce the following symmetric bracket on the dual space $\mathfrak{K}^*$, for two function(al)s $\mathcal{F}$ and $\mathcal{H}$, given by
\begin{equation} \label{HD-sym}
(\mathcal{F},\mathcal{H})^{(HD)}=-\psi\Big(\Big[\frac{\delta \mathcal{F}}{\delta \textbf{z}},\frac{\delta \mathcal{C}}{\delta \textbf{z}}\Big],[\frac{\delta \mathcal{H}}{\delta \textbf{z}} ,\frac{\delta \mathcal{C}}{\delta \textbf{z}}\Big]\Big)
\end{equation}
where the brackets on the right sides are Lie algebra brackets on $\mathfrak{K}$. An interesting feature of this symmetric bracket is to see that the Casimir function(al) $\mathcal{C}$ is a conserved quantity for the dynamics determined by the bracket \eqref{HD-sym} since $\dot{\mathcal{C}}=0$ due to the skew-symmetry of the Lie-bracket. On the other hand the generating function $\mathcal{H}$ dissipates that is
\begin{equation*}
\dot{\mathcal{H}} =(\mathcal{H},\mathcal{H})^{(HD)}=-\psi \Big( \Big[ \frac{
\delta \mathcal{H}}{\delta \textbf{z} },\frac{\delta \mathcal{C}}{\delta \textbf{z} }\Big] ,\Big[ \frac{
\delta \mathcal{C}}{\delta \textbf{z} },\frac{\delta \mathcal{H}}{\delta \textbf{z} }\Big] \Big)=-\left\Vert \Big[ \frac{\delta \mathcal{H}}{\delta \textbf{z} },\frac{\delta \mathcal{C}}{\delta
\textbf{z} }\Big] \right\Vert ^{2} \leq 0.
\end{equation*}
More general, the gradient flow of an observable $\mathcal{F}$ generated by a function(al) $\mathcal{H}$ is computed to be 
\begin{equation}
\dot{\mathcal{F}}=-\psi \left( \left[ \frac{\delta \mathcal{F}}{\delta \textbf{z} },\frac{
	\delta \mathcal{H}}{\delta \textbf{z} }\right] ,\left[ \frac{\delta \mathcal{C}}{\delta \textbf{z} },\frac{
	\delta \mathcal{H}}{\delta \textbf{z} }\right] \right) = \left\langle    \left[ \frac{\delta \mathcal{C}}{\delta \textbf{z} },\frac{
	\delta \mathcal{H}}{\delta \textbf{z} }\right]^{\flat} ,ad_{{\delta H}/{\delta \textbf{z} }}\frac{
	\delta \mathcal{F}}{\delta \textbf{z} } \right\rangle  
=\left\langle  -ad^*_{{\delta \mathcal{H}}/{\delta \textbf{z} }}\left[ \frac{\delta C}{\delta \textbf{z} },\frac{
	\delta \mathcal{H}}{\delta \textbf{z} }\right]^\flat , \frac{
	\delta \mathcal{F}}{\delta \textbf{z} } \right\rangle.
\end{equation}
Accordingly, we can write the equation of motion generated by $\mathcal{S}$ is  
\begin{equation} \label{HDB-Dyn}
\dot{\textbf{z}}=-ad_{{\delta C}/{%
\delta \textbf{z}}}^{\ast }\Big[\frac{\delta  \mathcal{C}}{\delta \textbf{z}} ,\frac{\delta  \mathcal{S}}{\delta \textbf{z}}\Big] ^{\flat}.
\end{equation} 
Now we are ready to add the Lie-Poisson bracket \eqref{LP-loc} and the symmetric bracket exhibited in \eqref{HD-sym} in order to define a metriplectic (Leibniz) bracket on $\mathfrak{K}^*$. By assuming that both the Lie-Poisson and the gradient dynamics generated by a single function $\mathcal{H}$ we have that
\begin{equation}
\dot{\mathcal{F}}=[\vert\mathcal{F},\mathcal{H}\vert]^{(HD)}=\{\mathcal{F},\mathcal{H}\}+a
(\mathcal{F},\mathcal{H})^{(HD)}.
\end{equation}
Then we compute the equation of motion as
\begin{equation} \label{MB-HD-Dyn}
\dot{\textbf{z}} \mp ad_{{\delta \mathcal{H}}/{\delta \textbf{z}}}^{\ast }z =(-a)ad_{{\delta C}/{
\delta \textbf{z}}}^{\ast }\left[ \frac{\delta \mathcal{C}}{\delta \textbf{z} },\frac{\delta \mathcal{H}}{
\delta \textbf{z} }\right] ^{\flat}.
\end{equation}

\section{Extensions of Lie Algebras}\label{Sec-MP}

In this section, we introduce a unifying construction, called \textit{extended structure}, for factorization of Lie algebras. Then, we examine its particular instances such as matched pair Lie algebra and 2-cocycle extension.

\subsection{Extended Structures} \label{brzezinski}

Let $(\G{g},[\bullet,\bullet])$ be a Lie algebra and, assume that, it acts on a vector space $\G{h}$ from the right that is 
\begin{equation}\label{Lieact1}
 \vartriangleleft :\mathfrak{h}\otimes \mathfrak{g}\rightarrow \mathfrak{h}
 ,\qquad \eta \otimes \xi \mapsto \eta \vartriangleleft \xi.
\end{equation}
 Our goal in this subsection is to construct the most general extension of  $\G{g}$ by $\G{h}$. To have this, we permit existence of the following maps
\begin{equation} \label{thm-m-h-const-maps}
\begin{split}
\Phi:\G{h}\ot\G{h}\longrightarrow \G{g},\qquad (\eta,\eta')\mapsto\Phi(\eta,\eta') \\
\kappa:\G{h}\ot\G{h}\longrightarrow \G{h} ,\qquad (\eta,\eta')\mapsto\kappa(\eta,\eta') 
\end{split}
\end{equation}
along with a linear map 
\begin{equation}\label{Lieact}
\vartriangleright :\mathfrak{h}\otimes \mathfrak{g}\rightarrow \mathfrak{g 
},\qquad \eta \otimes \xi \mapsto \eta \vartriangleright \xi. 
\end{equation}
Note here that \eqref{Lieact} is not an action since $\mathfrak{h}$ is only a vector space. The need of the operations \eqref{thm-m-h-const-maps} and \eqref{Lieact} will be more evident in the sequel where we examine this in the point of view of decomposition. 
The following theorem determines the conditions to define a Lie algebra structure on the direct sum $\G{K}=\G{g} \oplus \G{h}$, see also \cite{Agormili14, AgoreMili-book}.

\begin{theorem}\label{thm-m-h-const} The direct sum $\G{K}=\G{g} \oplus \G{h}$ is a Lie algebra via 
	\begin{equation} \label{mpla-2-cocyc-1}
\lbrack (\xi \oplus\eta ),\,(\xi'\oplus\eta')]_{_\Phi\bowtie}=\big( [\xi,\xi'
]+\eta \vartriangleright \xi'-\eta'\vartriangleright \xi + \Phi(\eta,\eta')
\big)\oplus\big(\kappa(\eta,\eta')+\eta \vartriangleleft \xi'-\eta' \vartriangleleft \xi \big).
\end{equation}
	where the mappings are the ones in \eqref{Lieact1}, \eqref{thm-m-h-const-maps} and \eqref{Lieact}, 
	if and only if, for any $\eta,\eta',\eta'' \in \G{h}$, and any $\xi,\xi' \in \G{g}$,
\begin{equation} \label{cocycle-compatibility}
\begin{split}
& \kappa(\eta,\eta) =0, \qquad \Phi(\eta,\eta) = 0, \\
& \kappa(\eta,\eta') \vartriangleleft \xi = \kappa(\eta,\eta' \vartriangleleft \xi) - \kappa(\eta',\eta \vartriangleleft \xi) + \eta \vartriangleleft (\eta' \vartriangleright \xi) - \eta' \vartriangleleft (\eta \vartriangleright \xi), \\
& \kappa(\eta,\eta')\vartriangleright \xi = [\xi,\Phi(\eta,\eta')] + \Phi(\eta,\eta'\vartriangleleft\xi) + \Phi(\eta\vartriangleleft\xi,\eta') + \eta \vartriangleright (\eta' \vartriangleright \xi) - \eta' \vartriangleright (\eta \vartriangleright \xi), \\
& \eta \vartriangleright [\xi,\xi'] = [\xi,\eta \vartriangleright \xi'] - [\xi',\eta \vartriangleright \xi] + (\eta \vartriangleleft \xi) \vartriangleright \xi' - (\eta \vartriangleleft \xi' )\vartriangleright \xi, \\
& \eta \vartriangleleft [\xi,\xi'] = (\eta \vartriangleleft \xi) \vartriangleleft \xi' - (\eta \vartriangleleft \xi' )\vartriangleleft \xi, \\
& \circlearrowright \Phi(\eta,\kappa(\eta',\eta'')) + \circlearrowright \eta \vartriangleright \Phi(\eta',\eta'') = 0, \\
& \circlearrowright \kappa(\eta,\kappa(\eta',\eta'')) + \circlearrowright \eta \vartriangleleft \Phi(\eta',\eta'') = 0,
\end{split}
\end{equation}
	where $\circlearrowright$ refers to the cyclic sum over the indicated elements.
\end{theorem}
\begin{proof}
	We first observe that
	\begin{equation}
	[\eta,\eta] = \big(\Phi(\eta,\eta), \kappa(\eta,\eta) \big) = 0
	\end{equation}
	if and only if
	\begin{equation}
	\Phi(\eta,\eta) = 0, \\ \qquad
	\kappa(\eta,\eta) = 0
	\end{equation}
	for any $\eta \in \G{h}$. Next, we shall consider the mixed Jacobi identities. Let us begin with 
	\begin{equation}\label{J-I}
	[ \xi,[\eta, \eta']] + [\eta,[\eta', \xi]] + [\eta',[\xi, \eta]] = 0,
	\end{equation}
	where
\begin{equation}
\begin{split}
	[ \xi,[\eta, \eta']] &= [\xi,(\Phi(\eta,\eta'),\, \kappa(\eta,\eta'))] = \big( - \kappa(\eta,\eta') \vartriangleright \xi  + [\xi, \Phi(\eta,\eta')]_\G{g},- \kappa(\eta,\eta') \vartriangleleft \xi  \big),
\\
	[\eta,[\eta',\xi] ] & = [\eta,(- \eta' \vartriangleright \xi,-\eta' \vartriangleleft \xi  ) ] =\big(- \eta \vartriangleright(\eta' \vartriangleright \xi) -\Phi(\eta' \vartriangleleft \xi, \eta), -\kappa(\eta' \vartriangleleft \xi, \eta) - \eta \vartriangleleft(\eta' \vartriangleright \xi) \big),
\\
	[\eta',[\xi, \eta] ] & =  [\eta',(\eta \vartriangleright \xi,\eta \vartriangleleft \xi  ) ]  =\big(\eta' \vartriangleright(\eta \vartriangleright \xi) + \Phi(\eta \vartriangleleft \xi, \eta'),\kappa(\eta \vartriangleleft \xi, \eta') + \eta' \vartriangleleft(\eta \vartriangleright \xi ) )\big).
	\end{split}
\end{equation}
	Hence, \eqref{J-I} is satisfied if and only if
\begin{equation}
\begin{split}
	\kappa(\eta,\eta')\vartriangleleft \xi &= \kappa(\eta \vartriangleleft \xi, \eta') - \kappa(\eta' \vartriangleleft \xi, \eta) + \eta' \vartriangleleft(\eta \vartriangleright \xi) - \eta \vartriangleleft(\eta' \vartriangleright \xi) 
\\
	[\xi, \Phi(\eta,\eta') ]_\G{g}  &= \kappa(\eta,\eta') \vartriangleright \xi + \eta \vartriangleright(\eta' \vartriangleright \xi) - \eta' \vartriangleright(\eta \vartriangleright \xi ) + \Phi(\eta' \vartriangleleft \xi, \eta) - \Phi(\eta \vartriangleleft \xi, \eta')
	\end{split}
\end{equation}
	for any $\eta,\eta' \in \G{h}$, and any $\xi \in \G{g}$.

	Next, we consider the Jacobi identity of an arbitrary $\eta \in \G{h}$, and any $\xi,\xi' \in \G{g}$, namely; 
	\begin{equation}\label{J-II}
	[[\xi, \xi'], \eta] + [[\xi', \eta],\xi] + [[\eta, \xi], \xi'] = 0.  
	\end{equation}
	In this case,
	\[
	[[\xi, \xi'], \eta] = \big(\eta \vartriangleleft [\xi, \xi']_{\G{g}} ,\, \eta \vartriangleright [\xi, \xi']_{\G{g}})\big),
	\]
	together with
\begin{equation}
\begin{split}
	[[\xi', \eta],\xi] &
	  =\big(( \eta \vartriangleleft \xi')\vartriangleleft \xi ,\, -(\eta \vartriangleleft \xi) \vartriangleright \xi' + [\xi,\eta \vartriangleright \xi']_\G{g}\big),
\\
	[[\eta, \xi], \xi']&
	= \big(-(\eta \vartriangleleft \xi) \vartriangleleft \xi',\,  (\eta \vartriangleleft \xi') \vartriangleright \xi - [\eta \vartriangleright \xi,\xi']_\G{g}\big).
	\end{split}
\end{equation}
		As such, \eqref{J-II} is satisfied if and only if 
\begin{equation}
\begin{split}
	\eta \vartriangleleft [\xi, \xi']_{\G{g}} & =   -( \eta \vartriangleleft \xi)\vartriangleleft \xi' + (\eta \vartriangleleft \xi) \vartriangleleft \xi' ,
\\
	\eta \vartriangleright [\xi, \xi']_{\G{g}} & = [\eta \vartriangleright \xi,\xi']_\G{g} + [\xi,\eta \vartriangleright \xi']_\G{g} +  (\eta \vartriangleleft \xi) \vartriangleright \xi' -  (\eta \vartriangleleft \xi') \vartriangleright \xi
	\end{split}
\end{equation}
	for any $\eta \in \G{h}$, and any $\xi,\xi' \in \G{g}$.
	Finally we consider the Jacobi identity 
	\begin{equation}\label{J-III}
	[[\eta, \eta'], \eta''] + [[\eta', \eta''], \eta] + [[\eta'', \eta], \eta'] = 0
	\end{equation}
	for any $\eta,\eta',\eta'' \in \G{h}$. We have,
	\begin{align*}
	[[\eta, \eta'], \eta''] &= \big((\Phi(\eta,\eta'),\kappa(\eta,\eta') ),\, \eta''\big)  \\
	& =\big(\eta'',\kappa(\kappa(\eta,\eta')) + \eta'' \vartriangleleft \Phi(\eta,\eta'),\, \eta'' \vartriangleright(\Phi(\eta,\eta') + \Phi(\eta'',\kappa(\eta,\eta'))\big),
	\end{align*}
	as well as
\begin{equation}
\begin{split}
	[[\eta', \eta''], \eta] &= [(\Phi(\eta',\eta''),\kappa(\eta',\eta'')) ,\, \eta]  \\
	&  = \big(\eta,\kappa(\kappa(\eta',\eta'')) + \eta \vartriangleleft \Phi(\eta',\eta'') ,\, \eta'' \vartriangleright \Phi(\eta',\eta) + \Phi(\eta,\kappa(\eta',\eta''))\big)
\\
	[[\eta'', \eta], \eta'] &= [(\Phi(\eta,\eta''),\kappa(\eta,\eta'')) ,\, \eta'] \\
	& = \big(\eta',\kappa(\kappa(\eta,\eta'')) + \eta' \vartriangleleft\Phi(\eta,\eta'') ,\, \eta'\vartriangleright\Phi(\eta,\eta'') + \Phi(\eta',\kappa(\eta,\eta''))\big).
	\end{split}
\end{equation}
	Accordingly, \eqref{J-III} is satisfied if and only if
\begin{equation}
\begin{split}
	\sum_{(\eta,\eta',\eta'')}\, \kappa(\eta'',\kappa(\eta,\eta')) + \sum_{(\eta,\eta',\eta'')}\, \eta''\vartriangleleft\Phi(\eta,\eta') = 0,
\\
	\sum_{(\eta,\eta',\eta'')}\, \eta''\vartriangleright\Phi(\eta,\eta') + \sum_{(\eta,\eta',\eta'')}\, \Phi(\eta'',\kappa(\eta,\eta')) = 0
	\end{split}
\end{equation}
	for any $\eta,\eta',\eta'' \in \G{h}$.
	\end{proof}

We denote the direct product space $\mathfrak{K} = \mathfrak{g} \oplus \mathfrak{h}$ equipped with the Lie algebra bracket \eqref{mpla-2-cocyc-1} by
$\mathfrak{K} = \mathfrak{g} _{_\Phi\bowtie} \mathfrak{h}$. In \cite{Agormili14, AgoreMili-book}, the realization presented in Theorem \ref{thm-m-h-const} has already been introduced under the name of \textit{extended structure}. We shall follow this terminology as well. We remark that, the last two identities in \eqref{cocycle-compatibility} are called twisted cocycle identity for $\Phi$ and twisted Jacobi identity for $\kappa$, respectively. 
In  the following subsection, we shall exploit that extended structure realizes both matched (double cross sum) Lie algebra and 2-cocycle extension Lie algebra as particular instances. 

Let us cite here some related studies presented in this section. See
\cite{BrHa99}  for a coalgebra discussion related with the extension presented here. We refer  \cite{JanViz16} for the extensions of Hamiltonian vector fields. See \cite{Kubo} for extensions of Poisson algebras.

\textbf{Decomposing a Lie algebra.} 
Instead of extending a Lie algebra with its representation space, one 
can decompose a Lie algebra into the internal direct sum of one of its Lie subalgebra and its complement. The latter is manifesting the inverse of the statement in Theorem \ref{thm-m-h-const}. So that, the (de)composition exhibited here is universal. Let us depict this argument. 

We start with a Lie algebra $\mathfrak{K}$, and assume a subalgebra, say $\mathfrak{g}$, of it. It is always possible to define  a complementary subspace $\mathfrak{h}\subset \mathfrak{K}$ so that $\mathfrak{K} = \mathfrak{g} \oplus \mathfrak{h}$. In most general case, for $\eta,\eta'$ in $\mathfrak{h}$, under the Lie algebra bracket of $\mathfrak{K}$, we have that
\begin{equation}\label{Phi-kappa}
[\eta,\eta'] =\Phi(\eta,\eta') \oplus \kappa(\eta,\eta') \in \mathfrak{g} \oplus \mathfrak{h} .
\end{equation}
Here, we define the mappings
\begin{eqnarray}
\Phi:\mathfrak{h} \times \mathfrak{h} \rightarrow \mathfrak{g}, \qquad 
\Phi(\eta,\eta') := proj_\mathfrak{g}[\eta,\eta'] \label{phi}
\\
\kappa:\mathfrak{h} \times \mathfrak{h} \rightarrow \mathfrak{h}, \qquad 
\kappa(\eta,\eta') := proj_\mathfrak{h}[\eta,\eta'] \label{kappa}
\end{eqnarray}
where $proj$ denotes the projection operator. Notice that, if $\Phi$ is identically zero then $\mathfrak{h}$ becomes a Lie subalgebra of $\mathfrak{K}$. In this case, $\kappa$ becomes the Lie algebra bracket on $\mathfrak{h}$. 

\begin{lemma}[Universal Lemma] \label{uni-prop-bre}
	Given a decomposition of Lie algebra $\G{K} = \G{g} \oplus \G{h}$, where $\G{g}$ is being a Lie subalgebra. The mappings $\Phi$ and $\kappa$ in \eqref{thm-m-h-const-maps} recovered from \eqref{Phi-kappa} 
	and the mutual \textit{actions} computed through
\begin{equation}\label{bracketactions}
	[0\oplus\eta,\xi\oplus 0] = \eta\vartriangleright \xi \oplus  \eta\vartriangleleft \xi 
\end{equation}
satisfies the conditions in \eqref{cocycle-compatibility}. That is, $\G{K}$ can be identified to the extended structure $\mathfrak{g} _{_\Phi\bowtie} \mathfrak{h}$ hence, \eqref{mpla-2-cocyc-1} reads a decomposition of the Lie bracket on $\G{K}$. 
 \end{lemma}

\textbf{Coordinate realizations.}  
Choose  a basis $\{\mathbf{e}_\alpha\}$ on an $N$-dimensional Lie algebra $\mathfrak{g}$ and a basis  $\{\mathbf{f}_a\}$ on $M$-dimensional vector space $\mathfrak{h}$. Notice that we reserve the Greek scripts while denoting the basis of the Lie algebra $\mathfrak{g}$ whereas we use the Latin scripts for the basis for the vector space $\mathfrak{h}$. We denote
\begin{equation}\label{nonbracket2}
[\mathbf{e}_\alpha,\mathbf{e}_\beta]=C^\theta_{\alpha \beta}\mathbf{e}_\theta, \qquad [\mathbf{f}_a,\mathbf{f}_b]=\Phi^\alpha_{ab}\mathbf{e}_\alpha+\kappa^d_{ab}\mathbf{f}_d,
\end{equation}
where the set $C^\theta_{\alpha \beta}$ determines the structure constants of the Lie subalgebra $\mathfrak{g}$ whereas the sets of constants $\Phi^\alpha_{ab}$ and $\kappa^d_{ab}$ are coordinate realizations of the mappings $\Phi$ and  $\kappa$ determined in \eqref{Phi-kappa}. We identify the mappings \eqref{Lieact1} and \eqref{Lieact} in terms of the basis $\mathbf{e}_\alpha$, $\mathbf{f}_a$ as 
\begin{equation} \label{local-act}
\mathbf{f}_{a}	\vartriangleleft \mathbf{e}_{\alpha}=R_{a \alpha}^{b}\mathbf{f}_{b}, \qquad 
 \mathbf{f}_{a}	\vartriangleright \mathbf{e}_{\alpha}=L_{a \alpha}^{\beta}\mathbf{e}_{\beta}.
\end{equation}
It is needless to say that, the scalars $L_{a \alpha}^{\beta}$ and  $R_{a \alpha}^{b}$ determine the mappings in a unique way. 

In the present finite case, the extended structure $\mathfrak{g} _{_\Phi\bowtie} \mathfrak{h}$ is an $N+M$ dimensional vector space. Referring to the basis of the constitutive subspaces, one can define a basis $\{\bar{\mathbf{e}}_{1}, \dots ,\bar{\mathbf{e}}_{N+M}\}$ on  $\mathfrak{g} _{_\Phi\bowtie} \mathfrak{h}$ as
\begin{equation} \label{coordmpL}
\{\bar{\mathbf{e}}_{\alpha},\bar{\mathbf{e}}_a\}\subset \mathfrak{g} _{_\Phi\bowtie} \mathfrak{h}, \qquad \bar{\mathbf{e}}_{\alpha}=\mathbf{e}_\alpha\oplus 0,\qquad \bar{\mathbf{e}}_a=0\oplus \mathbf{f}_a.
\end{equation}
In the light of these local choices \eqref{nonbracket2} and \eqref{local-act}, one can calculate the structure constants of the extended  Lie bracket \eqref{mpla-2-cocyc-1} via
\begin{equation} \label{non-bracket-constants}
\begin{split}
&\left[\bar{\mathbf{e}}_{\beta}, \bar{\mathbf{e}}_{\alpha}\right]_{_\Phi\bowtie}=\bar{C}_{\beta \alpha}^{\gamma}\bar{\mathbf{e}}_{\gamma} + \bar{C}_{\beta \alpha}^{a} \bar{\mathbf{e}}_{a} =\left [ \mathbf{e}_{\beta}\oplus 0,\mathbf{e}_{\alpha}\oplus 0 \right ]_{_\Phi\bowtie} =C^{\gamma}_{\beta \alpha}\mathbf{e}_{\gamma}\oplus 0,
\\
&\left[\bar{\mathbf{e}}_{\beta}, \bar{\mathbf{e}}_{a}\right]_{_\Phi\bowtie}=\bar{C}_{\beta a}^{\gamma}\bar{\mathbf{e}}_{\gamma} +\bar{C}_{\beta a}^{d}\bar{\mathbf{e}}_{d} =\left [ \mathbf{e}_{\beta}\oplus 0,0\oplus \mathbf{f}_{a} \right ]_{_\Phi\bowtie}=-L^{\gamma}_{a \beta}\mathbf{e}_{\gamma}\oplus -R_{a\beta}^d \mathbf{f}_d,
\\
&\left[\bar{\mathbf{e}}_{b}, \bar{\mathbf{e}}_{a}\right]_{_\Phi\bowtie}=\bar{C}_{b a}^{\gamma}\bar{\mathbf{e}}_{\gamma}+\bar{C}_{b a}^{d}\bar{\mathbf{e}}_{d} =\left[ 0\oplus \mathbf{f}_{b},0\oplus \mathbf{f}_{a} \right ]_{_\Phi\bowtie} = \Phi^\gamma_{ba}\mathbf{e}_\gamma \oplus \kappa^{d}_{b a}\mathbf{f}_{d}.
\end{split}
\end{equation}
As a result, structure constants of $\mathfrak{g} _{_\Phi\bowtie} \mathfrak{h}$ can be written as
\begin{equation} \label{sc-nonbracket}
\begin{split}
\bar{C}_{\beta \alpha}^{\gamma}&=C_{\beta \alpha}^{\gamma},\qquad \bar{C}_{\beta \alpha}^{a}=0, \qquad  \bar{C}_{\beta a}^{\gamma}=-L_{a \beta}^{\gamma},\\ \bar{C}_{\beta a}^{d}&=-R_{a \beta}^{d}, \qquad \bar{C}_{b a}^{\gamma}= \Phi^\gamma_{ba}, \qquad \bar{C}_{b a}^{d}=\kappa_{b a}^{d}.
\end{split}
\end{equation}

 \subsection{Double Cross Sum (Matched Pair) Lie Algebra} \label{doublecross}

 We shall recall the matched pair construction from \cite{Ma90,Ma902}. Let $(\G{g},[\bullet,\bullet]_{\G{g}})$ and $(\G{h},[\bullet,\bullet]_{\G{h}})$   be two Lie algebras admitting mutual actions 
\begin{equation}\label{matched-pair-mutual-actions}
\vartriangleright:\G{h} \otimes \G{g} \to \G{g}, \qquad \vartriangleleft:\G{h} \otimes \G{g} \to \mathfrak{h}.
\end{equation}
That is, we have that the following identities hold, for any $\xi,\xi' \in \mathfrak{g}$, and any $\eta,\eta' \in \mathfrak{h}$,  
\begin{equation} \label{matchpaircond}
\begin{split}
& [\eta,\eta']\vartriangleright \xi =  \eta \vartriangleright (\eta' \vartriangleright \xi) - \eta' \vartriangleright (\eta \vartriangleright \xi), \\
& \eta \vartriangleleft [\xi,\xi'] = (\eta \vartriangleleft \xi) \vartriangleleft \xi' -( \eta \vartriangleleft \xi') \vartriangleleft \xi. 
\end{split}
\end{equation}
The direct sum $\mathfrak{K}=\mathfrak{g}\oplus \mathfrak{h}$ can be endowed with a Lie algebra structure if some compatibility conditions are satisfied. 
\begin{theorem} \label{mp-prop}
The direct sum of two Lie algebras $\mathfrak{g}$ and $\mathfrak{h}$  under mutual actions \eqref{matched-pair-mutual-actions} is a Lie algebra if it is  equipped with the bilinear mapping
\begin{equation} 
\lbrack (\xi \oplus\eta ),\,(\xi'\oplus\eta')]_{\bowtie}=\big( [\xi,\xi'
]+\eta \vartriangleright \xi'-\eta'\vartriangleright \xi
\big)\oplus\big([\eta ,\eta']+\eta \vartriangleleft \xi'-\eta' \vartriangleleft \xi \big).  \label{mpla}
\end{equation}
satisfying the following compatibility conditions
\begin{equation} \label{comp-mpa}
\begin{split}
& [\eta,\eta'] \vartriangleleft \xi = [\eta,\eta' \vartriangleleft \xi] - [\eta',\eta \vartriangleleft \xi] + \eta \vartriangleleft (\eta' \vartriangleright \xi) - \eta' \vartriangleleft (\eta \vartriangleright \xi), \\
& \eta \vartriangleright [\xi,\xi'] = [\xi,\eta \vartriangleright \xi'] - [\xi',\eta \vartriangleright \xi] + (\eta \vartriangleleft \xi) \vartriangleright \xi' - (\eta \vartriangleleft \xi' )\vartriangleright \xi \\
\end{split}
\end{equation}
for any $\eta,\eta' \in \mathfrak{h}$, and any $\xi,\xi' \in \mathfrak{g}$.
\end{theorem}
If a Lie algebra $\mathfrak{K}$ is constructed in the realm of Proposition \ref{mp-prop}, then we call $\mathfrak{K}$ a matched pair (double cross sum) Lie algebra, and denote it as $\mathfrak{K}=\mathfrak{g}\bowtie \mathfrak{h}$, see also \cite{Zh10}. Notice that, we denote the matched pair Lie bracket \eqref{mpla} by $\lbrack \bullet, \bullet]_{\bowtie}$. If one of the actions in \eqref{matched-pair-mutual-actions} is trivial then we arrive at a semidirect product Lie algebra. So that matched pair Lie algebras are generalizations of the semidirect Lie algebras \cite{Cendra98,Cendra01,Cendra03,Holm98}. 

\textbf{From the extended structure to matched pairs.} Recall the extended structure in Subsection \ref{brzezinski}. In that realization, a more relaxed construction is given where $\G{h}$ is not necessarily assumed to be a Lie algebra.  
Consider particular case of that structure where the mapping $\Phi$ in \eqref{thm-m-h-const-maps} is taken to be  identically zero. In the realm of Theorem \ref{thm-m-h-const}, for the case of $\Phi\equiv 0$, the last condition in the list \eqref{cocycle-compatibility} gives that $\kappa$ mapping satisfies the Jacobi identity that is
\begin{equation} \label{kappa-Jac}
\circlearrowright \kappa(\eta,\kappa(\eta',\eta''))=0.	
\end{equation}
This reads that the vector space $\mathfrak{h}$ turns out to be a Lie algebra:
\begin{equation}\label{kappabracket}
[\eta,\eta']:=\kappa(\eta,\eta').
\end{equation} 
Further, for $\Phi\equiv 0$, the third line and the fifth line in the compatibility list \eqref{cocycle-compatibility} reduce to the action conditions \eqref{matchpaircond}, and the second and fourth lines in \eqref{cocycle-compatibility} become the matched pair compatibility conditions in \eqref{comp-mpa}. This observation says that Theorem \ref{mp-prop} is a particular case of Theorem \ref{thm-m-h-const}. That is, every matched pair Lie algebra is an extended structure.   
For the matched pair theory, Universal Lemma \eqref{uni-prop-bre} takes the following particular form

\begin{lemma}[Universal Lemma II] \label{universal-prop}
	Let $\mathfrak{K}$ be a Lie algebra with two Lie subalgebras $\mathfrak{g}$ and $\mathfrak{h}$ such that $\mathfrak{K}$ is isomorphic to the direct sum of $\mathfrak{g}$ and $\mathfrak{h}$ as vector spaces through the vector addition in $\mathfrak{K}$. Then $\mathfrak{K}$ is isomorphic to the matched pair $\mathfrak{g}\bowtie\mathfrak{h}$ as Lie algebras, and the mutual actions \eqref{matched-pair-mutual-actions} are derived from 
	\begin{equation} \label{mab-defn}
	[\eta,\xi] = \eta \vartriangleright \xi \oplus \eta \vartriangleleft \xi.
	\end{equation}
	Here, the inclusions of the subalgebras are defined to be 
	\begin{equation}
	\mathfrak{g} \longrightarrow \mathfrak{K}: \xi \to (\xi\oplus 0),\qquad \mathfrak{h} \longrightarrow \mathfrak{K}: \eta \to (0\oplus \eta).
	\end{equation}
\end{lemma}

\textbf{Coordinate realizations.} Assume that the coordinates are chosen as in Subsection  \ref{brzezinski}.  
To have the local characterization for the matched pair Lie algebra, we first analyse the structure constants given in \eqref{non-bracket-constants}. See that the first and the second lines remain the same but since $\Phi^\alpha_{ab}$ are all zero and, the constants $\kappa^{d}_{b a}$ turn out to be structure constants of the Lie algebra $\G{h}$. In order to highlight this, we denote the structure constants by $D^{d}_{b a}$. Let us record this 
\begin{equation}\label{nonbracket2-mp}
[\mathbf{e}_\alpha,\mathbf{e}_\beta]=C^\gamma_{\alpha \beta}\mathbf{e}_\gamma, \qquad [\mathbf{f}_a,\mathbf{f}_b]=\Phi^\alpha_{ab}\mathbf{e}_\alpha+\kappa^d_{ab}\mathbf{f}_d,
\end{equation}
Therefore, we have 
\begin{equation}
\left[\bar{\mathbf{e}}_{b}, \bar{\mathbf{e}}_{a}\right]=\bar{C}_{b a}^{\gamma}\bar{\mathbf{e}}_{\gamma}+\bar{C}_{b a}^{d}\bar{\mathbf{e}}_{d} =\left[ 0\oplus \mathbf{f}_{b},0\oplus \mathbf{f}_{a} \right ]_{\bowtie} =0\oplus D^{d}_{b a}\mathbf{f}_{d}.
\end{equation}
So that the structure constants of the matched Lie algebra are definitely equal to \eqref{sc-nonbracket} except that, in the present case, $\bar{C}_{b a}^{\gamma}=0$, and $\bar{C}_{b a}^{d}=D_{b a}^{d}$:  
\begin{equation} \label{sc-nonbracket-mp}
\begin{split}
\bar{C}_{\beta \alpha}^{\gamma}&=C_{\beta \alpha}^{\gamma},\qquad \bar{C}_{\beta \alpha}^{a}=0, \qquad  \bar{C}_{\beta a}^{\gamma}=-L_{a \beta}^{\gamma},\\ \bar{C}_{\beta a}^{d}&=-R_{a \beta}^{d}, \qquad \bar{C}_{b a}^{\gamma}= 0, \qquad \bar{C}_{b a}^{d}=D_{b a}^{d}.
\end{split}
\end{equation}

\subsection{2-cocycle Extension}\label{2coc-Sec}

In this subsection, we discuss another particular instance of extended structures exhibited in Subsection \ref{brzezinski}. In this case, we assume that the right action \eqref{Lieact1} of $\mathfrak{g}$ on $\mathfrak{h}$ and the Lie bracket on $\mathfrak{g}$ are trivial while all the other geometric ingredients of Theorem \ref{thm-m-h-const} are remaining the same. In other words, in this subsection,
\begin{equation}
	\eta \vartriangleright \xi=0, \qquad [\xi,\xi']=0
\end{equation}
for all $\xi$ and $\xi'$ in $\mathfrak{g}$, and for all $\eta$ in $\mathfrak{h}$.  
As a result of this selection, one observes that the extended Lie bracket \eqref{mpla-2-cocyc-1} and the list of conditions \eqref{cocycle-compatibility} will take particular forms. Let us examine them from the bottom to the top. Since the right action is trivial, the last condition in \eqref{cocycle-compatibility} turns out out to be the Jacobi identity \eqref{kappa-Jac} for $\kappa$. This manifests that the two-tuple $(\G{h},\kappa)$ becomes a Lie algebra. Accordingly, in this section, we denote $\kappa$ by a bracket notation $[\bullet,\bullet]$ as in \eqref{kappabracket}. On the other hand, the penultimate condition, namely the twisted 2-cocycle condition, in  \eqref{cocycle-compatibility} takes the particular form
\begin{equation}\label{ext-act}
\circlearrowright \Phi(\eta,[\eta',\eta'']) + \circlearrowright \eta \vartriangleright \Phi(\eta',\eta'') = 0.
\end{equation}
This determines $\Phi$ as a $\mathfrak{g}$-valued 2-cocycle on $\mathfrak{h}$ \cite{Fuks}. The second, the fourth and the fifth conditions in \eqref{cocycle-compatibility} are identically satisfied whereas the third line 
\begin{equation}
[\eta,\eta']\vartriangleright \xi =  \eta \vartriangleright (\eta' \vartriangleright \xi) - \eta' \vartriangleright (\eta \vartriangleright \xi)
\end{equation}
is exploiting that $\vartriangleright$ is a left Lie algebra action of $\mathfrak{h}$ on  $\mathfrak{g}$. Eventually, we arrive at the following reduced form of the Lie algebra bracket \eqref{mpla-2-cocyc-1}
 \begin{equation}\label{AAAAA}
 [\xi \oplus \eta, \xi' \oplus \eta']_{_\Phi\rtimes}= \big(\eta \vartriangleright  \xi'-\eta' \vartriangleright  \xi+ \Phi(\eta,\eta')    \big) \oplus[\eta,\eta']. 
 \end{equation} 
over the direct product space. In this case, we denote the total space by $\mathfrak{g}_{_\Phi\rtimes}\mathfrak{h}$ equipped with the Lie bracket \eqref{AAAAA} and call it $2$-cocycle extension of $\G{h}$ by a vector field $\G{g}$.

Once again, as a manifestation of Universal Lemma \ref{uni-prop-bre} we can discuss the decomposition point of view as follows. Assume a Lie algebra $\mathfrak{K}$ and one of its nontrivial centers, say $\G{g}$. Consider the decomposition $\mathfrak{g} \oplus \mathfrak{h}$
inducing nontrivial $\Phi$ and $\kappa$ mappings as in \eqref{phi}
and \eqref{kappa} and a left action \eqref{Lieact} then Universal Lemma \ref{uni-prop-bre} reads that, $\mathfrak{K}$ can be decomposed into a $2$-cocycle extension of $\G{h}$ by $\G{g}$ that is $\mathfrak{K}=\mathfrak{g}_{_\Phi\rtimes}\mathfrak{h}$. 

\textbf{Coordinate realizations.} Assume once more that the coordinates are chosen as in Subsection  \ref{brzezinski}.  
In this case, the Lie bracket on $\mathfrak{g}$ and, the right action of $\mathfrak{g}$ on $\mathfrak{h}$ are trivial. So that, the structure constants $C_{\beta \alpha}^\gamma$ of $\mathfrak{g}$
given in \eqref{nonbracket2} will be all zero while the constants
$\Phi^\alpha_{ab}$ and $D^d_{ab}$
determining $\Phi$ and $\kappa$ mappings remain the same. Since the right action is zero, $R_{a \beta}^d$ in \eqref{local-act} will be zero whereas the scalars $L_{a \alpha}^{\beta}$ are giving the left action. 
 If these aforementioned changes are applied to the system of equations \eqref{non-bracket-constants}, one can reach the structure constants of 2-cocycles. 

\section{Lie-Poisson Dynamics on Extensions}\label{LP-Ext-Sec}
Dual of a Lie algebra admits Lie-Poisson bracket according to the definition in \eqref{LP-Bra}. In the present section, following the order of the extensions and couplings in Section \ref{Sec-MP}, we compute associated Lie-Poisson brackets. 

\subsection{Lie-Poisson Systems on Duals of Extended Structures}
\label{2-LP-Coec}

Assume the Lie algebraic framework in Subsection \ref{brzezinski}, and let that all the conditions in Theorem \ref{thm-m-h-const} hold. 
Now we start with the left action $\vartriangleright$ in (\ref{Lieact}) then, freeze an element $\eta$ in $\mathfrak{h}$ in this operation to have a linear mapping $\eta \vartriangleright$ on the subalgebra $\mathfrak{g}$. This linear mapping and the dual action $\overset{\ast }{\vartriangleleft} \eta$ are
\begin{equation} \label{eta-star}
\begin{split}
\eta \vartriangleright&:\mathfrak{g}\longrightarrow  \mathfrak{g}, \qquad \xi \mapsto \eta \vartriangleright \xi
\\
\overset{\ast }{\vartriangleleft} &:\mathfrak{g}^*\otimes \G{h} \longrightarrow \mathfrak{g}^*, 
\qquad \langle \mu \overset{\ast }{\vartriangleleft} \eta, \xi \rangle=\langle \mu, \eta \vartriangleright \xi \rangle.
\end{split}
\end{equation}
This dual mapping is a right representation of $\mathfrak{h}$ on the dual space $\mathfrak{g}^*$. 
Later, by freezing $\xi\in \mathfrak{g}$ in the left action $\vartriangleright$ in (\ref{Lieact}), we define a linear mapping $\mathfrak{b}_\xi: \mathfrak{h} \mapsto \mathfrak{g}$. We record here this linear mapping $\mathfrak{b}_\xi$ and the dual mapping $\mathfrak{b}_\xi^*$ as
\begin{align}
\label{b}
\mathfrak{b}_\xi&: \mathfrak{h} \longrightarrow \mathfrak{g},\qquad \mathfrak{b}_\xi(\eta)=\eta\vartriangleright \xi,
\\
\label{b*}
\mathfrak{b}_\xi^*&:\mathfrak{g}^*\longrightarrow \mathfrak{h}^*, \qquad \langle \mathfrak{b}_\xi^*\mu,\eta \rangle= \langle \mu, \mathfrak{b}_\xi \eta \rangle = \langle \mu, \eta\vartriangleright \xi  \rangle.
\end{align}

Consider the right action $\vartriangleleft$ in (\ref{Lieact1}). In this operation, we freeze $\xi$ in $\G{g}$ to have an automorphism   on $\mathfrak{h}$, denoted by $\vartriangleleft\xi$. We write  $\vartriangleleft\xi$ and its dual $\xi \overset{\ast }{\vartriangleright}$ in the following display 
\begin{equation} \label{xi-star}
\begin{split}
\vartriangleleft\xi&:\mathfrak{h} \longrightarrow \mathfrak{h}, \qquad \eta \mapsto 
\eta \vartriangleleft\xi
\\
\overset{\ast }{\vartriangleright}&: \G{g}\times \mathfrak{h}^* \longrightarrow \mathfrak{h}^*, \qquad 
\langle \xi \overset{\ast }{\vartriangleright}\nu, \eta \rangle 
=\langle\nu,  \eta \vartriangleleft\xi\rangle.
\end{split}
\end{equation}
See that, $\overset{\ast }{\vartriangleright}$ is a left representation of $\mathfrak{g}$ on the dual $\mathfrak{h}^*$. Further, we freeze an element, say $\eta$ in $\mathfrak{h}$, in the right action (\ref{Lieact1}). This enables us to define a linear mapping $\mathfrak{a}_\eta$ from $\mathfrak{g}$ to $\mathfrak{h}$. Here are the mapping $\mathfrak{a}_\eta$ and its dual $\mathfrak{a}_\eta^*$ in a respective order
\begin{align} 
\mathfrak{a}_\eta&:\mathfrak{g}\mapsto \mathfrak{h}, \qquad \mathfrak{a}_\eta(\xi)=\eta\vartriangleleft \xi \label{a}
\\
\label{a*}
\mathfrak{a}_\eta^*&:\mathfrak{h}^*\mapsto \mathfrak{g}^*, \qquad \langle \mathfrak{a}_\eta^* \nu,\xi \rangle =
\langle \nu,\mathfrak{a}_\eta \xi \rangle=\langle \nu,\eta\vartriangleleft \xi \rangle.
\end{align}

Let us recall mappings $\Phi$ and $\kappa$  displayed in \eqref{phi} and \eqref{kappa}, respectively. Define two functions $\kappa_\eta$ and $\Phi_\eta$ as
\begin{align} 
\kappa_\eta &: \mathfrak{h} \rightarrow \mathfrak{h} \qquad \kappa_\eta(\eta') := \kappa(\eta,\eta')\label{kappan}
\\ 
\Phi_\eta &: \mathfrak{h} \rightarrow \mathfrak{g} \qquad \Phi_\eta(\eta') := \Phi(\eta,\eta')\label{phin}
\end{align}
where $\eta,\eta' \in \mathfrak{h}$, $\nu \in \mathfrak{h}^*$ and $\mu \in \mathfrak{g}^*$. 
 According to these definitions, the dual mappings are calculated as
\begin{align}  
\kappa_{\eta}^* &: \mathfrak{h}^* \rightarrow \mathfrak{h}^* \qquad \langle \kappa_{\eta}^*\nu, \eta' \rangle = \langle \nu, -\kappa_\eta(\eta') \rangle = -\langle \nu, \kappa(\eta,\eta') \rangle,\label{kappa_coad}
\\ \label{phi_coad}
\Phi_{\eta}^* &:  \mathfrak{g}^* \rightarrow \mathfrak{h}^* \qquad \langle \Phi_{\eta}^*\mu, \eta' \rangle = \langle \mu, -\Phi_\eta(\eta') \rangle = -\langle \mu, \Phi(\eta,\eta') \rangle,
\end{align}
respectively. 
\begin{proposition} \label{adfortwisted1}
	The adjoint action on $\mathfrak{g}_\Phi\bowtie  \mathfrak{h}$ is being the extended Lie bracket in \eqref{mpla-2-cocyc-1}, the coadjoint
	action $\ad^{\ast}$ of an element $\xi\oplus\eta$ in $\mathfrak{g}_\Phi\bowtie  \mathfrak{h}$ onto an element $\mu\oplus\nu$ in the dual space $\mathfrak{g}^*\oplus \mathfrak{h}^{\ast}$ is computed to be
	\begin{equation}  \label{ad-*-}
	\ad_{(\xi\oplus\eta)}^{\ast}(\mu\oplus\nu)=\underbrace{ \big(ad^{\ast}_{\xi} \mu -\mu \overset{\ast }{%
			\vartriangleleft}\eta - \mathfrak{a}_{\eta}^{\ast}\nu\big)}_{\in ~ \mathfrak{g}^*}\oplus \underbrace{ \big(
		\kappa^{\ast}_{\eta} \nu+\Phi_{\eta}^*\mu +\xi \overset{\ast }{\vartriangleright}\nu+ \mathfrak{b}%
		_{\xi}^{\ast}\mu\big )}_{\in ~ \mathfrak{h}^*}.
	\end{equation}
	Here, $ad^{\ast}$ (in italic font) is for the infinitesimal coadjoint actions of subalgebras to their duals.
\end{proposition}
By using the equations \eqref{kappa_coad} and \eqref{phi_coad}, (plus/minus) extended Lie-Poisson bracket is computed as
\begin{equation} 
\begin{split} \label{LiePoissononcocycle}
\left\{ \mathcal{H},\mathcal{F}\right\}_{_\Phi\bowtie}(\mu\oplus\nu) =& \pm \left \langle \mu\oplus\nu, \left[ \Big(\frac{\delta \mathcal{H}}{\delta \mu}\oplus \frac{\delta \mathcal{H}}{\delta \nu} \Big), 
\Big(\frac{\delta \mathcal{F}}{\delta \mu}\oplus \frac{\delta \mathcal{F}}{\delta \nu} \Big)
\right ]_{_\Phi\bowtie}\right \rangle  \\
=&
\pm \left\langle \mu ,\left[\frac{\delta \mathcal{H}}{%
	\delta \mu},\frac{\delta \mathcal{F}}{\delta \mu}\right] \right\rangle
\pm\left\langle \nu ,\kappa \left(\frac{\delta \mathcal{H}}{\delta \nu},\frac{\delta
	\mathcal{F}}{\delta \nu}\right)\right\rangle \pm \underbrace{\left\langle \mu , \Phi \left(\frac{\delta \mathcal{H}}{\delta \nu},\frac{\delta
		\mathcal{F}}{\delta \nu}\right) \right\rangle}_{\text{A: from twisted cocycle}} \\ 
&\pm \underbrace{\left\langle \mu ,\frac{\delta \mathcal{H}}{\delta \nu}
	\vartriangleright \frac{\delta \mathcal{F}}{\delta \mu}\right\rangle
	\mp\left\langle \mu ,\frac{\delta \mathcal{F}}{\delta \nu}\vartriangleright
	\frac{\delta \mathcal{H}}{\delta \mu}\right\rangle}_ {\text{B: action of $\mathfrak{h}$ on $\mathfrak{g}$ from the left}}
\pm\underbrace{ \left\langle \nu ,\frac{\delta \mathcal{H}}{\delta \nu}\vartriangleleft \frac{\delta \mathcal{F}}{\delta \mu}\right\rangle
	\mp\left\langle \nu ,\frac{\delta \mathcal{F}}{\delta \nu}\vartriangleleft
	\frac{\delta \mathcal{H}}{\delta \mu}\right\rangle} _{\text{C: action of $\mathfrak{g}$ on $\mathfrak{h}$ from the right}}
\end{split}
\end{equation}
for two functions $\mathcal{H}, \mathcal{F}$. We assume the reflexivity condition which reads that ${\delta \mathcal{H}}/{\delta \mu}$ and ${\delta \mathcal{F}}/{\delta \mu}$ are elements of $\mathfrak{g}$ whereas ${\delta \mathcal{H}}/{\delta \nu}$ and ${\delta \mathcal{F}}/{\delta \nu}$ are elements of $\mathfrak{h}$.  The Lie bracket on the first line in \eqref{LiePoissononcocycle} is the extended Lie bracket $[\bullet,\bullet]_{_\Phi\bowtie}$ in \eqref{mpla-2-cocyc-1}. In the Poisson bracket, the term labelled by A is a manifestation of the existence of twisted cocycle $\Phi$. The terms labelled by B are  due to the \textit{left action} of $\mathfrak{h}$ on $\mathfrak{g}$ whereas the terms labelled by C are due to the right action of $\mathfrak{g}$ on $\mathfrak{h}$.

Recall 
the (plus/minus) Lie-Poisson equation in \eqref{eqnofmotion} determined as a coadjoint flow. In the light of the Lie-Poisson bracket \eqref{LiePoissononcocycle}, governed by a Hamiltonian function $\mathcal{H}=\mathcal{H}(\mu,\nu)$), for the present picture, the (plus/minus) Lie-Poisson equation is computed as
\begin{equation}\label{LPEghcocycle}
\begin{split}
& \underbrace{\dot{\mu} = \pm
	ad^{\ast}_{\frac{\delta\mathcal{H}}{\delta\mu}}(\mu)}_{\text{Lie-Poisson Eq. on} \  \mathfrak{g}^*}\mp\underbrace{\mu\overset{\ast }{\vartriangleleft}
	\frac{\delta\mathcal{H}}{\delta\nu}}_{\text{action of} \ \mathfrak{h}}
\mp
\underbrace{\mathfrak{a}_{\frac{\delta\mathcal{H}}{\delta\nu}}^{\ast}\nu,}
_{\text{action of} \ \mathfrak{g}} 
\\
&\dot{\nu} =\pm
\kappa^{\ast}_{\frac{\delta\mathcal{H}}{\delta\nu}}(\nu)
\pm \underbrace{\Phi_{\frac{\delta\mathcal{H}}{\delta\nu}}^*(\mu)}
_{\text{twisted cocycle}}
\pm
\underbrace{
	\frac{\delta\mathcal{H}}{\delta\mu} \overset{\ast }{\vartriangleright}\nu}_{\text{action of} \ \mathfrak{g}}
\pm \underbrace{\mathfrak{b}
	_{\frac{\delta\mathcal{H}}{\delta\mu}}^{\ast}\mu.
}_{\text{action of} \ \mathfrak{h}}
\end{split}
\end{equation}
Notice that we have determined and labelled the terms in the Lie-Poisson in order to identify them. As depicted in \eqref{LPEghcocycle}, one can easily follow how the Lie-Poisson  dynamics on $\G{g}^*$ is extended by addition of terms coming from the mutual actions of $\G{h}$ and $\G{g}$ on each other as well as from the twisted 2-cocycle term. 

\textbf{Coordinate realizations.} We follow the notation in Subsection  \ref{brzezinski}. Recall, $(N+M)$-dimensional  extended structure $\mathfrak{K}=\mathfrak{g}\,_\Phi\bowtie \mathfrak{h}$. Denote the dual basis of $\mathfrak{g}^*$ and $\mathfrak{h}^*$ by $\{\mathbf{e}^\alpha\}$ and  $\{\mathbf{f}^a\}$, respectively. Then, define the dual basis 
\begin{equation}
\{\bar{\mathbf{e}}^{\alpha},\bar{\mathbf{e}}^a\}\subset \mathfrak{g}^* \oplus  \mathfrak{h}^*, \qquad \bar{\mathbf{e}}^{\alpha}=\mathbf{e}^\alpha\oplus 0,\qquad \bar{\mathbf{e}}^a=0\oplus \mathbf{f}^a
\end{equation}
on the dual space $\mathfrak{g}^*\oplus \mathfrak{h}^*$. 
Using this basis, we can write an element of $ \mathfrak{g}^* \oplus \mathfrak{h}^* $ as follows
\begin{equation}
(\mu,\nu)=\mu_{\alpha}\bar{\mathbf{e}}^{\alpha}+\nu_{a}\bar{\mathbf{e}}^a.
\end{equation}
In this picture, the mappings \eqref{eta-star} and \eqref{xi-star} turn out to be
\begin{equation}\label{1}
(\mu_\alpha \mathbf{e}^\alpha) \overset{\ast}{\vartriangleleft} (\eta^a \mathbf{f}_a)=\mu_\alpha \eta^a L_{a\beta}^\alpha \mathbf{e}^\beta
, \qquad 
(\xi^\alpha \mathbf{e}_\alpha) \overset{\ast}{\vartriangleright} (\nu_a\mathbf{f}^a)=\nu_a \xi^\alpha R_{b\alpha}^a \mathbf{f}^b,
\end{equation}
where $L_{a\beta}^\alpha$ and $R_{b\alpha}^a$ are scalars in \eqref{local-act} determining the \textit{actions}. Later, we compute the dual mappings in (\ref{b*}), (\ref{a*}) and also \eqref{kappa_coad}, \eqref{phi_coad} in terms of local coordinates as follows:
\begin{align}\label{6}
\mathfrak{b}^*_{(\xi^\alpha \mathbf{e}_\alpha)} (\mu_\alpha \mathbf{e}^\alpha)& =\mu_\alpha \xi^\beta L_{a \beta}^\alpha \mathbf{f}^a, \qquad 
\mathfrak{a}^*_{(\eta^a \mathbf{f}_a)} (\nu_a \mathbf{e}^a)=\nu_a \eta^b R_{b \alpha}^a \mathbf{e}^\alpha,
\\
	\kappa^{\ast}_\eta \nu &=-\kappa^a_{bd} \nu_a  \eta^b \mathbf{f}^d, \qquad \Phi^{\ast}_\eta \mu= -\Phi_{b k}^\alpha \mu_\alpha \eta^b \mathbf{f}^k. 
\end{align}
Therefore, the (plus/minus) Lie-Poisson bracket \eqref{LiePoissononcocycle} is written in coordinates as  
\begin{equation}
\begin{split}\label{Lie-pois-double-non-bracket}
\left \{ \mathcal{H},\mathcal{F} \right \}_{_\Phi\bowtie}(\mu\oplus\nu)=&\pm \mu_\alpha C_{\beta \gamma}^\alpha \frac{\partial  \mathcal{H}}{\partial \mu_\beta}\frac{\partial \mathcal{F}}{\partial \mu_\gamma} 
\pm
\nu_a \kappa_{b d}^a \frac{\partial \mathcal{H}}{\partial \nu_b}\frac{\partial \mathcal{F}}{\partial \nu_d} 
\pm
 \mu_\alpha \Phi_{bk}^\alpha \frac{\partial \mathcal{H}}{\partial \nu_b} \frac{\partial \mathcal{F}}{\partial \nu_k} \\
&
\pm
\mu_\alpha L_{a \beta }^\alpha \big(\frac{\partial \mathcal{H}}{\partial \nu_a}\frac{\partial \mathcal{F}}{\partial \mu_\beta}
-\frac{\partial \mathcal{F}}{\partial \nu_a}\frac{\partial \mathcal{H}}{\partial \mu_\beta}\big)
\pm
\nu_a R_{b \beta}^a \big( \frac{\partial \mathcal{H}}{\partial \nu_b}\frac{\partial \mathcal{F}}{\partial \mu_\beta}- \frac{\partial \mathcal{F}}{\partial \nu_b}\frac{\partial \mathcal{H}}{\partial \mu_\beta}\big).
\end{split}
\end{equation}
whereas the (plus/minus) Lie-Poisson dynamics \eqref{LPEghcocycle} as
\begin{equation}\label{Lie-pois-nonbracket-dyn}
\begin{split} 
\dot{\mu}_\beta&=
\pm
\mu_\rho C_{\beta \alpha}^\rho \frac{\partial \mathcal{H}}{\partial \mu_\alpha} 
\mp
\mu_\alpha L_{a \beta }^\alpha  \frac{\partial \mathcal{H}}{\partial \nu_a} 
\mp
\nu_a R_{b \beta }^a \frac{\partial \mathcal{H}}{\partial \nu_b} ,
\\
\dot{\nu}_d&=
\pm
\mu_\alpha \Phi_{db}^\alpha  \frac{\partial \mathcal{H}}{\partial \nu_b} 
\pm
\nu_a \kappa_{d b}^a \frac{\partial \mathcal{H}}{\partial \nu_b}   
\pm
\nu_a R_{d \alpha }^a  \frac{\partial \mathcal{H}}{\partial \mu_\alpha} 
\pm
\mu_\alpha L_{d \beta }^\alpha \frac{\partial \mathcal{H}}{\partial \mu_\beta} .
\end{split}
\end{equation}
\subsection{Matched Lie-Poisson Systems}\label{Sec-MP-LP}

In Subsection \ref{doublecross}, matched pair Lie algebra $\G{g}\bowtie \G{h}$ is realized as a particular instance of extended structure by choosing the twisted 2-cocycle $\Phi$ as trivial. So that, it is argued that, for a matched pair Lie algebra, both of the constitutive spaces, $\G{g}$ and $\G{h}$ are Lie subalgebras. Therefore, in this case, the duals of each of these subspaces, namely  $\G{g}^*$ and $\G{h}^*$, admit Lie-Poisson flows. This lets us to claim that the Lie-Poisson dynamics on the dual $\mathfrak{g}^*\oplus \mathfrak{h}^*$ of matched pair can be considered as the collective motion of two Lie-Poisson subdynamics \cite{OS16}. Algebraically, this corresponds to matching of two Lie coalgebras.  We refer  \cite{Mi80} for the Lie coalgebra structure of the dual space. 

On the dual of a matched pair, both of the dual actions $\overset{\ast }{\vartriangleleft} $ and $\overset{\ast }{\vartriangleright}$,  exhibited in $\eqref{eta-star}$
and $\eqref{xi-star}$ respectively, are equally valid. Notice that, for the present discussion $\vartriangleright$ is a true left action of $\G{h}$ on $\G{g}$ so that $\overset{\ast }{\vartriangleleft} $ is a true right dual action of $\G{h}$ on $\G{g}^*$. This is not true for extended structure since $\G{h}$ is not assumed to be a Lie subalgebra. It is immediate to observe that the dual mappings $\mathfrak{b}_\xi^*$ and $\mathfrak{a}_\eta^*$, in $\eqref{b*}$
and $\eqref{a*}$ respectively, are remaining the same. The difference of matched pair construction form the extended structure is that the $\kappa$ mapping, in \eqref{kappa}, is a Lie bracket and that $\Phi$ mapping, in \eqref{phi}, is zero. As previously stated, we prefer to denote $\kappa$ by a bracket, so we write the mapping \eqref{kappa_coad} as 
\begin{equation}\label{kappaequivad}
	\kappa_{\eta}^*\nu=ad^{\ast}_{\eta} \nu.
\end{equation}
These observations lead us to the following proposition as a particular case of Proposition \ref{adfortwisted1}, see also \cite{EsPaGr17,OS16}. 
\begin{proposition} \label{ad-*-prop}
	The adjoint action on $\mathfrak{g}\bowtie  \mathfrak{h}$ is being the matched pair Lie bracket in \eqref{mpla},	the  coadjoint
	action $\ad^{\ast}$ of an element $\xi\oplus\eta$ in $\mathfrak{g}\oplus \mathfrak{h}$ onto an element $\mu\oplus\nu$ in the dual space $\mathfrak{g}^*\oplus \mathfrak{h}^{\ast}$ is computed to be
	\begin{equation}  \label{ad-*}
	\ad_{(\xi\oplus\eta)}^{\ast}(\mu\oplus\nu)=\underbrace{ \big(ad^{\ast}_{\xi} \mu -\mu \overset{\ast }{%
			\vartriangleleft}\eta - \mathfrak{a}_{\eta}^{\ast}\nu\big)}_{\in ~ \mathfrak{g}^*}\oplus \underbrace{ \big(
		ad^{\ast}_{\eta} \nu +\xi \overset{\ast }{\vartriangleright}\nu+ \mathfrak{b}%
		_{\xi}^{\ast}\mu\big )}_{\in ~ \mathfrak{h}^*}.
	\end{equation}
	\end{proposition}
Later, on the dual space $\mathfrak{g}^\ast \oplus \mathfrak{h}^\ast$, the (plus/minus) Lie-Poisson bracket of double cross sum is computed to be
\begin{equation} 
\begin{split} \label{LiePoissonongh}
& \left\{ \mathcal{H},\mathcal{F}\right\}_{\bowtie}(\mu\oplus\nu) = \underbrace{\pm \left\langle \mu ,\left[\frac{\delta \mathcal{H}}{%
		\delta \mu},\frac{\delta \mathcal{F}}{\delta \mu}\right]\right\rangle
	\pm \left\langle \nu ,\left[\frac{\delta \mathcal{H}}{\delta \nu},\frac{\delta
		\mathcal{F}}{\delta \nu}\right]\right\rangle}_{\text{A: direct product}} \,
\underbrace{\mp\left\langle \mu ,\frac{\delta \mathcal{H}}{\delta \nu}
	\vartriangleright \frac{\delta \mathcal{F}}{\delta \mu}\right\rangle
	\pm  \left\langle \mu ,\frac{\delta \mathcal{F}}{\delta \nu}\vartriangleright
	\frac{\delta \mathcal{H}}{\delta \mu}\right\rangle}_ {\text{B: via the left action of $\mathfrak{h}$ on 
		$\mathfrak{g}$}} \\
&  \hspace{3cm} \underbrace{\mp \left\langle \nu ,\frac{\delta \mathcal{H}}{\delta \nu}\vartriangleleft \frac{\delta \mathcal{F}}{\delta \mu}\right\rangle
	\pm \left\langle \nu ,\frac{\delta \mathcal{F}}{\delta \nu}\vartriangleleft
	\frac{\delta \mathcal{H}}{\delta \mu}\right\rangle} _{\text{C: via the
		right action of $\mathfrak{g}$ on 
		$\mathfrak{h}$}}.
\end{split}
\end{equation}
Notice that, the terms labelled by A are just the sum of individual Poisson brackets on the dual spaces $\mathfrak{g}^* $ and $\mathfrak{h}^* $ of the constitutive Lie subalgebras $\mathfrak{g}$ and $\mathfrak{h}$, respectively. The terms labeled by B is a result of the left action of $\mathfrak{h}$ on $\mathfrak{g}$ whereas the terms labelled by C is due to the right action of $\mathfrak{g}$ on $\mathfrak{h}$. For the case of one-sided actions, that is semidirect product theories, B or C drops. If there is no action then, both  B and C drop. In the light of the (plus/minus) matched pair Lie-Poisson bracket \eqref{LiePoissonongh}, matched pair Lie-Poisson equations generated by a Hamiltonian function 
$\mathcal{H}=\mathcal{H}(\mu,\nu)$ on $\mathfrak{g}^\ast\oplus\mathfrak{h}^\ast$ is computed to be
\begin{equation}\label{LPEgh}
\begin{split}
& \underbrace{\dot{\mu}= \pm
	ad^{\ast}_{\frac{\delta\mathcal{H}}{\delta\mu}}(\mu)}_{\text{Lie-Poisson Eq. on }\mathfrak{g}^*}\,
\underbrace{\mp\mu\overset{\ast }{\vartriangleleft}
	\frac{\delta\mathcal{H}}{\delta\nu}}
_{\text{action of } \mathfrak{h}}\,
\underbrace{\pm\mathfrak{a}_{\frac{\delta\mathcal{H}}{\delta\nu}}^{\ast}\nu}
_{\text{action of }\mathfrak{g}}, 
\\
&\underbrace{\dot{\nu} =\pm
	ad^{\ast}_{\frac{\delta\mathcal{H}}{\delta\nu}}(\nu)}_
{\text{Lie-Poisson Eq. on }\mathfrak{h}^*}\,
\underbrace{\mp
	\frac{\delta\mathcal{H}}{\delta\mu} \overset{\ast }{\vartriangleright}\nu}_{\text{action of }\mathfrak{g}}\,
 \underbrace{\pm\mathfrak{b}
	_{\frac{\delta\mathcal{H}}{\delta\mu}}^{\ast}\mu
}_{\text{action of }\mathfrak{h}}
.
\end{split}
\end{equation}
The first terms on the right hand sides are the individual equations of motions. The other terms are the dual and cross actions appearing as  manifestations of the mutual actions.

\textbf{Coordinate realizations.} Recall the Lie-Poisson bracket in \eqref{Lie-pois-double-non-bracket} and the Lie-Poisson equations \eqref{Lie-pois-nonbracket-dyn} computed for the case of extended structures. In the matched pair case, we take the constants determining twisted 2-cocycle as zero that is $\Phi_{b a}^{\gamma}=0$, and the structure constants for the Lie algebra $\G{h}$ as $\kappa_{b a}^{d}=D_{b a}^{d}$. So that, the (plus/minus) matched Lie-Poisson bracket \eqref{LiePoissonongh} takes the following form in the coordinates  
\begin{equation}
\begin{split}\label{Lie-pois-double-non-bracket-mp}
\left \{ \mathcal{H},\mathcal{F} \right \}_{_\Phi\bowtie}(\mu\oplus\nu)=&\pm \mu_\alpha C_{\beta \gamma}^\alpha \frac{\partial  \mathcal{H}}{\partial \mu_\beta}\frac{\partial \mathcal{F}}{\partial \mu_\gamma} 
\pm
\nu_a D_{b d}^a \frac{\partial \mathcal{H}}{\partial \nu_b}\frac{\partial \mathcal{F}}{\partial \nu_d} 
\pm
\mu_\alpha L_{a \beta }^\alpha \big(\frac{\partial \mathcal{H}}{\partial \nu_a}\frac{\partial \mathcal{F}}{\partial \mu_\beta}
-\frac{\partial \mathcal{F}}{\partial \nu_a}\frac{\partial \mathcal{H}}{\partial \mu_\beta}\big)
\\
&
\pm
\nu_a R_{b \beta}^a \big( \frac{\partial \mathcal{H}}{\partial \nu_b}\frac{\partial \mathcal{F}}{\partial \mu_\beta}- \frac{\partial \mathcal{F}}{\partial \nu_b}\frac{\partial \mathcal{H}}{\partial \mu_\beta}\big). \\
\end{split}
\end{equation}
The matched (plus/minus) Lie-Poisson dynamics in \eqref{LPEgh} is computed to be
\begin{equation}\label{Lie-pois-nonbracket-dyn-mp}
\begin{split} 
\dot{\mu}_\beta&=
\pm
\mu_\rho C_{\beta \alpha}^\rho \frac{\partial \mathcal{H}}{\partial \mu_\alpha} 
\mp
\mu_\alpha L_{a \beta }^\alpha  \frac{\partial \mathcal{H}}{\partial \nu_a} 
\mp
\nu_a R_{b \beta }^a \frac{\partial \mathcal{H}}{\partial \nu_b} ,
\\
\dot{\nu}_d&=
\pm
\nu_a D_{d b}^a \frac{\partial \mathcal{H}}{\partial \nu_b}   
\pm
\nu_a R_{d \alpha }^a  \frac{\partial \mathcal{H}}{\partial \mu_\alpha} 
\pm
\mu_\alpha L_{d \beta }^\alpha \frac{\partial \mathcal{H}}{\partial \mu_\beta} .
\end{split}
\end{equation}

\subsection{Lie Poisson Dynamics on Duals of 2-cocycles}
\label{LPD-2c}

In Subsection \ref{2coc-Sec}, it is shown that 2-cocycle extension $\mathfrak{g}_{_\Phi\rtimes}\mathfrak{h}$ of a Lie algebra $h$ over its representation space $\G{g}$, as a particular case of extended structure $\mathfrak{g}_{_\Phi\bowtie}\mathfrak{h}$.   
Thus, the Lie-Poisson dynamics on the dual space of 2-cocycle extensions would be derived through the Lie-Poisson dynamics on the dual space of extended structures which is given in Subsection \ref{2-LP-Coec}. 
So that, we follow Subsection \ref{2coc-Sec} and make particular choice that the Lie bracket on $\mathfrak{g}$ is trivial. This results with several consequences. One is that the left action $\vartriangleright$ and the right dual action $\overset{\ast }{\vartriangleleft}$, in \eqref{eta-star}, are trivial. Also, observe that the coadjoint action on $\mathfrak{g}^*$ becomes identically zero. In addition, in this case, $\Phi$ turns out to be a true 2-cocycle and, $\kappa$ becomes a Lie bracket on the vector space $\mathfrak{h}$. Thus, the dual of $\kappa$ suits the coadjiont action as in \eqref{kappaequivad}. We employ all these modifications to the Lie-Poisson bracket \eqref{LiePoissononcocycle} on the dual of  extended structure to arrive at the (plus/minus)
Lie-Poisson bracket on the dual of 2-cocycle extension $\mathfrak{g}_{_\Phi\rtimes}\mathfrak{h}$ as
\begin{equation} \label{centextLiepois}
\left\{ \mathcal{H},\mathcal{F}\right\}_{_\Phi\rtimes}(\mu\oplus\nu) = \pm \left\langle \nu ,[\frac{\delta \mathcal{H}}{\delta \nu},\frac{\delta
 	\mathcal{F}}{\delta \nu}] \right\rangle \pm \underbrace{\left\langle \mu , \Phi \left(\frac{\delta \mathcal{H}}{\delta \nu},\frac{\delta
		\mathcal{F}}{\delta \nu}\right) \right\rangle}_{\text{A: 2-cocycle}} \,
  \underbrace{\pm\left\langle \mu ,\frac{\delta \mathcal{H}}{\delta \nu}
	\vartriangleright \frac{\delta \mathcal{F}}{\delta \mu}\right\rangle
	\mp \left\langle \mu ,\frac{\delta \mathcal{F}}{\delta \nu}\vartriangleright
	\frac{\delta \mathcal{H}}{\delta \mu}\right\rangle}_ {\text{B: left action of $\mathfrak{h}$ on $\mathfrak{g}$}}.
\end{equation}
Here, the first term on the right hand side is the Lie-Poisson bracket on $\G{h}^*$. The term labelled as A is due to 2-cocycle $\Phi$ whereas the terms labelled as B are due to the left action of $\G{h}$ on $\G{g}$. For the Lie-Poisson bracket \eqref{centextLiepois},  the Lie-Poisson equations governed by a Hamiltonian function $\mathcal{H}=\mathcal{H}(\mu,\nu)$is computed to be
\begin{equation}\label{centextLiepois2}
\dot{\mu} = \mp \underbrace{\mu\overset{\ast }{\vartriangleleft}
	\frac{\delta\mathcal{H}}{\delta\nu}}
_{\text{action of} \ \mathfrak{h}},\qquad
\underbrace{\dot{\nu}  =\pm
	ad^{\ast}_{\frac{\delta\mathcal{H}}{\delta\nu}}(\nu)}_
{\text{Lie-Poisson Eq. on }\mathfrak{h}^*}\,
\underbrace{\pm \Phi_{\frac{\delta\mathcal{H}}{\delta\nu}}^*(\mu)}
_{\text{2-cocycle}}\,
 \underbrace{\pm \mathfrak{b}
	_{\frac{\delta\mathcal{H}}{\delta\mu}}^{\ast}\mu
}_{\text{action of} \ \mathfrak{h}}.
\end{equation}
A direct observation gives that the Lie-Poisson equation \eqref{centextLiepois2} is a particular case of 
the Lie Poisson equation \eqref{LPEghcocycle} where $\overset{\ast }{\vartriangleright}$, $\mathfrak{a}^*$ are zero. 

\textbf{Coordinate realizations.} In Subsection \ref{2coc-Sec}, 2-cocycle extension is written in terms of coordinates for finite dimensional cases. Referring to Subsection \ref{2-LP-Coec}, we write the Lie-Poisson bracket \eqref{centextLiepois} in coordinates as 
\begin{equation}
\label{Lie-pois-double-non-bracket-tc}
\left \{ \mathcal{H},\mathcal{F} \right \}_{_\Phi\bowtie}(\mu\oplus\nu)= 
\pm
\nu_a D_{b d}^a \frac{\partial \mathcal{H}}{\partial \nu_b}\frac{\partial \mathcal{F}}{\partial \nu_d} 
\pm
 \mu_\alpha \Phi_{bk}^\alpha \frac{\partial \mathcal{H}}{\partial \nu_b} \frac{\partial \mathcal{F}}{\partial \nu_k} 
\pm
\mu_\alpha L_{a \beta }^\alpha \big(\frac{\partial \mathcal{H}}{\partial \nu_a}\frac{\partial \mathcal{F}}{\partial \mu_\beta}
-\frac{\partial \mathcal{F}}{\partial \nu_a}\frac{\partial \mathcal{H}}{\partial \mu_\beta}\big).
\end{equation}
Notice that, we write the structure constants on $\G{h}$ as $D_{b d}^a $. Further, we can write the Lie Poisson equations in \eqref{centextLiepois2} as 
\begin{equation}\label{Lie-pois-nonbracket-dyn-tc}
\dot{\mu} _\beta=
\mp
\mu_\alpha L_{a \beta }^\alpha  \frac{\partial \mathcal{H}}{\partial \nu_a} , \qquad 
\dot{\nu}_d=
\pm
\mu_\alpha \Phi_{db}^\alpha  \frac{\partial \mathcal{H}}{\partial \nu_b} 
\pm
\nu_a D_{d b}^a \frac{\partial \mathcal{H}}{\partial \nu_b}   
\pm
\mu_\alpha L_{d \beta }^\alpha \frac{\partial \mathcal{H}}{\partial \mu_\beta} .
\end{equation}

\section{Illustration: Decomposing 3 Particles BBGKY Hierarchy} \label{examples}

In this section, we provide a concrete example to Lie algebraic constructions and the Lie-Poisson structures introduced up to now. For this, we focus on BBGKY hierarchy in plasma dynamics \cite{Ha}. In \cite{marsden1984hamiltonian}, it is proved that BBGKY hierarchy can be recasted as a Lie-Poisson equation. The formulations presented in that study is for $n>3$. In the present work, we focus on  $n=3$  which is missing in \cite{marsden1984hamiltonian}. Accordingly, we first determine the dynamics of BBGKY hierarchy for $n=3$. Then, we shall investigate its Lie-Poisson form. 
Two decomposition of the BBGKY dynamics will be presented (1) as a matched pair and, (2) as an extended structure. 
 
\subsection{BBGKY Dynamics for 3 Particles}

Assume that a plasma rests in a finite $3D$ manifold $Q$ in $\mathbb{R}^3$. Being a cotangent bundle, $P=T^*Q$ is a symplectic and a Poisson manifold \cite{MaRa13,LiMa12}. Define the product symplectic space $P^3=P\times P \times P$ endowed with the product symplectic and product Poisson structures. 

We denote the $3$ particle density function by $f_3=f_3 (z_1,z_2,z_3)$ on $P^3$. The dynamics of $3$-particle plasma density function is governed by the Vlasov equation
\begin{equation}\label{Vlasov}
\frac{\partial f_3}{\partial t}=\{H(z_1,z_2,z_3),f_3(z_1,z_2,z_3)\}
\end{equation}  
where $H$ is the total energy of the plasma particles \cite{MaRaScSpWe83,MaWe82,Mo82}. Here, $\{\bullet,\bullet\}$ denotes the canonical Poisson bracket on $P_3$ with respect to the variables $z_1,z_2,z_3$. We refer \cite{EsGRGuPa19,EsGu12,GiHoTr,Gu10} for some recent studies on Vlasov motion related with the geometry  here. In the present work, we assume that, the particle energy function is in form
\begin{equation}\label{en-func}
\begin{split}
H&=\sum_i H_1(z_i) + \sum_{i<j} H_2(z_i,z_j)+ H_3(z_1,z_2,z_3)\\ &= H_1(z_1)+H_1(z_2)+H_1(z_3)+H_2(z_1,z_2)+H_2(z_1,z_3)+H_2(z_2,z_3)+H_3(z_1,z_2,z_3).
\end{split}
\end{equation}
Here, the functions $H_2$ and $H_3$ are assumed to be symmetric. 
See that, the dynamical equations in\cite{marsden1984hamiltonian} is for $H_3=0$. We present here the theory with a   
nontrivial $H_3$. 

\textbf{Dynamics of moments.} Now, we determine the moments of the plasma density function $f_3=f_3 (z_1,z_2,z_3)$ on $P^3$ as
\begin{equation} \label{moments}
\begin{split}
f_1(z_1):=3\int f_3(z_1,z_2,z_3)dz_2 dz_3
\\
f_2(z_1,z_2):=6\int f_3(z_1,z_2,z_3) dz_3.
\end{split}
\end{equation}
 To find the dynamics of the moments $f_1$ and $f_2$, we simply take the partial derivatives of \eqref{moments} then, directly substitute the Vlasov equation   \eqref{Vlasov} into these expressions. 
 
In order to arrive at the equation of governing the moment functions, we record the following identity, on the Poisson space $P$,
\begin{equation} \label{eqn}
\int \{h(z),f(z)\}_z ~ dz=0,
\end{equation}
which is valid for any two functions. The equation \eqref{eqn} is the result of that we have omitted the boundary terms. In \eqref{eqn}, $\{\bullet ,\bullet\}_z $ stands for the Poisson bracket on $P$. We shall be referring to this identity in the sequel.  
We compute the dynamics of $f_1$ as  
\begin{equation} \label{Dyn-f-1}
\begin{split}
	\frac{\partial f_1}{\partial t}&=3 \int \{H(z_1,z_2,z_3), f_3(z_1,z_2,z_3)\} dz_2dz_3 
	\\
	& =3 \int \Big \{ \sum_i H_1(z_i) + \sum_{i<j} H_2(z_i,z_j)+ H_3(z_1,z_2,z_3),f_3(z_1,z_2,z_3) \Big  \} dz_2dz_3
	\\
	& =3 \int \Big  \{ \sum_i H_1(z_1)  ,f_3(z_1,z_2,z_3) \Big  \}_{z_1} dz_2dz_3	
	+3 \int 
\{  H_2(z_1,z_2)+H_2(z_1,z_3), f_3(z_1,z_2,z_3)
\}_{z_1}	dz_2dz_3	
 \\
&\qquad +3 \int   \{ H_3(z_1,z_2,z_3),f_3(z_1,z_2,z_3)   \}_{z_1}dz_2dz_3
\\
& = \Big  \{ \sum_i H_1(z_1)  ,3 \int f_3(z_1,z_2,z_3) dz_2dz_3	\Big  \}_{z_1} 
	+ \int  \Big 
\{  H_2(z_1,z_2), 6\int f_3(z_1,z_2,z_3) dz_3
\Big \}_{z_1}	dz_2	\\
& \qquad +
3 \int   \{ H_3(z_1,z_2,z_3),f_3(z_1,z_2,z_3)   \}_{z_1}dz_2dz_3
\\	
	&=\{H_1(z_1),f_1(z_1)\}_{z_1}+\int \{H_2(z_1,z_2),f_2(z_1,z_2)\}_{z_1}dz_2  +3 \int \{H_3(z_1,z_2,z_3),f_3(z_1,z_2,z_3)\}_{z_1}dz_2dz_3.
\end{split}
\end{equation}
In the second equality, we have employed the energy function $H$ in \eqref{en-func}. In the third equality, we have used the identity \eqref{eqn} several times. We have substituted the definitions of the densty functions $f_1$ and $f_2$ in the last equality. Here, the notation $\{\bullet ,\bullet\}_{z_1}$ is the Poisson bracket only for $z_1$ variable even if the functions inside the bracket  are depending on some other variables. 

In a similar fashion, the dynamics of the moment function $f_2$,  can be computed as
\begin{equation} \label{Dyn-f-2}
\begin{split}
	\frac{\partial f_2}{\partial t}&=6 \int \{H(z_1,z_2,z_3), f_3(z_1,z_2,z_3)\} dz_3
\\
	&=6 \int \Big \{ \sum_i H_1(z_i) + \sum_{i<j} H_2(z_i,z_j)+ H_3(z_1,z_2,z_3),f_3(z_1,z_2,z_3) \Big  \} dz_3 
	\\
	&= 6 \int \{H_1(z_1)+H_1(z_2),f_3(z_1,z_2,z_3)\}_{z_1,z_2}dz_3+6\int  \{H_2(z_1,z_2),f_3(z_1,z_2,z_3)\}_{z_1,z_2} dz_3  
\\
&\qquad + 
6\int  \{H_2(z_1,z_3)+H_2(z_2,z_3),f_3(z_1,z_2,z_3)\}_{z_1,z_2} dz_3 
+
6\int  \{H_3(z_1,z_2,z_3),f_3(z_1,z_2,z_3)\}_{z_1,z_2} dz_3
\\
&= \{H_1(z_1)+H_1(z_2),6 \int f_3(z_1,z_2,z_3) dz_3\}_{z_1,z_2}+ \{H_2(z_1,z_2),6\int f_3(z_1,z_2,z_3)dz_3 \}_{z_1,z_2} 
\\
&\qquad + 
6\int  \{H_2(z_1,z_3)+H_2(z_2,z_3),f_3(z_1,z_2,z_3)\}_{z_1,z_2} dz_3 
+
6\int  \{H_3(z_1,z_2,z_3),f_3(z_1,z_2,z_3)\}_{z_1,z_2} dz_3 
\\
&= \{H_1(z_1)+H_1(z_2), f_2(z_1,z_2) \}_{z_1,z_2}+ \{H_2(z_1,z_2),f_2(z_1,z_2) \}_{z_1,z_2} 
\\
&\qquad + 
6\int  \{H_2(z_1,z_3)+H_2(z_2,z_3),f_3(z_1,z_2,z_3)\}_{z_1,z_2} dz_3 
+
6\int  \{H_3(z_1,z_2,z_3),f_3(z_1,z_2,z_3)\}_{z_1,z_2} dz_3.
\end{split}
\end{equation}
Here, $\{\bullet ,\bullet\}_{z_1,z_2}$ 
is the Poisson bracket only for $z_1$ and $z_2$ variables even if the functions inside the bracket are depending on some other variables.

\subsection{Lie-Poisson Realization of BBGKY hierarchy}\label{Alg-BBGKY}

We denote $A_1$ as space of smooth functions on $P$. Equipped with the canonical Poisson bracket, $A_1$ is a Lie algebra. Similarly, we define $A_2$ and $A_3$ as Lie algebras of the symmetric functions on $P^2$ and $P^3$, respectively. 
There are hierarchical embeddings defined to be 
\begin{equation}\label{embedding-BBGKY}
\begin{split}
& A_1 \longrightarrow A_2, \qquad K_1(z_1)\mapsto K_1^{(2)}(z_1,z_2):= K_1(z_1)+ K_1(z_2)\\
& A_1 \longrightarrow A_3, \qquad K_1(z_1)\mapsto  K_1^{(3)}(z_1,z_2,z_3):= K_1(z_1)+ K_1(z_2)+ K_1(z_3)
\\
& A_2 \longrightarrow A_3, \qquad K_2(z_1,z_2)\mapsto K_2^{(3)}(z_1,z_2,z_3):=K_2(z_1,z_2)+K_2(z_2,z_1)
+
K_2(z_1,z_3)\\ & \hspace{8cm}+K_2(z_3,z_1)+K_2(z_2,z_3)+K_2(z_3,z_2).
\end{split}
\end{equation}
We define the following direct sum 
\begin{equation}\label{A-3}
\C{A}:=A_3\oplus A_2 \oplus A_1, 
\end{equation}
then, introduce the following mapping from the product space $\C{A}$ to the Lie algebra $A_3$ of symmetric functions on $P^3$ that is
\begin{equation} \label{alpha-1}
\begin{split}
\alpha: \C{A} \longrightarrow A_3 ,\qquad 
(K_3,K_2,K_1)  \longrightarrow K_3(z_1,z_2,z_3)+K_2^{(3)}(z_1,z_2,z_3)+K_1^{(3)}(z_1,z_2,z_3) 
\end{split}
\end{equation}
Notice that, $\alpha$ in \eqref{alpha-1} turns out to be a Lie algebra homomorphism if, on the domain space $\C{A} $, the following Lie bracket is assumed
\begin{equation}\label{hiearchybracket}
\begin{split} 
\big [(K_3,K_2,K_1),(L_3,L_2,L_1)\big]_{\C{A}}=\Big(\{K_3,L_3\}+\{K_3,L_2^{(3)}\}+\{K_3,L_1^{(3)}\}+\{K_2^{(3)},L_3\}+\{K_1^{(3)},L_3\}\\ +\{K_2^{(3)},L_2^{(3)}\}, 
\{K_2,L^{(2)}_1\}_{z_1,z_2}+\{K^{(2)}_1,L_2\}_{z_1,z_2},\{K_1,L_1\}_{z_1}\Big),
\end{split}
\end{equation}
where, on the right hand side, the notation $\{\bullet,\bullet\}$ without a subscript is the Poisson bracket on $P_3$, $\{\bullet,\bullet\}_{z_1,z_2}$ is the Poisson bracket only for $z_1$ and $z_2$ variables, $\{\bullet,\bullet\}_{z_1}$ is the Poisson bracket only for $z_1$ variable. 

The dual mapping of $\alpha$ in \eqref{alpha-1} computed to be 
\begin{equation}
\alpha^*: A_3^*\longrightarrow  \C{A}^*, \qquad  f_3\mapsto \Big(f_3,f_2=\int 6f_3(z_1,z_2,z_3)dz_3, f_1=\int 3 f_3(z_1,z_2,z_3)dz_2dz_3\Big),
\end{equation}
becomes a momentum and a Poisson map. Notice that, the mapping $\alpha^*$ determines the moment functions exhibited in \eqref{moments}. 
  
\textbf{Coadjoint action.} Assume that the adjoint action of the $\C{A}$ on itself is the Lie bracket $[\bullet,\bullet]_{\C{A}}$ in \eqref{hiearchybracket}. The coadjoint action of the space $\C{A}$ on its dual $\C{A}^*$ is
\begin{equation}\label{Cala}
 \begin{split}
&\Big\langle ad^*_{(L_3,L_2,L_1)}(f_3,f_2,f_1),(K_3,K_2,K_1)\Big\rangle = - \Big\langle (f_3,f_2,f_1),ad_{(L_3,L_2,L_1)}(K_3,K_2,K_1)\Big\rangle
\\
& \hspace{1 cm}= \Big\langle (f_3,f_2,f_1),\big[(K_3,K_2,K_1),(L_3,L_2,L_1)\big]_{\C{A}}\Big\rangle
\\
&\hspace{1 cm}
=\Big\langle 
f_3, \{K_3,L_3\}+\{K_3,L_2^{(3)}\}+\{K_3,L_1^{(3)}\}+\{K_2^{(3)},L_3\}+\{K_1^{(3)},L_3\}+\{K_2^{(3)},L_2^{(3)}\}\Big\rangle 
\\
 &\hspace{3 cm}+ \big\langle f_2,
\{K_2,L^{(2)}_1\}_{z_1,z_2}+\{K^{(2)}_1,L_2\}_{z_1,z_2} \big\rangle +  \big\langle  f_1, \{K_1,L_1\}_{z_1}\big\rangle.
 \end{split}
 \end{equation}
In the second line, we have substituted the Lie algebra bracket \eqref{hiearchybracket}. In the last equality, the first pairing is the one  available between $A_3^*$ and $A_3$ with the symplectic volume $dz_1dz_2dz_3$ whereas the second one is between $A_2^*$ and $A_2$ with the symplectic volume $dz_1dz_2$, and finally the last one is between $A_1^*$ and $A_1$ with the symplectic volume $dz_1$.
It is evident that, in order to arrive at the explicit expression of the coadjoint action from \eqref{Cala}, we need to single out the functions $K_3$, $K_2$ and $K_1$. For this, we first recall the following association property of the smooth functions 
\begin{equation} \label{eqn-2}
\int h(z) \{k(z),l(z)\}_z ~ dz=\int k(z) \{l(z),h(z)\}_z ~ dz. 
\end{equation}
To see this identity, we simply consider the Leibniz identity 
\begin{equation}\label{leib}
 \{h(z)k(z),l(z)\}_z=h(z)\{k(z),l(z)\}_z+ k(z)\{h(z),l(z)\}_z,
\end{equation}
then, take the integral of this expression. In this case, the left hand side turns out to be zero due to \eqref{eqn}. The integrals on the right hand side of \eqref{leib}  is exactly \eqref{eqn-2} after a reordering.

Let us apply the identity \eqref{eqn-2} to the first pairing on the last equality in \eqref{Cala}. We have that 
\begin{equation} \label{La-1}
\begin{split}
&
\int
f_3(z_1,z_2,z_3) \Big( \{K_3,L_3\}+\{K_3,L_2^{(3)}\}+\{K_3,L_1^{(3)}\}+\{K_2^{(3)},L_3\}+\{K_1^{(3)},L_3\}
\\
 &\hspace{2 cm}
 +\{K_2^{(3)},L_2^{(3)}\}\Big )(z_1,z_2,z_3)  dz_1dz_2dz_3  
\\
&\hspace{1 cm} =\int \Big( K_3 \{L_3,f_3\} (z_1,z_2,z_3)  
+ K_3 \{L_2^{(3)},f_3\} (z_1,z_2,z_3)   + K_3 \{L_1^{(3)},f_3\}
(z_1,z_2,z_3)\Big) dz_1dz_2dz_3
\\
 &\hspace{2 cm} +\int K_2(z_1,z_2)\Big(6\int\{L_3,f_3\}_{z_1,z_2}dz_3 \Big) dz_1dz_2 
+ \int K_1(z_1)\Big(3\int\{L_3,f_3\}_{z_1}dz_2dz_3 \Big) dz_1 
\\
 &\hspace{2 cm} +2\int K_2(z_1,z_2) \{L_2,f_2\}_{z_1,z_2} dz_1dz_2 \\
 &\hspace{2 cm} 
+ \int K_2(z_1,z_2)\Big(12\int\{L_2(z_1,z_3)+L_2(z_2,z_3),f_3\}_{z_1,z_2}dz_3 \Big) dz_1dz_2,
\end{split}
\end{equation} 
 where we have employed  the identities \eqref{eqn}  and \eqref{eqn-2} and the definitions of moments in \eqref{moments}. In a similar way, we compute the pairings on the last line of \eqref{Cala} as follows 
\begin{equation}\label{La-2}
\begin{split}
&
\int
f_2(z_1,z_2) \Big(\{K_2,L^{(2)}_1\}_{z_1,z_2} +\{K^{(2)}_1,L_2\}_{z_1,z_2} \big)(z_1,z_2) dz_1dz_2 
+  \int   f_1(z_1) \{K_1,L_1\}_{z_1}(z_1) dz_1
\\
&\hspace{1 cm} = \int K_2 (z_1,z_2) \{L^{(2)}_1 ,f_2\}  _{z_1,z_2} (z_1,z_2)dz_1dz_2 
+ \int  K_1 (z_1)
 \Big(
\int \{L_2(z_1,z_2),f_2(z_1,z_2)\}_{z_1} dz_2 
 \Big) dz_1 \\
 &\hspace{2 cm} + \int K_1(z_1) \{L_1,f_1 \}_{z_1}(z_1) dz_1.
\end{split}
\end{equation} 
In \eqref{La-1}  and \eqref{La-2}, we collect the terms involving  $K_3$, $K_2$ and $K_1$ in an order and then, arrive at the coadjoint flow 
\begin{equation}
ad^*_{(L_3,L_2,L_1)}(f_3,f_2,f_1)= (\tilde{f}_3,\tilde{f}_2,\tilde{f}_1)
\end{equation}
where 
\begin{equation}
\begin{split}
\tilde{f}_3(z_1,z_2,z_3)&= 
\{L_3+L_2^{(3)}+L_1^{(3)},f_3\} (z_1,z_2,z_3), 
\\
\tilde{f}_2(z_1,z_2)&= 
\{L_1^{(2)}+2L_2, f_2\}_{z_1,z_2}(z_1,z_2)
 + 
12\int  \{L_2(z_1,z_3)+L_2(z_2,z_3),f_3(z_1,z_2,z_3)\}_{z_1,z_2} dz_3 
\\
&\qquad+
6\int  \{L_3,f_3\}_{z_1,z_2} (z_1,z_2,z_3) dz_3
\\
\tilde{f}_1(z_1)&= \{L_1,f_1\}_{z_1}(z_1)+2\int \{L_2(z_1,z_2),f_2(z_1,z_2)\}_{z_1}dz_2+3 \int \{L_3(z_1,z_2,z_3),f_3(z_1,z_2,z_3)\}_{z_1}dz_2dz_3
\end{split}
\end{equation}

\textbf{Lie-Poisson equation.}
We assume now the following functional
\begin{equation}
\mathcal{H}(f_3,f_2,f_1)=\int H_3(z_1,z_1,z_3)f_3(z_1,z_1,z_3)
dz_1dz_2dz_3+ \int \frac{1}{2}
H_2(z_1,z_2) f_2(z_1,z_1)
dz_1dz_2  + \int H_1(z_1)f_2(z_1)dz_1. 
\end{equation}
Then we have that 
\begin{equation}
\frac{\delta \mathcal{H}}{\delta (f_3,f_2,f_1)}
=\Big(\frac{\delta \mathcal{H}}{\delta f_3},
\frac{\delta \mathcal{H}}{\delta f_2},
\frac{\delta \mathcal{H}}{\delta f_1}
\Big)
=\big(H_3(z_1,z_1,z_3),
\frac{1}{2}H_2(z_1,z_2),
H_1(z_1)
\big)\in \C{A}
\end{equation}
which are the energy function in \eqref{en-func} under the isomorphism $\alpha$ in \eqref{alpha-1}. If the three tuple  $H=(H_3,(1/2)H_2,H_1)$ is substituted in to the coadjoint action then we arrive at 
\begin{equation}\label{LP-BBGKY}
\frac{\partial}{\partial t}  
(f_3,f_2,f_1)=
ad^*_{{\delta \mathcal{H}}/{\delta (f_3,f_2,f_1)}}(f_3,f_2,f_1)
=
ad^*_{(H_3,(1/2)H_2,H_1)}(f_3,f_2,f_1).
\end{equation}
A direct calculation proves that the coadjoint flow \eqref{LP-BBGKY}
is precisely the system of equations  \eqref{Dyn-f-1}, \eqref{Dyn-f-2} and \eqref{Vlasov} governing the dynamics of the moments. 

 \subsection{BBGKY Hierarchy as a Matched Pair}\label{D-Alg-BBGKY}
 
Recall the direct sum $\C{A}=A_3\oplus A_2 \oplus A_1$ in \eqref{A-3}. One evident  decomposition of $\C{A}$ is given by
\begin{equation} \label{a3}
	\C{A}=\mathfrak{g}_{32}\oplus\mathfrak{h}_1, \qquad \qquad \mathfrak{g}_{32}= A_3\oplus A_2 \text{ and } \mathfrak{h}_1=A_1.
\end{equation}
In the present subsection, we examine this realization in the matched pair decomposition point of view. 

\textbf{Decomposition of the Lie algebra.}
It is a direct calculation two show that the Lie bracket (\ref{hiearchybracket}) is closed if it is restricted to the constitutive subspaces $\mathfrak{g}_{32}$ and $\mathfrak{h}_1$. This reads that both $\mathfrak{g}_{32}$ and $\mathfrak{h}_1$ are Lie subalgebras. So, as a manifestation of Universal Lemma  \ref{universal-prop}, this decomposition can be written as a matched pair decomposition introduced in Subsection \ref{doublecross}. We first determine the Lie algebra brackets on $\mathfrak{g}_{32}$ and $\mathfrak{h}_1$ by restricting the bracket \eqref{hiearchybracket} to the subspaces $\mathfrak{g}_{32}$ and $\mathfrak{h}_1$
 \begin{equation} \label{Alg-32-1}
 \begin{split}
 [\bullet,\bullet]_{32}&: \mathfrak{g}_{32}\otimes \mathfrak{g}_{32}\longrightarrow \mathfrak{g}_{32}, \qquad [(K_3,K_2),(L_3,L_2)]_{32}=
 \Big(\{K_3+K_2^{(3)},L_3\}+\{K_3+K_2^{(3)},L_2^{(3)}\}, 0\Big),
 \\
 [\bullet,\bullet]_{1}&:
 \mathfrak{h}_1\otimes \mathfrak{h}_1\longrightarrow \mathfrak{h}_1,  \qquad [K_1,L_1]_1=\{K_1,L_1\}_{z_1},
 \end{split}
 \end{equation}
 respectively. 
In order to compute mutual actions, we recall the identity \eqref{bracketactions}. In the present case, we compute
\begin{equation}
\big(K_1 \vartriangleright (L_3,L_2)\big) \oplus \big( K_1 \vartriangleleft (L_3,L_2) \big):=[(0,0,K_1),(L_3,L_2,0)]_{\C{A}}
= \big( \{K_1^{(3)},L_3\},\{K_1^{(2)},L_2\}_{z_1,z_2} \big)
 \oplus
0
\end{equation}
then conclude that the left action of $\mathfrak{h}_1$ on $\mathfrak{g}_{32}$, and the right action of $\mathfrak{g}_{32}$ on $\mathfrak{h}_1$ as
\begin{equation} \label{Act-32-1}
\begin{split}
\vartriangleright &:\mathfrak{h}_1 \otimes  \mathfrak{g}_{32}\longrightarrow \mathfrak{g}_{32},\qquad K_1 \vartriangleright  (L_3,L_2)= \big( \{K_1^{(3)},L_3\},\{K_1^{(2)},L_2\}_{z_1,z_2} \big),
\\
\vartriangleleft &: \mathfrak{h}_1 \otimes  \mathfrak{g}_{32}\longrightarrow \mathfrak{h}_1, 
\qquad K_1 \vartriangleleft  (L_3,L_2)=0,
	\end{split}
\end{equation}
respectively. Notice that, the right action $\vartriangleleft$ is trivial so that the Lie algebra $\C{A}=\mathfrak{g}_{32}\rtimes \mathfrak{h}_1$ is a semidirect product Lie algebra. The Lie bracket $[\bullet,\bullet]_{\C{A}}$ in \eqref{hiearchybracket} admits the following decomposition
\begin{equation}
[(K_3,K_2)\oplus K_1,(L_3,L_2)\oplus L_1]=\big([(K_3,K_2),(L_3,L_2)]_{32}+ K_1 \vartriangleright (L_3,L_2)-L_1 \vartriangleright (K_3,K_2) \big )\oplus[K_1,L_1]_{1}. 
\end{equation}
Here, the subalgebras $[\bullet,\bullet]_{32}$ and 
 $[\bullet,\bullet]_{1}$  are the ones in \eqref{Alg-32-1}, and the left action is in \eqref{Act-32-1}. This  realization is precisely in the matched pair Lie bracket form \eqref{mpla} where the left action $\vartriangleleft$ is trivial. Let us apply this to the Lie-Poisson formulation of BBGKY Dynamics \eqref{LiePoissonongh}.

\textbf{Decomposition of BBGKY dynamics.} The dual spaces of the constitutive Lie subalgebras $\mathfrak{g}_{32}$ and $\mathfrak{h}_1$ are $\mathfrak{g}_{32}^*=A_3^*\oplus A_2^* $ and $\mathfrak{h}_1^*=A_1^*$, respectively. So that, we can write $\C{A}^*=\mathfrak{g}_{32}^*\oplus \mathfrak{h}_1^*$. The coadjoint action of $\mathfrak{g}_{32}$ on $\mathfrak{g}_{32}^*$, and the coadjoint action of $\mathfrak{h}_{1}$ on $\mathfrak{h}_{1}^*$ are
\begin{equation}
\begin{split}
&ad_{(L_3,L_2)}(f_3,f_2)=\Big(\{L_3,f_3\} + 6\int \{L_2(z_1,z_2),f_3(z_1,z_2,z_3)\}_{z_1,z_2}dz_3, 2\{L_2(z_1,z_2) ,f_2(z_1,z_2) \}_{z_1,z_2}  \\ &
\hspace{2cm} +
6\int \{L_3(z_1,z_2,z_3),f_3(z_1,z_2,z_3)\}_{z_1,z_2}dz_3 + 12\int\{L_2(z_1,z_3)+L_2(z_2,z_3),f_3\}_{z_1,z_2}dz_3\Big),
\\
&ad_{K_1}f_1=\{K_1(z_1) ,f_1(z_1) \}_{z_1},
\end{split}
\end{equation}
respectively. Recall the mutual actions of 
$\mathfrak{g}_{32}$ and $\mathfrak{h}_{1}$ on each other given in  \eqref{Act-32-1}. The dual of these actions are computed to be 
\begin{equation}
\begin{split}
\overset{\ast }{\vartriangleleft} &:\mathfrak{g}_{32}^*\otimes \mathfrak{h}_{1} \longrightarrow \mathfrak{g}_{32}^*, \qquad (f_3,f_2)\overset{\ast }{\vartriangleleft}K_1 = \big( 3\{f_3(z_1,z_2,z_3),K_1(z_1)\} _{z_1},2\{f_2(z_1,z_2),K_1(z_1)\} _{z_1} \big)
\\
\overset{\ast }{\vartriangleright} &:\mathfrak{g}_{32} \otimes \mathfrak{h}_{1}^* \longrightarrow \mathfrak{h}_{1}^* ,\qquad (L_3,L_2)\overset{\ast }{\vartriangleright}  f_1 =0.
\end{split}
\end{equation}
The mapping $\mathfrak{b}$ in \eqref{b} and its dual \eqref{b*} are computed to be  
\begin{equation}
\begin{split}
\mathfrak{b}_{(L_3,L_2)}&: \mathfrak{h}_{1} \longrightarrow \mathfrak{g}_{32} , \qquad 
\mathfrak{b}_{(L_3,L_2)}(K_1):=K_1 
\vartriangleright
(L_3,L_2)
\\
\mathfrak{b}_{(L_3,L_2)}^*&: \mathfrak{g}_{32} ^* \longrightarrow  \mathfrak{h}_{1}^*, \qquad   \mathfrak{b}_{(L_3,L_2)}^*(f_3,f_2)=3\int\{L_3,f_3\}_{z_1}dz_2dz_3 + 2 
\int\{L_2,f_2\}_{z_1}dz_2.
\end{split}
\end{equation}
Since the left action is trivial both the mapping $\mathfrak{a}$ in \eqref{a} and its dual \eqref{a*} are trivial. It is now straight forward to modify the matched Lie-Poisson equation \eqref{LPEgh} to the present version and  determine 
the coadjoint flow as
\begin{equation}\label{LPEgh-BBGKY-MP}
\begin{split}
  \frac{d (f_3,f_2)}{dt} & = 
	ad^{\ast}_{(L_3,L_2)}(f_3,f_2)
-
(f_3,f_2)\overset{\ast }{\vartriangleleft}K_1
\\
\frac{d f_1}{dt} &=
	ad^{\ast}_{K_1} f_1  
+\mathfrak{b}_{(L_3,L_2)}^*(f_3,f_2).
\end{split}
\end{equation}
This equations is precisely the coadjoint flow \eqref{LP-BBGKY} realization of BBGKY dynamics. Evidently takes the classical form if $(L_3,L_2,K_1)=(H_3,(1/2)H_2,H_1)$. 

\subsection{BBGKY Hierarchy as an Extended Structure}\label{LPe-BBGKY}

Recall once more the direct sum $\C{A}=A_3\oplus A_2 \oplus A_1$ in \eqref{A-3}. We have examined a matched pair decomposition \eqref{a3} of this sum. An alternative decomposition of $\C{A}$ can be given by
\begin{equation} \label{a32}
	\C{A}=\mathfrak{g}_{3}\oplus\mathfrak{h}_{21}, \qquad \qquad 
	\mathfrak{g}_{3}:= A_3,\text{ and }  \mathfrak{h}_{21}= A_2\oplus A_1.
\end{equation}
\textbf{Decomposition of the Lie algebra.} It is straightforward to see that, $\mathfrak{g}_{3}$ is a subalgebra of $\C{A}$ with induced bracket
\begin{equation} \label{Alg-3}
 [\bullet,\bullet]_{3}: \mathfrak{g}_{3}\otimes \mathfrak{g}_{3}\longrightarrow \mathfrak{g}_{3}, \qquad [K_3,L_3]_{3}=
\{K_3,L_3\},
 \end{equation}
 where $\{\bullet,\bullet\}$ is the Poisson bracket on $P^3$. 
On the other hand, the subspace $\mathfrak{h}_{21}$ fails to be so. Indeed, the Lie bracket \eqref{hiearchybracket} of two generic elements $(0,K_2,K_1)$ and $(0,L_2,L_1)$ in $\mathfrak{h}_2$ is
\begin{equation}
	[(0,K_2,K_1),(0,L_2,L_1)]_{\C{A}}=\{K_2^{(3)},L_2^{(3)}\} \oplus \big(\{K_2,L^{(2)}_1\}_{z_1,z_2}+\{K^{(2)}_1,L_2\}_{z_1,z_2},\{K_1,L_1\}_{z_1}\big)\in \mathfrak{g}_{3}\oplus\mathfrak{h}_{21},
\end{equation}
where the first term on the right hand side falls into $\mathfrak{g}_{3}$ whereas the the second and the third terms are in $\mathfrak{h}_{21}$. 
So that, this decomposition should be analysed in  the light of extended structures presented in Subsection \ref{brzezinski}. 
Accordingly, by referring to \eqref{Phi-kappa}, we define
\begin{equation}\label{k-p-Bre}
\begin{split}
	\Phi &:\mathfrak{h}_{21} \otimes \mathfrak{h}_{21}  \longrightarrow \mathfrak{g}_{3}, \qquad   ((K_2,K_1),(L_2,L_1))\mapsto \{K_2^{(3)},L_2^{(3)}\} \\
	\kappa &: \mathfrak{h}_{21} \otimes \mathfrak{h}_{21}
	\longrightarrow \mathfrak{h}_{21}, \qquad 
	((K_2,K_1),(L_2,L_1))\mapsto \big(\{K_2,L^{(2)}_1\}_{z_1,z_2}+\{K^{(2)}_1,L_2\}_{z_1,z_2},\{K_1,L_1\}_{z_1}\big).
		\end{split}
\end{equation}
Now, we are ready to compute mutual \textit{actions} defined in \eqref{Lieact} and \eqref{Lieact1} between the constitutive spaces $\mathfrak{h}_{21}$ and $\mathfrak{g}_3$. To obtain the fomulas, we employ the identity \eqref{bracketactions} that is
\begin{equation}
\begin{split}
(K_2,K_1) \vartriangleright L_3 \oplus (K_2,K_1)  \vartriangleleft L_3 &= 
[(0,K_2,K_1),(L_3,0,0)]_{\C{A}}  \\
&=\big(\{K^{(3)}_2,L_3\}+\{K^{(3)}_1,L_3\}\big) \oplus (0,0) 
\in \mathfrak{g}_3 \oplus \mathfrak{h}_{21}
\end{split}
\end{equation} 
which gives that 
\begin{equation}\label{L-Act-Bre-BBGKY}
\begin{split}
\vartriangleright&:\mathfrak{h}_{21} \otimes \mathfrak{g}_3 \longrightarrow \mathfrak{g}_3, \qquad  (K_2,K_1) \vartriangleright L_3=\{K^{(3)}_2,L_3\}+\{K^{(3)}_1,L_3\} \\
\vartriangleleft&: \mathfrak{h}_{21} \otimes \mathfrak{g}_3 \longrightarrow \mathfrak{h}_{21}, \qquad (K_2,K_1)  \vartriangleleft L_3=(0,0).
\end{split}
\end{equation}
Due to Universal Lemma \ref{uni-prop-bre}, the decomposition $\mathfrak{g}_{3}\oplus\mathfrak{h}_{21}$ of the Lie algebra $\C{A}$ reads the decomposition of the Lie bracket $\{\bullet,\bullet\}_{\C{A}}$ given in \eqref{hiearchybracket} into the form of extended  Lie bracket \ref{brzezinski} where the right action $\vartriangleleft$ is trivial 
\begin{equation} \label{mpla-2-cocyc-1-BBGKY}
\begin{split}
&[K_3 \oplus(K_2,K_1) ,L_3 \oplus(L_2,L_1)]_{_\Phi\bowtie}=\big( \{K_3,L_3\}+(K_2,K_1) \vartriangleright L_3-(L_2,L_1)\vartriangleright K_3 \\ &\hspace{6cm}+ \Phi((K_2,K_1) ,(L_2,L_1) )
\big)\oplus \kappa((K_2,K_1) ,(L_2,L_1) ),
\end{split}
\end{equation}
where the left action is the one given in \eqref{L-Act-Bre-BBGKY} whereas $\Phi$ and $\kappa$ mapping are those available in \eqref{k-p-Bre}. 

\textbf{An alternative decomposition of the Lie algebra.} 
The hierarchy of the moment functions suggests an alternative formulation of $\Phi$ and $\kappa$ mappings in \eqref{k-p-Bre}. This is due to the fact that the term $[K_2^{(3)},L_2^{(3)}]$ can be written as a sum of  some terms in $\mathfrak{h}_{21}$ and some terms in $\mathfrak{g}_{3}$. Indeed,
\begin{equation}
\begin{split}
[K_2^{(3)},L_2^{(3)}]&=2(\{K_2(z_1,z_2),L_2(z_1,z_2)\}_{z_1,z_2})^{(3)}(z_1,z_2,z_3) \\ &\qquad +4\{K_2(z_1,z_2),L_2(z_1,z_3)+L_2(z_2,z_3) \}_{z_1,z_2}  + 4\{K_2(z_1,z_3),L_2(z_1,z_2)+L_2(z_2,z_3) \}_{z_1,z_3} \\&\qquad + 4\{K_2(z_2,z_3),L_2(z_1,z_2)+L_2(z_1,z_3)\}_{z_2,z_3}.
		\end{split}
\end{equation}
Here, as depicted in the display, the first term on the right hand side can be written as the image of the symmetric function $\{K_2(z_1,z_2),L_2(z_1,z_2)\}$ under the embedding $A_2 \mapsto A_3$ given in \eqref{embedding-BBGKY}. Accordingly, instead of $\Phi$ and $\kappa$ mappings in \eqref{k-p-Bre}, we can propose the following alternatives 
\begin{equation}\label{k-p-Bre-alt}
\begin{split}
	\tilde{\Phi} &:\mathfrak{h}_{21} \otimes \mathfrak{h}_{21}  \longrightarrow \mathfrak{g}_{3}, \\ & ((K_2,K_1),(L_2,L_1))\mapsto 4\{K_2(z_1,z_2),L_2(z_1,z_3)+L_2(z_2,z_3) \}  + 4\{K_2(z_1,z_3),L_2(z_1,z_2)+L_2(z_2,z_3) \} \\&\hspace{7cm}+ 4\{K_2(z_2,z_3),L_2(z_1,z_2)+L_2(z_1,z_3) \}  \\
	\tilde{\kappa} &: \mathfrak{h}_{21} \otimes \mathfrak{h}_{21}
	\longrightarrow \mathfrak{h}_{21}, \\ & 
	((K_2,K_1),(L_2,L_1))\mapsto \Big(\{K_2(z_1,z_2),L_1(z_1)+L_1(z_2)\}_{z_1,z_2}+\{K_1(z_1)+K_1(z_2),L_2(z_1,z_2)\}_{z_1,z_2}\\&\hspace{7cm}+2\{K_2(z_1,z_2),L_2(z_1,z_2)\}_{z_1,z_2}
	,\{K_1(z_1),L_1(z_1)\}_{z_1}\Big).
		\end{split}
\end{equation}
Evidently, this observation reads an alternative Lie bracket operation on $\C{A}$ as well. We denote this by a tilde notation $[ \bullet  , \bullet ]_{\tilde{\Phi}\bowtie}$ and record as follows
\begin{equation} \label{mpla-2-cocyc-1-BBGKY-2}
\begin{split}
&[ K_3 \oplus(K_2,K_1)  , L_3 \oplus(L_2,L_1) ]_{\tilde{\Phi}\bowtie}=\big( \{K_3,L_3\}+(K_2,K_1) \vartriangleright L_3-(L_2,L_1)\vartriangleright K_3 \\ &\hspace{6cm}+ \tilde{\Phi}((K_2,K_1) ,(L_2,L_1) )
\big)\oplus \tilde{\kappa}((K_2,K_1) ,(L_2,L_1) ),
\end{split}
\end{equation}
where $\{K_3,L_3\}$ is the Poisson bracket on $P^3$, and $\vartriangleright$ is the left action in \eqref{L-Act-Bre-BBGKY}. 

\textbf{Decomposition of the dynamics:   $\C{A}^*=\mathfrak{g}_{3}^*\oplus\mathfrak{h}_{21}^*$.} We start with the dualization of the mutual actions in \eqref{L-Act-Bre-BBGKY}. Notice that, the left action is trivial, so that it induces a trivial dual action. For the right action, we compute the dual action as
\begin{equation}
\mathfrak{g}_{3}^*\overset{\ast }{\vartriangleleft}\mathfrak{h}_{21}\longrightarrow 
\mathfrak{g}_{3}^*, \qquad  f_3
\overset{\ast }{\vartriangleleft} (K_2,K_1)=\{f_3,K_2^{(3)}\}+
\{f_3,K_1^{(3)}\}.
\end{equation}
Using the right action in \eqref{L-Act-Bre-BBGKY}, and in the lights of the definitions in \eqref{b} and\eqref{b*}, we compute the following mapping along with its dual 
\begin{equation}
\begin{split}
\mathfrak{b}_{L_3}&:\mathfrak{h}_{21}\longrightarrow \mathfrak{g}_{3},\qquad \mathfrak{b}_{L_3}(K_2,K_1)=(K_2,K_1) \vartriangleright L_3=[K^{(3)}_2,L_3]+[K^{(3)}_1,L_3] 
\\
\mathfrak{b}^*_{L_3}&:\mathfrak{g}_{3}^* \longrightarrow \mathfrak{h}_{21}^* ,\qquad \mathfrak{b}^*_{L_3}f_3=\big(6\int \{L_3,f_3 \}_{z_1,z_2} dz_3, 3\int \{L_3,f_3 \}_{z_1} dz_2dz_3 \big).
\end{split}
\end{equation}
Further, according to \eqref{kappan} and \eqref{phin}, by freezing the first entries of $\tilde{\Phi}$ and $\tilde{\kappa}$ in \eqref{k-p-Bre-alt}, we arrive at linear mappings. One od these mappings is $\tilde{\Phi}_{(K_2,K_1)}$ from $\mathfrak{h}_{21}$ to $\mathfrak{g}_{3}$, and other is $\tilde{\kappa}_{(K_2,K_1)}$  from $\mathfrak{h}_{21}$ to $\mathfrak{h}_{21}$. Dualizations of these mappings result with 
\begin{equation}
\begin{split}
&\tilde{\Phi}_{(K_2,K_1)}^*: \mathfrak{g}_{3}^* \longrightarrow 
\mathfrak{h}_{21}^*,\quad \tilde{\Phi}_{(K_2,K_1)}^*f_3=\Big(24
\int \{K_2(z_1,z_3),f_3(z_1,z_2,z_3)\}_{z_1} dz_3 ,0\Big),
\\
&
\tilde{\kappa}_{(K_2,K_1)}^*:\mathfrak{h}_{21}^*
\longrightarrow 
\mathfrak{h}_{21}^*, \quad \tilde{\kappa}_{(K_2,K_1)}^*(f_2,f_1)=\Big( 
2\{K_1(z_1),f_2(z_1,z_2)\}_{z_1}
+2\{K_2(z_1,z_2),f_2(z_1,z_2)\}_{z_1,z_2}, \\&\hspace{7cm}
2\int 
\{K_2(z_1,z_2),f_2(z_1,z_2)\}_{z_1}dz_2+\{K_1(z_1),f_1(z_1)\}_{z_1}
\Big).
\end{split}
\end{equation}
Now, recall the decomposed Lie-Poisson equation \eqref{LPEghcocycle}. Since, in the present case, the right action is trivial, we take the terms involving $\overset{\ast }{\vartriangleleft}$ and $\mathfrak{a}^*$ as zero.   
After substituting all these into the Lie-Poisson equation, we have that
\begin{equation}\label{LPEghcocycle-BBGKY}
\begin{split}
&  \frac{df _3}{dt} = 
	ad^{\ast}_{K_3}(f _3)
	 - f _3 \overset{\ast }{\vartriangleleft}
	(K_2,K_1),
	\\
&\frac{d (f_2,f_1)}{dt} =
\tilde{\kappa}^{\ast}_{(K_2,K_1)}(f_2,f_1)
+ \tilde{\Phi}_{(K_2,K_1)}^*f_3
+ \mathfrak{b}^*_{L_3}f_3. 
\end{split}
\end{equation}
Here, $ad^{\ast}_{K_3}(f _3)=\{K_3,f_3\}$ is the coadjoint action of $\mathfrak{g}_3$ on its dual space $\mathfrak{g}_3^*$. If we take $K_3=H_3$, $K_2=(1/2)H_2$ and $K_1=H_1$ then, the system is exactly the dynamics of the moments in \eqref{Dyn-f-1}, \eqref{Dyn-f-2} and \eqref{Vlasov} by decomposing the coadjoint flow \eqref{LP-BBGKY}.

\section{Coupling of 2-cocycles} \label{Sec-Cop-co}

In Subsection \ref{2coc-Sec}, 2-cocycle extensions are exhibited as particular instances of extended structures. In this section, we couple two 2-cocycle extensions under mutual actions. This will be achieved by matched pair theory available in Subsection \ref{doublecross}. Our goal is to explore conditions for a matched pair of 2-cocycles to be a 2-cocycle of a matched pair. We shall, further, study dynamics on the coupled system to have the Lie-Poisson equations for the collective 
motion. 

\subsection{Coupling of 2-cocycle Extensions} \label{centrallyextend}

We start with two Lie algebras, say $\mathfrak{l}$ and $\mathfrak{k}$, and two vector spaces $V$ and $W$. Assume a $V$-valued 2-cocycle $\varphi$ on $\mathfrak{l}$, and a $W$-valued 2-cocycle $\phi$ on $\mathfrak{k}$ given by
\begin{equation}\label{phi-psi}
\varphi: \mathfrak{l} \times \mathfrak{l} \rightarrow V, \qquad  \phi:\mathfrak{k} \times \mathfrak{k} \rightarrow W.
\end{equation}
Further, we consider a left action of the Lie algebra $\mathfrak{l}$ on the vector spaces $V$, and a left action of the Lie algebra $\mathfrak{k}$ on the vector spaces $W$ that is
\begin{equation} \label{actionsoflk}
\begin{split}
\downharpoonleft &:\mathfrak{l}\otimes V \rightarrow V,\qquad l \otimes v \mapsto l \downharpoonleft v, \\
\downharpoonright &:\mathfrak{k} \otimes W \rightarrow W,\qquad k \otimes w \mapsto k \downharpoonright w. \\
\end{split}
\end{equation} 
If these actions are compatible with the 2-cocycle maps as in \eqref{ext-act} then, one arrives at the following 2-cocycle extensions 
\begin{equation}\label{2-2-coc}
\mathfrak{g}:=V {_\varphi\rtimes} \:  \mathfrak{l},\qquad \mathfrak{h} :=W {_\phi\rtimes} \: \mathfrak{k}.
\end{equation}
In the light of the discussions done in Subsection \ref{2coc-Sec}, after some proper modifications of \eqref{AAAAA}, we have the following Lie algebra brackets on the extended Lie algebras $\mathfrak{g}$ and $\mathfrak{h}$
\begin{equation} \label{bracketontwococ}
\begin{split}
[v \oplus l, v' \oplus l']_{_\varphi\rtimes}&= \big(l \downharpoonleft v'-l' \downharpoonleft v+ \varphi(l,l')  \big) \oplus[l,l'], \\
[w \oplus k, w' \oplus k']_{_\phi\rtimes}&= \big(k \downharpoonright w'-k' \downharpoonright 
w+ \phi(k,k')  \big) \oplus[k,k'], \\ 
\end{split}
\end{equation}
 respectively. Here, the bracket $[l,l']$ is the Lie algebra bracket on $\mathfrak{l}$ whereas the bracket $[k,k']$ is the Lie  bracket on $\mathfrak{k}$.    

\textbf{Matching of 2-cocycles.}
We now examine matched pair of 2-cocycle extensions $\mathfrak{g}=V {_\varphi\rtimes} \:  \mathfrak{l}$ and $\mathfrak{h}=W {_\phi\rtimes} \: \mathfrak{k}$. To this end, we consider the followings $3$ sets of mappings: 

\textbf{(1)} We first consider mutual Lie algebras actions of $\mathfrak{l}$ and $\mathfrak{k}$ on each other 
\begin{equation}\label{Lieact-k-l}
\begin{split}
\blacktriangleright &:\mathfrak{k}\otimes \mathfrak{l}\rightarrow \mathfrak{l},\qquad k \otimes l \mapsto k \blacktriangleright l,
\\
\blacktriangleleft &:\mathfrak{k}\otimes \mathfrak{l}\rightarrow \mathfrak{k}
 ,\qquad k \otimes l \mapsto k \blacktriangleleft l.
 \end{split}
\end{equation}
We assume that these actions satisfy the compatibility condition \eqref{comp-mpa} hence, determine a matched pair Lie algebra denoted by  $\mathfrak{l} \bowtie \mathfrak{k}$. 

\textbf{(2)} In order to extend the mutual actions given in \eqref{Lieact-k-l} to the product spaces $\mathfrak{g}$ and $\mathfrak{h}$ in \eqref{2-2-coc}, we introduce   
a right action of $\mathfrak{l}$ on $W$ and, a left action $\mathfrak{k}$ on $V$ given by
\begin{equation}\label{actionsofkl}
\begin{split}
\curvearrowright &:\mathfrak{k}\otimes V \rightarrow V,\qquad k \otimes v \mapsto k \curvearrowright v, \\
\curvearrowleft&:W \otimes \mathfrak{l} \rightarrow W,\qquad w \otimes l \mapsto w \curvearrowleft l, \\
\end{split}
\end{equation}
respectively. 

\textbf{(3)} In addition, it is always possible to have the following cross representations
\begin{equation}\label{epsiot}
\begin{split}
\epsilon&: \mathfrak{k} \otimes \mathfrak{l} \rightarrow V, \qquad \epsilon(k,l) \in V,  \\
 \iota&: \mathfrak{k} \otimes \mathfrak{l} \rightarrow W, \qquad \iota(k,l) \in W.
\end{split}
\end{equation}

Referring to the mappings \eqref{Lieact-k-l}, \eqref{actionsofkl} and \eqref{epsiot}, we define mutual actions of 2-cocycle extensions $\mathfrak{g}=V {_\varphi\rtimes} \:  \mathfrak{l}$ and $\mathfrak{h} =W {_\phi\rtimes} \: \mathfrak{k}$ as follows
\begin{equation}\label{left-comp}
\begin{split}
\vartriangleright &: ~ (W {_\phi\rtimes} \: \mathfrak{k}) \times (V {_\varphi\rtimes} \:  \mathfrak{l}) \longrightarrow  V {_\varphi\rtimes} ~  \mathfrak{l}, \qquad ((w\oplus k),(v\oplus l))\mapsto \left( k \curvearrowright v+\epsilon(k,l)) \oplus (k \blacktriangleright l \right),
\\ 
\vartriangleleft &: ~ (W {_\phi\rtimes} \: \mathfrak{k}) \times (V {_\varphi\rtimes} \:  \mathfrak{l}) \longrightarrow  W {_\phi\rtimes} \: \mathfrak{k}, \qquad ((w \oplus k),(v \oplus l))\mapsto \left( w \curvearrowleft l+\iota(k,l))\oplus (k \blacktriangleleft l \right).
\end{split}
\end{equation}
It is possible to see that $\vartriangleright$ is a left action whereas $\vartriangleleft$ is a right action.
In order to construct a matched pair of $\mathfrak{h} =W {_\phi\rtimes} \: \mathfrak{k}$ and $\mathfrak{g}=V {_\varphi\rtimes} \:  \mathfrak{l}$,  one needs to justify the compatibility conditions in \eqref{comp-mpa}. A direct observation gives that, for the actions \eqref{left-comp}, the compatibility conditions \eqref{comp-mpa}  consist of $4$ equations. Two of them, those for the second terms in the decompositions, involve only the left $\blacktriangleright$ and the right $\blacktriangleleft$ actions in \eqref{Lieact-k-l}. These two equations are precisely the matched pair compatibility conditions for $\mathfrak{l} \bowtie \mathfrak{k}$. Since, we assume that $\mathfrak{l} \bowtie \mathfrak{k}$ is a matched pair, these two compatibility conditions are automatically satisfied. So, we left with other two compatibility conditions. For any $k,k'$ in $\G{k}$, $l,l'$ in $\G{l}$, $v,v'$ in $V$, and $w,w'$ in $W$, these equations are computed to be
\begin{equation}
\begin{split}
& k \curvearrowright(l \downharpoonleft v'-l' \downharpoonleft v+\varphi(l,l'))+\epsilon(k,\lbrack l, l' \rbrack)=l \downharpoonleft(k \curvearrowright v'+\epsilon(k,l'))-(k \blacktriangleright l')\downharpoonleft v+ \varphi(l,k \blacktriangleright l') \\
&\hspace{4cm} -l' \downharpoonleft(k \curvearrowright v+\epsilon(k,l))+(k\blacktriangleright l) \downharpoonleft v'-\varphi(l',k \blacktriangleright l)  +(k \blacktriangleleft l) \curvearrowright v'-(k \blacktriangleleft l')\curvearrowright v\\
&\hspace{4cm} +\epsilon(k \blacktriangleleft l,l')-\iota(k \blacktriangleleft l',l) -(w \curvearrowleft l'+\iota(k,l'))\curvearrowleft l 
\\
& (k \downharpoonright w'-k' \downharpoonright w+\phi(k,k'))\curvearrowleft l+\iota([k,k'],l)=k\downharpoonright(w'\curvearrowleft l+\iota(k',l))-(k' \blacktriangleleft l)\downharpoonright w+\phi(k,k' \blacktriangleleft l) \\
&\hspace{4cm}-k'\downharpoonright(w \curvearrowleft l+\iota(k,l))+(k \blacktriangleleft l) \downharpoonright w'-\phi(k',k\blacktriangleleft l)  -w \curvearrowleft(k' \blacktriangleright l)\\ &\hspace{4cm} +w'\curvearrowleft(k \blacktriangleright l)+\iota(k',k \blacktriangleright l)
-\epsilon(k,k'\blacktriangleright l) +k \curvearrowright (k' \curvearrowright v+\epsilon(k',l)),
\end{split}
\end{equation}
where $\downharpoonleft$ and $\downharpoonright$ are the left actions in \eqref{actionsoflk}, $\blacktriangleright$ and  $\blacktriangleleft$ are actions in \eqref{Lieact-k-l}, $\curvearrowright$ and $\curvearrowleft$ are the actions in \eqref{actionsofkl}, $\epsilon$ and $\iota$ are the mappings in \eqref{epsiot}. By assuming that these conditions are satisfied,  
we define the following matched pair Lie algebra 
  \begin{equation}\label{mp-2c}
  \mathfrak{g} \bowtie \mathfrak{h} 
=
(V {_\varphi\rtimes} \: \mathfrak{l})\bowtie (W {_\phi\rtimes} \: \mathfrak{k}).
  \end{equation}
To compute the matched Lie algebra bracket on this total space, we recall the general formula of the matched Lie bracket in \eqref{mpla} then, by employing the actions in \eqref{left-comp}, compute
\begin{equation}\label{bracketlk-1}
\big[\big ( (v\oplus l)\oplus(w\oplus k)\big),\big((v'\oplus l')
\oplus
(w'\oplus k')\big)\big]_{\bowtie}=  (\bar{v}\oplus \bar{l})\oplus(\bar{w}\oplus \bar{k}) 
\end{equation}
where 
\begin{equation}\label{bracketlk-2}
\begin{split}
& 
\bar{v}= l \downharpoonleft v'-l'\downharpoonleft v+k \curvearrowright v'-k'\curvearrowright v+\epsilon(k,l') -\epsilon(k',l)+\varphi(l,l'),
\\
& 
\bar{l}= [l,l']+k \blacktriangleright l'-k' \blacktriangleright l , 
\\
& 
\bar{w}= k \downharpoonright w'-k' \downharpoonright w+w \curvearrowleft l'-w' \curvearrowleft l+\iota(k,l') -\iota(k',l)+\phi(k,k')  ,
\\
&
\bar{k}=[k,k']+k \blacktriangleleft l'-k' \blacktriangleleft l.
\end{split}
\end{equation}

\textbf{Matched pair as a 2-cocycle.}
Now, we investigate that under which conditions the matched pair $\mathfrak{g} \bowtie \mathfrak{h} $ in \eqref{mp-2c} turns out to be a 2-cocycle extension by itself. To show this, we need to determine a left action a 2-cocycle. Let us determine these one by one.

\textbf{(Left action).} Recall the mutual actions in \eqref{Lieact-k-l} and the matched pair algebra $\mathfrak{l} \bowtie \mathfrak{k}$. 
It is evident that $\mathfrak{l} \bowtie \mathfrak{k}$ is a Lie subalgebra of $\mathfrak{g} \bowtie \mathfrak{h} $. 
Define a left action of  $\mathfrak{l} \bowtie \mathfrak{k}$ on the product space $V \oplus W$ as follows:
\begin{equation}\label{fifthact}
\rrbiprod: (\mathfrak{l} \bowtie \mathfrak{k}) \times (V \oplus W) \longrightarrow (V \oplus W), \quad  
(l\oplus k) ~ \rrbiprod ~ (v\oplus w) =\big(l \downharpoonleft v+k \curvearrowright v \big) \oplus\big( -w \curvearrowleft l+k \downharpoonright w\big),
\end{equation}  
where we have employed the left actions $\downharpoonleft$ and $\downharpoonright$ in \eqref{actionsoflk}, the actions
$\curvearrowright$ and $\curvearrowleft$ in \eqref{actionsofkl}. 
To be a left action, \eqref{fifthact} needs to satisfy the first condition in \eqref{matchpaircond}. We compute this as  
\begin{equation}
\begin{split}
(k \blacktriangleright l')\downharpoonleft v-(k' \blacktriangleright l)\downharpoonleft v+(k \blacktriangleleft l') \curvearrowright v-(k' \blacktriangleleft l) \curvearrowright v
&=l \downharpoonleft (k' \curvearrowright v)-l' \downharpoonleft (k \curvearrowright v) \\
&\quad k \curvearrowright (l' \downharpoonleft v)-k' \curvearrowright(l \downharpoonleft v) \\
\end{split}
\end{equation}
for the first entry in \eqref{fifthact}. Notice that, this is an equation defined on the vector space $V$. For the second entry, we have
\begin{equation}
\begin{split}
-w \curvearrowleft(k \blacktriangleright l')+w \curvearrowleft(k' \blacktriangleright l)+(k \blacktriangleleft l')\downharpoonright w-(k' \blacktriangleleft l)\downharpoonright w&= -(k' \downharpoonright w)\curvearrowleft l+(k \downharpoonright w)\curvearrowleft l' \\
&-k \downharpoonright(w \curvearrowleft l')+k'\downharpoonright(w \curvearrowleft l)
\end{split}
\end{equation}
where $\blacktriangleright$ and $\blacktriangleleft$ are the mutual actions exhibited in \eqref{Lieact-k-l}.

\textbf{(2-cocycle).}
 Later, we introduce a $(V \oplus W)$-valued 2-cocycle on $\mathfrak{l} \bowtie \mathfrak{k}$, in terms of the  2-cocycles $\varphi$ and $\phi$ given in \eqref{phi-psi}, as follows
\begin{equation}\label{theta}
\Theta: (\mathfrak{l} \bowtie \mathfrak{k}) \times (\mathfrak{l} \bowtie \mathfrak{k}) \longrightarrow V \oplus W, \quad \Theta((l\oplus k),(l' \oplus k'))=\big(\varphi(l , l')+\epsilon(k,l')-\epsilon(k',l)\big) \oplus \big(\phi(k, k')+\iota(k,l')-\iota(k',l)\big),
\end{equation}
One needs to ask the compatibility conditions, in \eqref{ext-act},
\begin{equation} \label{con-2}
\begin{split}
	0=&\varphi(l,k''\blacktriangleright l'-k'\blacktriangleright l'')+\epsilon(k,[l'',l'])+\epsilon(k,k''\blacktriangleright l'-k'\blacktriangleright l'')-\epsilon([k'',k'],l)-\epsilon(k''\blacktriangleleft l'-k'\blacktriangleleft l'',l)\\
	&- l \downharpoonleft(\varphi(l',l'')+\epsilon(k',l'')-\epsilon(k'',l'))-k \curvearrowright(\varphi(l',l'')+\epsilon(k',l'')-\epsilon(k'',l'))\\
	&\varphi(l',k\blacktriangleright l''-k''\blacktriangleright l)+\epsilon(k',[l,l''])+\epsilon(k',k\blacktriangleright l''-k''\blacktriangleright l)-\epsilon([k,k''],l')-\epsilon(k\blacktriangleleft l''-k''\blacktriangleleft l,l')\\
	&- l' \downharpoonleft(\varphi(l'',l)+\epsilon(k'',l)-\epsilon(k,l''))-k \curvearrowright(\varphi(l'',l)+\epsilon(k'',l)-\epsilon(k,l''))\\
	&\varphi(l'',k'\blacktriangleright l-k\blacktriangleright l')+\epsilon(k'',[l',l])+\epsilon(k'',k'\blacktriangleright l-k\blacktriangleright l')-\epsilon([k',k],l'')-\epsilon(k'\blacktriangleleft l-k\blacktriangleleft l',l'')\\
	&- l'' \downharpoonleft(\varphi(l,l')+\epsilon(k,l')-\epsilon(k',l))-k \curvearrowright(\varphi(l,l')+\epsilon(k,l')-\epsilon(k',l))\\
	0=&\phi(l,k''\blacktriangleright l'-k'\blacktriangleright l'')+\iota(k,[l'',l'])+\iota(k,k''\blacktriangleright l'-k'\blacktriangleright l'')-\iota([k'',k'],l)-\iota(k''\blacktriangleleft l'-k'\blacktriangleleft l'',l)\\
	&- l \downharpoonleft(\phi(l',l'')+\iota(k',l'')-\iota(k'',l'))-k \curvearrowright(\phi(l',l'')+\iota(k',l'')-\iota(k'',l'))\\
	&\phi(l',k\blacktriangleright l''-k''\blacktriangleright l)+\iota(k',[l,l''])+\iota(k',k\blacktriangleright l''-k''\blacktriangleright l)-\iota([k,k''],l')-\iota(k\blacktriangleleft l''-k''\blacktriangleleft l,l')\\
	&- l' \downharpoonleft(\phi(l'',l)+\iota(k'',l)-\iota(k,l''))-k \curvearrowright(\phi(l'',l)+\iota(k'',l)-\iota(k,l''))\\
	&\varphi(l'',k'\blacktriangleright l-k\blacktriangleright l')+\iota(k'',[l',l])+\iota(k'',k'\blacktriangleright l-k\blacktriangleright l')-\iota([k',k],l'')-\iota(k'\blacktriangleleft l-k\blacktriangleleft l',l'')\\
	&- l'' \downharpoonleft(\phi(l,l')+\iota(k,l')-\iota(k',l))-k \curvearrowright(\phi(l,l')+\iota(k,l')-\iota(k',l)).
\end{split}
\end{equation}
We are ready now to define a 2-cocycle extension of the Lie alegbra $\mathfrak{l} \bowtie \mathfrak{k}$ via $V \oplus W$. We record this by 
\begin{equation} \label{2-co-mp}
(V \oplus W) {_\Theta\rtimes} (\mathfrak{l} \bowtie \mathfrak{k}).
\end{equation}
To arrive at the Lie bracket on this space, one only needs to employ the explicit definitions of the left action $\rrbiprod$ in 2-cocycle $\Theta$ into the generic formula \eqref{AAAAA} of Lie bracket for 2-cocycle extensions. This results with 
 \begin{equation}\label{AAAAA--}
 \begin{split}
 &[(v\oplus w) \oplus (l\oplus k), (v '\oplus w ') \oplus (l'\oplus k')]_{_\Theta\rtimes}= \big((l\oplus k) \, \rrbiprod \, (v '\oplus w ')-(l'\oplus k') \, \rrbiprod \, (v\oplus w) \\
 &\hspace{4cm}
 + \Theta ((l\oplus k),(l'\oplus k'))    \big) \oplus[(l\oplus k),(l'\oplus k')],
 \end{split}
 \end{equation} 
where $ \rrbiprod$ is the left action in \eqref{fifthact}, and the bracket $[(l\oplus k),(l'\oplus k')]$ is the matched pair Lie bracket on $\mathfrak{l} \bowtie \mathfrak{k}$. 
 It is immediate to see that extended Lie algebra bracket on \eqref{AAAAA--} is precisely  equal to the matched Lie bracket given  in \eqref{bracketlk-1} and \eqref{bracketlk-2}, up to some reordering. 
Eventually, we are ready now to collect all the discussions done so far in the following proposition.
\begin{proposition}\label{Prop-Gokhan}
The matched pair Lie algebra $\G{g}\bowtie \G{h}$, given in \eqref{mp-2c}, of 2-cocycle extensions $\mathfrak{g}=V {_\varphi\rtimes} \mathfrak{l}$ and $ \mathfrak{h}=W {_\phi\rtimes} \mathfrak{k}$ is a 2-cocycle extension admitting a left action $\rrbiprod$ in \eqref{fifthact} and a 2-cocycle $\Theta$ in \eqref{theta} that is 
\begin{equation}\label{claim}
 (V {_\varphi\rtimes} \:\mathfrak{l}) \bowtie  (W {_\phi\rtimes} \: \mathfrak{k}) \cong (V\oplus W) {_\Theta\rtimes} \: (\mathfrak{l} \bowtie \mathfrak{k}).
\end{equation}  	 
\end{proposition}
Even though, we have derived all the mappings and conditions up to now explicitly. There is a short but implicit way to arrive at that proposition by employing Universal Lemma \ref{universal-prop}. For this, we first embed Lie subalgebras to the space $(V \oplus W) {_\Theta\rtimes} (\mathfrak{l} \bowtie \mathfrak{k})$ as follows
\begin{equation}
\begin{split}
\mathfrak{g}&\longrightarrow (V \oplus W) {_\Theta\rtimes} (\mathfrak{l} \bowtie \mathfrak{k}), \qquad (v \oplus l) \mapsto (v \oplus 0)\oplus(l \oplus 0)
\\
\mathfrak{h}&\longrightarrow (V \oplus W) {_\Theta\rtimes} (\mathfrak{l} \bowtie \mathfrak{k}), \qquad (w\oplus k) \mapsto (0 \oplus w) \oplus( 0 \oplus k).
\end{split}
\end{equation}
Then Universal Lemma \ref{universal-prop} reads that the total space admits a matched pair decomposition. 

\textbf{A Particular case.} 
For future reference, we now examine a particular case of Proposition \ref{Prop-Gokhan}. First, we choose the left actions $\downharpoonleft $ and $\downharpoonright$ in \eqref{actionsoflk} as trivial. So that, the Lie brackets on the 2-cocycle extensions in \eqref{bracketontwococ} turn out to be 
 \begin{equation} 
 \begin{split}
 [v \oplus l, v' \oplus l']_{_\varphi\rtimes}&=  \varphi(l,l')   \oplus[l,l'], \\
 [w \oplus k, w' \oplus k']_{_\phi\rtimes}&= \phi(k,k')  \oplus[k,k']. \\ 
 \end{split}
 \end{equation}
In addition, consider that the mutual actions of $\G{k}$ and $\G{l}$ in \eqref{Lieact-k-l} and the cross representations in \eqref{actionsofkl} are all zero. Hence, we have that, the mutual actions of $\G{h}$ and $\G{g}$ in \eqref{left-comp} are 
\begin{equation}\label{heismutual}
\begin{split}
\vartriangleright &: ((w\oplus k),(v\oplus l))\mapsto \left( \epsilon(k,l) \oplus 0 \right),
\\ 
\vartriangleleft &: ((w\oplus k),(v\oplus l))\mapsto \left(\iota(k,l) \oplus 0 \right).
\end{split}
\end{equation} 
A direct calculation gives us that, in the present setting, $\vartriangleright$ is a left action if and only if $\epsilon(k,[l,l'])=0$, and $\vartriangleleft$ is a right action  if and only if  $\iota([k,k'],l)=0$. Eventually, we claim that, the assumptions in this case reduce the matched Lie algebra bracket in \eqref{bracketlk-1} and \eqref{bracketlk-2} as
\begin{equation}\label{heisbracket}
\begin{split}
\big[\big ( (v\oplus l)\oplus (w\oplus k)\big),\big((v'\oplus l')\oplus(w'\oplus k')\big)\big]_{\bowtie}=&\Big(\big(\epsilon(k,l')-\epsilon(k',l)+\varphi(l,l')\big)\oplus 0\Big) \\ 
&\oplus \Big(\big(\iota(k,l') -\iota(k',l)+\phi(k,k') \big) \oplus 0 \Big)
\end{split}
\end{equation}
where $\epsilon$ and $\iota$ are as in \eqref{epsiot}. In order to implement Proposition \ref{Prop-Gokhan}, we exploit a proper left action and a 2-cocycle operator. Here, according to \eqref{fifthact}, we take the left action as trivial whereas the 2-cocycle term $\Theta$ is precisely equal to the one in \eqref{theta}.

\subsection{Lie Poisson Dynamics on Duals of Matched 2-cocycle Extensions} \label{twococyclematchedLie}

We have exhibited matched pair of 2-cocycle extensions in the previous subsection. To arrive at Hamiltonian dynamics on the dual picture, we make use the Lie-Poisson formalism on 2-cocycle extensions, that is the theory in Subsection \ref{LPD-2c}. In that subsection, there exist both the Lie-Poisson bracket \eqref{centextLiepois} and the Lie-Poisson equations  \eqref{centextLiepois2}  
for the case of 2-cocycles.

\textbf{Lie-Poisson brackets.}
Let us consider the following notation on the dual spaces 
\begin{equation}\label{elements}
\mu=\mu_{V^*} \oplus \mu_{W^*}  \in V^* \oplus W^*,\qquad 
\nu= \nu_{\mathfrak{l}^*} \oplus \nu_{\mathfrak{k}^*}  \in  \mathfrak{l}^* \oplus \mathfrak{k} ^*.
\end{equation}
Substitute the 2-cocycle $\Theta$ in \eqref{theta} and the left action $\rrbiprod$ in \eqref{fifthact} in the Lie-Poisson bracket \eqref{centextLiepois}.
Therefore, for this situation, (plus/minus) Lie-Poisson bracket is
\begin{equation} \label{Liepoiscenter}
\begin{split}
\left\{ \mathcal{H},\mathcal{F}\right\}_{_\Theta\rtimes}(\mu\oplus\nu) 
= &
\pm 
\Big\langle 
\nu_{\mathfrak{l}^*} \oplus \nu_{\mathfrak{k}^*} ,
\big[
\big( 
\frac{\delta \mathcal{H}}{\delta 
\nu_{\mathfrak{l}^*}}\oplus \frac{\delta \mathcal{H}}{\delta \nu_{\mathfrak{k}^*}}
\big) 
,
\big( 
\frac{\delta \mathcal{F}}{\delta 
\nu_{\mathfrak{l}^*}}\oplus \frac{\delta \mathcal{F}}{\delta \nu_{\mathfrak{k}^*}}
\big) 
\big] 
 	\Big\rangle
 	 \\ 
 	 & \pm 
 	 \Big\langle 
 	 \mu_{V^*} \oplus \mu_{W^*}, 
 	 \Theta \Big(
 	 \frac{\delta \mathcal{H}}{\delta 
\nu_{\mathfrak{l}^*}}\oplus \frac{\delta \mathcal{H}}{\delta \nu_{\mathfrak{k}^*}},
\frac{\delta \mathcal{F}}{\delta 
\nu_{\mathfrak{l}^*}}\oplus \frac{\delta \mathcal{F}}{\delta \nu_{\mathfrak{k}^*}}
\Big) \Big \rangle 
\\
&
  \pm
    \Big\langle \mu_{V^*} \oplus \mu_{W^*} ,
  \big( \frac{\delta \mathcal{H}}{\delta 
\nu_{\mathfrak{l}^*}}\oplus \frac{\delta \mathcal{H}}{\delta \nu_{\mathfrak{k}^*}}\big)
\rrbiprod 
	  \big(
	\frac{\delta \mathcal{F}}{\delta 
\mu_{V^*}}\oplus \frac{\delta \mathcal{F}}{\delta \mu_{W^*}}  \big)
	  \Big \rangle \\
	&
  \mp
    \Big\langle \mu_{V^*} \oplus \mu_{W^*} ,
  \big( \frac{\delta \mathcal{F}}{\delta 
\nu_{\mathfrak{l}^*}}\oplus \frac{\delta \mathcal{F}}{\delta \nu_{\mathfrak{k}^*}}\big)
\rrbiprod 
	  \big(
	\frac{\delta \mathcal{H}}{\delta 
\mu_{V^*}}\oplus \frac{\delta \mathcal{H}}{\delta \mu_{W^*}}  \big)
	  \Big \rangle,
	\end{split}
\end{equation}
where the bracket on the first line is the matched pair Lie bracket on $\G{l}\bowtie\G{k}$. Here, the first pairing is the one between   $\G{l}^*\times \G{k}^*$ and $\G{l}\bowtie\G{k}$, the others are the pairing between $V^*\oplus W^*$ and $V\oplus W$. It is maybe needless to say that we have assumed that all the vector spaces studied in here are reflexive. 
If the explicit expressions for the Lie bracket on $\G{l}\bowtie\G{k}$, the left action $\rrbiprod$ in \eqref{fifthact}, and the 2-cocycle $\Theta$ in \eqref{theta} are substituted into the bracket \eqref{Liepoiscenter}, one arrives at the bracket as 
\begin{equation} \label{Liepoiscenter-2}
\begin{split}
&
\left\{ \mathcal{H},\mathcal{F}\right\}_{_\Theta\rtimes}(\mu\oplus\nu) 
=
\pm \big
\langle \nu_{\mathfrak{l}^*}, \big [ \frac{\delta \mathcal{H}}{\delta \nu_{\mathfrak{l}^*}}, \frac{\delta \mathcal{F}}
{\delta \nu_{\mathfrak{l}^*}}\big] +\frac{\delta \mathcal{H}}{ \delta \nu_{\mathfrak{k}^*}} \blacktriangleright \frac{\delta \mathcal{F}}{\delta  \nu_{\mathfrak{l}^*}}-\frac{\delta \mathcal{F}}{\delta  \nu_{\mathfrak{k}^*}} \blacktriangleright \frac{\delta \mathcal{H}}{\delta \nu_{\mathfrak{l}^*}} \big \rangle 
\\
&
\hspace{3,1cm}
\pm \big
\langle \nu_{\mathfrak{k}^*},\big \lbrack \frac{\delta \mathcal{H}}{\delta \nu_{\mathfrak{k}^*}}, \frac{\delta \mathcal{F}}{\delta \nu_{\mathfrak{k}^*}}\big \rbrack
+
\frac{\delta \mathcal{H}}{\delta \nu_{\mathfrak{k}^*}} \blacktriangleleft \frac{\delta \mathcal{F}}{\delta\nu_{\mathfrak{l}^*}}
-
\frac{\delta \mathcal{F}}{\delta  \nu_{\mathfrak{k}^*}} \blacktriangleleft \frac{\delta \mathcal{H}}{\delta \nu_{\mathfrak{l}^*}} 
\big \rangle 
\\
&
\hspace{3,1cm}  \pm \big
\langle \mu_{V^*}, \frac{\delta \mathcal{H}}{\delta \nu_{\mathfrak{l}^*}} \downharpoonleft \frac{\delta \mathcal{F}}{\delta \mu_{V^*}}
+
\frac{\delta \mathcal{H}}{\delta \nu_{\mathfrak{k}^*}} 
\curvearrowright 
\frac{\delta \mathcal{F}}{\delta \mu_{V^*}}
+
\frac{\delta \mathcal{F}}{\delta \nu_{\mathfrak{l}^*}}
\downharpoonleft 
\frac{\delta \mathcal{H}}{\delta \mu_{V^*}}
+
\frac{\delta \mathcal{F}}{\delta \nu_{\mathfrak{k}^*}} 
\curvearrowright
\frac{\delta \mathcal{H}}{\delta \mu_{V^*}}
 \big \rangle
\\
& \hspace{3,1cm}
\pm \big \langle \mu_{V^*},\epsilon(\frac{\delta \mathcal{H}}{\delta \nu_{ \mathfrak{k}^*}},\frac{\delta \mathcal{F}}{\delta \nu_{ \mathfrak{l}^*}})-\epsilon(\frac{\delta \mathcal{F}}{\delta \nu_{ \mathfrak{k}^*}},\frac{\delta \mathcal{H}}{\delta \nu_{ \mathfrak{l}^*}})
+
\varphi(\frac{\delta \mathcal{H}}{\delta \nu_{\mathfrak{l}^*}},\frac{\delta \mathcal{F}}{\delta \nu_{\mathfrak{l}^*}}) \\
&
\hspace{3,1cm} \mp \big
\langle \mu_{W^*}, -\frac{\delta \mathcal{F}}{\delta \mu_{W^*}} \curvearrowleft \frac{\delta \mathcal{H}}{\delta \nu_{\mathfrak{l}^*}}+\frac{\delta \mathcal{H}}{\delta \nu_{ \mathfrak{k}^*}} \downharpoonright \frac{\delta \mathcal{F}}{\delta \mu_{W^*}}-\frac{\delta \mathcal{H}}{\delta \mu_{W^*}}\curvearrowleft\frac{\delta \mathcal{F}}{\delta \nu_{\mathfrak{l}^*}}+\frac{\delta \mathcal{F}}{\delta \nu_{\mathfrak{k}^*}} \downharpoonright \frac{\delta \mathcal{H}}{\delta \mu_{W^*}}
\big\rangle. \\
& \hspace{3,1cm} \pm \big \langle \mu_{W^*} , \iota(\frac{\delta \mathcal{H}}{\delta \nu_{ \mathfrak{k}^*}},\frac{\delta \mathcal{F}}{\delta \nu_{ \mathfrak{l}^*}})-\iota(\frac{\delta \mathcal{F}}{\delta \nu_{ \mathfrak{k}^*}},\frac{\delta \mathcal{H}}{\delta \nu_{ \mathfrak{l}^*}})+ \phi(\frac{\delta \mathcal{H}}{\delta \nu_{ \mathfrak{k}^*}},\frac{\delta \mathcal{F}}{\delta \nu_{ \mathfrak{k}^*}}) \big \rangle
\end{split}
\end{equation}
Here, $\blacktriangleright$ and $\blacktriangleleft$ are the actions in \eqref{Lieact-k-l},  $\curvearrowleft $ and $\curvearrowright$ are in \eqref{actionsofkl} whereas   $\downharpoonleft $ 
and $\downharpoonleft$ are those in \eqref{actionsoflk}. 

In the light of Proposition \ref{Prop-Gokhan}, one can write the bracket \eqref{Liepoiscenter} as a matched pair Lie-Poisson bracket. To this end, by reordering the elements in \eqref{elements}, define
\begin{equation}
\tilde{\mu}=\mu_{V^*}\oplus \nu_{\mathfrak{l}^*}\in \mathfrak{g}=V {_\varphi\rtimes} \mathfrak{l},
\qquad
\tilde{\nu}=\mu_{W^*}\oplus \nu_{\mathfrak{k}^*}\in \mathfrak{h}=W {_\phi\rtimes}  \mathfrak{k}.
\end{equation} 
Then the matched pair Lie-Poisson bracket in \eqref{LiePoissonongh} takes the form
\begin{equation} 
\begin{split} \label{LiePoissonongh-1}
 \left\{ \mathcal{H},\mathcal{F}\right\}_{\bowtie}(\tilde{\mu}\oplus\tilde{\nu}) =& \pm 
\left\{ \mathcal{H},\mathcal{F}\right\}_{\varphi\rtimes }(\tilde{\mu})
	\pm 
\left\{ \mathcal{H},\mathcal{F}\right\}_{\phi\rtimes }(\tilde{\nu}) 
\\
 & 
\mp
\Big  \langle 
\mu_{V^*}\oplus \nu_{\mathfrak{l}^*} 
,
\big( 
\frac{\delta \mathcal{H}}{\delta 
\mu_{W^*}}\oplus \frac{\delta \mathcal{H}}{\delta \nu_{\mathfrak{k}^*}}
\big)
	\vartriangleright 
\big( 
\frac{\delta \mathcal{F}}{\delta 
\mu_{V^*}}\oplus \frac{\delta \mathcal{F}}{\delta \nu_{\mathfrak{l}^*}}
\big)
-
\big( 
\frac{\delta \mathcal{F}}{\delta 
\mu_{W^*}}\oplus \frac{\delta \mathcal{F}}{\delta \nu_{\mathfrak{k}^*}}
\big)
	\vartriangleright 
\big( 
\frac{\delta \mathcal{H}}{\delta 
\mu_{V^*}}\oplus \frac{\delta \mathcal{H}}{\delta \nu_{\mathfrak{l}^*}}
\big)
\Big \rangle
 \\
&  \mp \Big \langle \mu_{W^*}\oplus \nu_{\mathfrak{k}^*} 
,
\big( 
\frac{\delta \mathcal{H}}{\delta 
\mu_{W^*}}\oplus \frac{\delta \mathcal{H}}{\delta \nu_{\mathfrak{k}^*}}
\big)
	\vartriangleleft 
\big( 
\frac{\delta \mathcal{F}}{\delta 
\mu_{V^*}}\oplus \frac{\delta \mathcal{F}}{\delta \nu_{\mathfrak{l}^*}}
\big)
-
\big( 
\frac{\delta \mathcal{F}}{\delta 
\mu_{W^*}}\oplus \frac{\delta \mathcal{F}}{\delta \nu_{\mathfrak{k}^*}}
\big)
	\vartriangleleft 
\big( 
\frac{\delta \mathcal{H}}{\delta 
\mu_{V^*}}\oplus \frac{\delta \mathcal{H}}{\delta \nu_{\mathfrak{l}^*}}
\big)
\Big \rangle, 
\end{split}
\end{equation}
where the actions $\vartriangleleft$ and $\vartriangleright$ are in \eqref{left-comp}. Notice that, the Lie-Poisson brackets on the right hand side of the first line are the individual Lie-Poisson brackets on the dual spaces $\G{g}^*$ and $\G{h}^*$, respectively. Those terms available in the second line of \eqref{LiePoissonongh-1} are manifestations of the left action of $\G{h}$ on $\G{g}$ whereas the third  line is due to the right action of $\G{g}$ on $\G{h}$. A direct computation gives that the Lie-Poisson bracket \eqref{LiePoissonongh-1} is equal to the Lie-Poisson bracket in \eqref{Liepoiscenter-2}. 

\textbf{Dual actions:} 
First, define the dual actions of
$\curvearrowright$ and $\curvearrowleft$ in \eqref{actionsofkl}  as
\begin{equation} \label{dualofyanok}
\begin{split}  
\overset{\ast}{\curvearrowleft} &: V^* \rightarrow V^*, \qquad \langle
\mu_{V^*} \overset{\ast}{\curvearrowleft} k,v \rangle=\langle \mu_{V^*},
k \curvearrowright v \rangle, \\
\overset{\ast}{\curvearrowright} &: W^* \rightarrow W^*, \qquad \langle
l  \overset{\ast}{\curvearrowright} \mu_{W^*},w \rangle=\langle \mu_{W^*}, w
\curvearrowleft l \rangle,
\end{split} 
\end{equation} 
respectively. See that, $\overset{\ast}{\curvearrowleft}$ is a right action whereas $\overset{\ast}{\curvearrowright}$ is a left action. Secondly, introduce the dual (right) actions of $\downharpoonleft$ and $\downharpoonright$  in \eqref{actionsoflk} as 
\begin{equation} \label{dualofduzok}
\begin{split}
\overset{\ast}{\downharpoonleft} k &: W^* \rightarrow W^*, \qquad
\langle \mu_{W^*} \overset{\ast}{\downharpoonleft} k,w \rangle=\langle
\mu_{W^*}, k \downharpoonright w \rangle, \\
 \overset{\ast}{\downharpoonright} l &: V^*
\rightarrow V^*, \qquad \langle \mu_{V}
\overset{\ast}{\downharpoonright} l,v \rangle=\langle \mu_{V^*}, l
\downharpoonleft v \rangle,
\end{split}
\end{equation}
respectively. Then, determine the dual actions of $\blacktriangleleft$ and $\blacktriangleright$ in \eqref{Lieact-k-l}
\begin{equation}\label{dualofkara}
\begin{split}
\overset{\ast}{\blacktriangleleft}&: \mathfrak{k} \otimes \mathfrak{l}^* \longrightarrow \mathfrak{l}^*, \qquad \langle k \overset{\ast}{\blacktriangleleft} \nu_{\mathfrak{l}^*}, l \rangle=\langle \nu_{\mathfrak{l}^*}, k \blacktriangleright l \rangle, \\
\overset{\ast}{\blacktriangleright}&: \mathfrak{l} \otimes \mathfrak{k}^* \longrightarrow \mathfrak{k}^*, \qquad \langle l \overset{\ast}{\blacktriangleright} \nu_{\mathfrak{k}^*}, k \rangle=\langle \nu_{\mathfrak{k}^*}, k \blacktriangleleft l \rangle,
\end{split}
\end{equation} 
respectively.
We define the dual mappings of the 2-cocycles $\varphi$ and $\phi$ as
\begin{equation}\label{dualofvp}
\begin{split}
\varphi^*_{l}&:V^* \longrightarrow \mathfrak{l}^*, \qquad \langle \varphi^*_{l} \mu_{V^*},l'\rangle=-\langle \mu_{V^*}, \varphi_{l} l' \rangle, \\
\phi^*_{k}&:W^* \longrightarrow \mathfrak{k}^*, \qquad \langle \phi^*_{k} \mu_{W^*},k'\rangle=-\langle \mu_{W^*}, \phi_{k} k' \rangle.
\end{split}
\end{equation} 
Lastly, exhibit the dual mappings of $\epsilon$ and $\iota$ in \eqref{epsiot} as
\begin{equation}\label{dualofei}
\begin{split}
\epsilon^*_{k}&:V^* \longrightarrow \mathfrak{l}^*, \qquad \langle \epsilon^*_{k} \mu_{V^*},l \rangle=-\langle \mu_{V^*}, \epsilon_{k} l \rangle, \\
\iota^*_{k}&:W^* \longrightarrow \mathfrak{l}^*, \qquad \langle \iota^*_{k} \mu_{W^*},l \rangle=-\langle \mu_{W^*}, \iota_{k} l \rangle.
\end{split}
\end{equation}

\textbf{Lie-Poisson equations.} According to the equations
\eqref{centextLiepois2}, it suffice to define dual mappings of the
action $\rrbiprod$ \eqref{fifthact} and $\Theta$ \eqref{theta} to write Lie-Poisson dynamics. For $\rrbiprod$, by the definiton, we compute
\begin{equation}\label{dualoffifth1}  
\begin{split} 
  \langle
(\mu_{V^*}\oplus \mu_{W^*})\,\overset{\ast}{\lrbiprod}\,
(l\oplus k),(v\oplus w)\rangle &=\langle\mu_{V^*},l \downharpoonleft v+ k
\curvearrowright v \rangle+\langle \mu_{W^*},-w \curvearrowleft l+k
\downharpoonright w \rangle, 
\\
 &=\langle \mu_{V^*}, l \downharpoonleft v
\rangle+\langle \mu_{V^*}, k \curvearrowright v \rangle+\langle
\mu_{W^*},-w \curvearrowleft l \rangle+ \langle \mu_{W^*},
k\downharpoonright w \rangle. 
\end{split}
\end{equation} 
Then, we have 
\begin{equation} \label{dotactiondual}
(\mu_{V^*}\oplus \mu_{W^*})~\overset{\ast}{\lrbiprod}~
(l\oplus k)= \big(\mu_{V^*} \overset{\ast}{\curvearrowleft} k+\mu_{V^*}
\overset{\ast}{\downharpoonright} l \big) \oplus \big(  l  \overset{\ast}{\curvearrowright} \mu_{W^*}+\mu_{W^*} \overset{\ast}{\downharpoonleft} k \big).   
\end{equation}
For $\Theta$, we compute 
\begin{equation} \label{thetadual}
\Theta^*_{(l\oplus k)}(\mu_{V^*}\oplus \mu_{W^*})=(\varphi^*_{l}\mu_{V^*}+\epsilon^*_{k}\mu_{V^*}-\iota^*_{k}\mu_{W^*}) \oplus (\phi^*_{k}\mu_{W^*}+\iota^*_{k}\mu_{W^*}  -\epsilon^*_{k}\mu_{V^*}),
\end{equation}
where $\varphi^*_{l}, \phi^*_{k}$  defined as \eqref{dualofvp} and $\epsilon^*_{k}, \iota^*_{k}$ defined as \eqref{dualofei}.
There are two more dual mappings we need to get for the left side of the equation \eqref{centextLiepois2}. These are $\mathfrak{b}^*_{(l\oplus k)}(\mu_{V^*}\oplus \mu_{W^*})$ and the coadjoint action of $(\frac{\delta \mathcal{H}}{\delta \nu_{\mathfrak{l}^*}}\oplus \frac{\delta \mathcal{H}}{\delta \nu_{\mathfrak{k}^*}})$ on the dual element $ (\nu_{\mathfrak{l}^*}\oplus \nu_{\mathfrak{k}^*})$. Being   a matched pair, we can employ the equation \eqref{ad-*} for the case of $\mathfrak{l}\bowtie \mathfrak{k}$. Accordingly, we arrive at
\begin{equation}\label{addual}
ad^*_{(\frac{\delta \mathcal{H}}{\delta \nu_{\mathfrak{l}^*}}\oplus \frac{\delta \mathcal{H}}{\delta \nu_{\mathfrak{k}^*}})}(\nu_{\mathfrak{l}^*}\oplus  \nu_{\mathfrak{k}^*})
=\big(ad^*_{\frac{\delta \mathcal{H}}{\delta \nu_{\mathfrak{l}^*}}} \nu_{\mathfrak{l}^*}+\nu_{\mathfrak{l}^*} \overset{\ast}{\blacktriangleleft} \frac{\delta \mathcal{H}}{\delta \nu_{\mathfrak{k}^*}} + \nu_{\mathfrak{k}^*} \overset{\ast}{\blacktriangleright} \frac{\delta \mathcal{H}}{\delta \nu_{\mathfrak{k}^*}} \big) \oplus \big( ad^*_{\frac{\delta \mathcal{H}}{\delta \nu_{\mathfrak{k}^*}}} \nu_{\mathfrak{k}^*}-\nu_{\mathfrak{l}^*} \overset{\ast}{\blacktriangleleft} \frac{\delta \mathcal{H}}{\delta \nu_{\mathfrak{l}^*}} - \nu_{\mathfrak{k}^*} \overset{\ast}{\blacktriangleright} \frac{\delta \mathcal{H}}{\delta \nu_{\mathfrak{l}^*}}\big), 
\end{equation}
where $\overset{\ast}{\blacktriangleleft}$ and $\overset{\ast}{\blacktriangleright} $ denote dual actions in the equation \eqref{dualofkara}. 
Finally, using equation \eqref{b*}, we arrive 
\begin{equation} \label{bdual}
\mathfrak{b}^*_{(l,k)}\mu_{V^*}=(\mu_{V^*}\overset{\ast}{\downharpoonright} l+ \mu_{V^*} \overset{\ast}{\curvearrowleft} k) \oplus (\mu_{W^*} \overset{\ast}{\downharpoonleft}k-\mu_{W^*}\overset{\ast}{\curvearrowright} l).
\end{equation} 
 Therefore, according to the equations
 \eqref{dotactiondual}, \eqref{thetadual}, \eqref{addual},\eqref{bdual} the (plus/minus) Lie-Poisson equations (governed by Hamiltonian function $\mathcal{H}=\mathcal{H}((\mu_{V^*},\mu_{W^*}),(\nu_{\mathfrak{l}^*},\nu_{\mathfrak{k}^*}))$  is computed as
\begin{equation}\label{Liepoismc1}
\begin{split}
\dot{\mu}_{V^*}&= \mu_{V^*}  \overset{\ast}{\downharpoonright} \frac{\delta \mathcal{H}}{\delta \nu_{\mathfrak{l}^*}}+\mu_{V^*}\overset{\ast}{\curvearrowleft}\frac{\delta \mathcal{H}}{\delta \nu_{\mathfrak{k}^*}},
\\
\dot{\mu}_{W^*} &= -\frac{\delta \mathcal{H}}{\delta \nu_{\mathfrak{l}^*}}\overset{\ast}{\curvearrowright}\mu_{W^*} + \mu_{W^*} \overset{\ast}{\downharpoonleft} \frac{\delta \mathcal{H}}{\delta \nu_{\mathfrak{k}^*}},
\\
\dot{\nu}_{\mathfrak{l}^*}&= \varphi^*_{l}\mu_{V^*}+\epsilon^*_{k}\mu_{V^*}-\iota^*_{k}\mu_{W^*}+ad^*_{\frac{\delta \mathcal{H}}{\delta \nu_{\mathfrak{l}^*}}} \nu_{\mathfrak{l}^*}+\nu_{\mathfrak{l}^*} \overset{\ast}{\blacktriangleleft} \frac{\delta \mathcal{H}}{\delta \nu_{\mathfrak{k}^*}} + \nu_{\mathfrak{k}^*} \overset{\ast}{\blacktriangleright} \frac{\delta \mathcal{H}}{\delta \nu_{\mathfrak{k}^*}}+\mu_{V^*}\overset{\ast}{\downharpoonright} l+ \mu_{V^*} \overset{\ast}{\curvearrowleft} k, \\
\dot{\nu}_{\mathfrak{k}^*}&= \phi^*_{k}\mu_{W^*}+\iota^*_{k}\mu_{W^*}-\epsilon^*_{k}\mu_{V^*}+ad^*_{\frac{\delta \mathcal{H}}{\delta \nu_{\mathfrak{k}^*}}} \nu_{\mathfrak{k}^*}-\nu_{\mathfrak{l}^*} \overset{\ast}{\blacktriangleleft} \frac{\delta \mathcal{H}}{\delta \nu_{\mathfrak{l}^*}} - \nu_{\mathfrak{k}^*} \overset{\ast}{\blacktriangleright} \frac{\delta \mathcal{H}}{\delta \nu_{\mathfrak{l}^*}}+\mu_{W^*} \overset{\ast}{\downharpoonleft}k-\mu_{W^*}\overset{\ast}{\curvearrowright} l.
\end{split}
\end{equation}

\textbf{Particular case:}
Now, we examine how Lie-Poisson equations look like for the particular case we gave in \eqref{centrallyextend}. It is possible to briefly remind: we took   
the left actions $\downharpoonleft $ and $\downharpoonright$ in \eqref{actionsoflk}, the mutual actions of $\G{k}$ and $\G{l}$ in \eqref{Lieact-k-l} and the cross representations in \eqref{actionsofkl} as trivial. Since the calculations of these choices are made in the previous section, it is possible to see the effects directly for Lie-Poisson equations \eqref{Liepoismc1}. See that, both $\dot{\mu}_{V^*}$ and $\dot{\mu}_{W^*}$ are zero and
\begin{equation}
\begin{split}
\dot{\nu}_{\mathfrak{l}^*}&=\varphi^*_{l}\mu_{V^*}+\epsilon^*_{k}\mu_{V^*}-\iota^*_{k}\mu_{W^*}+ad^*_{\frac{\delta \mathcal{H}}{\delta \nu_{\mathfrak{l}^*}}} \nu_{\mathfrak{l}^*} \\
\dot{\nu}_{\mathfrak{k}^*}&=\phi^*_{k}\mu_{W^*}+\iota^*_{k}\mu_{W^*}-\epsilon^*_{k}\mu_{V^*}+ad^*_{\frac{\delta \mathcal{H}}{\delta \nu_{\mathfrak{k}^*}}} \nu_{\mathfrak{k}^*}.
\end{split}
\end{equation}

\section{Couplings of Dissipative Systems}
 \label{Sec-C-DS}

In the present section, we consider two Lie algebras  $\G{g}$ and $\G{h}$ under mutual actions \eqref{matched-pair-mutual-actions} satisfying the compatibility conditions in \eqref{matched-pair-mutual-actions}. So that, we have a well-define matched pair Lie algebra structure $\mathfrak{g}\bowtie\mathfrak{h}$ equipped with the matched pair Lie bracket \eqref{mpla}. As explained in Subsection \ref{Sec-MP-LP}, on the dual space $\mathfrak{g}^{\ast}\oplus\mathfrak{h}^{\ast}$, there exists matched Lie-Poisson bracket $\{\bullet, \bullet\}_{\bowtie}$ displayed in  \eqref{LiePoissonongh}. This Poisson structure lets us to arrive at matched pair Lie-Poisson equations \eqref{LPEgh} governing the collective motion of the individual Lie-Poisson dynamics on $\G{g}^*$ and $\G{h}^*$. In all the subsection, we follow this construction. 
After presenting discussions on coupling of Rayleigh dissipation in the following subsection, we shall follow the order given in Subsection \ref{Sec-Sym-Bra}, and examine the coupling problem of symmetric brackets.

\subsection{Rayleigh type Dissipation}

In \eqref{RD-Eqn}, Rayleigh type dissipation is introduced to the Lie-Poisson system by means the coadjoint action and a linear operator, from the Lie algebra the dual space to the algebra. 
In this subsection, we provide a way to couple two 
Lie-Poisson dynamics admitting Rayleigh type dissipative terms. 
For this, we first determine the dynamics of constitutive systems. Assume that, on the dual space $\G{g}^*$, for the Lie-Poisson dynamics, Rayleigh type dissipation is provided by a linear operator $\Upsilon^\G{g}:\G{g}^*\mapsto \G{g}$ that is 
\begin{equation}\label{RD-Eqn-1}
\dot{\mu}\mp ad^*_{\frac{\delta \mathcal{H}}{ \delta \mu}} \mu = \mp ad^*_{\Upsilon^\G{g}(\mu)} \mu.
\end{equation}
Assume also that on $\G{h}^*$, for the Lie-Poisson dynamics, Rayleigh type dissipation is provided by a linear operator $\Upsilon^\G{h}:\G{h}^*\mapsto \G{h}$. so that, 
\begin{equation}\label{RD-Eqn-2}
\dot{\nu}\mp ad^*_{\frac{\delta \mathcal{H}}{ \delta \nu}} \nu = \mp ad^*_{\Upsilon^\G{h}(\nu)} \nu.
\end{equation}

To couple the dynamics in \eqref{RD-Eqn-1} and \eqref{RD-Eqn-2}, we introduce a linear operator from the dual space $\mathfrak{g}^{\ast}\oplus\mathfrak{h}^{\ast}$  to the matched pair Lie algebra $\mathfrak{g}\bowtie\mathfrak{h}$ given by
\begin{equation}
\mathfrak{g}^*\oplus\mathfrak{h}^* \longrightarrow \mathfrak{g} \oplus\mathfrak{h},\qquad  (\mu\oplus \nu) \mapsto (\Upsilon^{\G{g}}(\mu)\oplus \Upsilon^{\G{h}}(\nu)),
\end{equation}
where $\Upsilon^{\G{g}}$ and $\Upsilon^{\G{h}}$ are the linear mappings, in \eqref{RD-Eqn-1} and \eqref{RD-Eqn-2}, generating the dissipation for the individual systems. The dissipative term generated by the mapping $\Lambda$ is computed to be  
\begin{equation} \label{Rayleigh}
\mp ad^{*}_{\Upsilon^{\G{g}}(\mu) \oplus  \Upsilon^{\G{h}}(\nu)}(\mu \oplus \nu)=\big(\mp ad^{*}_{\Upsilon^{\G{g}}(\mu)}\mu
\pm
\mu \overset{\ast }{\vartriangleleft }\Upsilon^{\G{h}}(\nu)\mp \mathfrak{a}^{*}_{\Upsilon^{\G{h}}(\nu)}\nu \big) 
\oplus 
\big(\mp ad^{*}_{\Upsilon^{\G{h}}(\nu)}\nu
\pm
\Upsilon^{\G{g}}(\mu) \overset{\ast }{\vartriangleright } \nu\mp \mathfrak{b}^{*}_{\Upsilon^{\G{g}}(\mu)}\mu\big),
\end{equation} 
where the dual actions $\overset{\ast }{\vartriangleleft }$ and $\overset{\ast }{\vartriangleright }$ are those given in \eqref{eta-star} and \eqref{xi-star}, respectively. Notice that, the cross actions $\mathfrak{a}^{*}$ and $\mathfrak{b}^{*}$ are the ones in \eqref{a*} and \eqref{b*}, respectively. 
dissipation term will be obtained as above. Observe that, while coupling the dissipative terms in \eqref{Rayleigh}, we respect the mutual actions.  So that, the collective dissipative term is manifesting the mutual actions. It reduces to the direct sum of the dissipative terms of the individual motions if the actions are trivial. 

Obeying the general construction in \eqref{RD-Eqn}, we merge the dissipative terms in \eqref{Rayleigh} with the matched pair Lie-Poisson equations \eqref{LPEgh}. This reads the coupled system
\begin{equation} \label{eqofmoofRay}
\begin{split}
&\dot{\mu}
\mp
ad^{*}_{\delta \mathcal{H}/ \delta \mu}\mu\pm  \frac{\delta \mathcal{H}}{\delta \mu}  \overset{\ast }{\vartriangleleft }\nu
\mp
\mathfrak{a}^{*}_{\delta \mathcal{H} /  \delta \nu} \nu
=
\mp
ad^{*}_{\Upsilon^{\G{g}}(\mu)}\mu
\pm
\mu \overset{\ast }{\vartriangleleft }\Upsilon^{\G{h}}(\nu)
\mp
\mathfrak{a}^{*}_{\Upsilon^{\G{h}}(\nu)}\nu
\\
&\dot{\nu}
\mp
ad^{*}_{\delta \mathcal{H}/ \delta \nu}\nu 
\pm
\frac{\delta \mathcal{H}}{\delta \nu} \overset{\ast }{\vartriangleright } 
\mp
\mathfrak{b}^{*}_{\delta \mathcal{H} /  \delta \mu} \mu=
\mp
ad^{*}_{\Upsilon^{\G{h}}(\nu)}\nu
\pm
\Upsilon^{\G{g}}(\mu) \overset{\ast }{\vartriangleright }\nu
\mp 
\mathfrak{b}^{*}_{\Upsilon^{\G{g}}(\mu)}\mu.
\end{split}
\end{equation}
It is evident that, this formulation respects the mutual actions. By taking one these actions, one arrives at the semidirect product theory for the Lie-Poisson system with Rayleigh type dissipation. If both of the actions are trivial, it is immediate to see that the system \eqref{eqofmoofRay} turns out to be simple collection of the individual motions in \eqref{RD-Eqn-1} and \eqref{RD-Eqn-2}.

\subsection{Matched Double Bracket}

Recall that, in \eqref{doubledissi}, we have presented double bracket in terms of the structure constants of Lie algebra. So that, to have a symmetric bracket on the matched Lie-Poisson geometry, we first recall the structure constants of the matched pair Lie algebra given in \eqref{sc-nonbracket-mp}. Then referring to the coordinate realization of matched pair Lie-Poisson bracket, in \eqref{Lie-pois-double-non-bracket}, we compute associated  Poisson bivector $\Lambda$ as
\begin{equation}\label{cosym2}
\Lambda_{\alpha \beta}=\pm C_{\beta \alpha}^\gamma \mu_\gamma, \qquad
\Lambda_{\alpha b}=\mp R_{b \alpha}^d \nu_d
\mp L_{b \alpha}^\gamma \mu_\gamma, \qquad 
\Lambda_{a \beta }= \pm  R_{ \beta a}^d \nu_d \pm  L_{ \beta a}^\gamma \mu_\gamma, \qquad
\Lambda_{a b }=\pm  D_{ a b}^d \nu_d,
\end{equation}
where $C_{\beta \alpha}^\gamma$'s are structure constants on $\G{g}$, $D_{ a b}^d$'s are structure constants on $\G{h}$. Here, $R_{b \alpha}^d$'s and $L_{b \alpha}^\gamma $'s are constants defining the right and the left actions according to the exhibitions in \eqref{local-act}, respectively. In accordance with this coordinate realizations and in the light of the definition \eqref{doubledissi}, matched Double bracket dissipation $(\mathcal{F},\mathcal{S})^{(mD)}$ for two functions $ \mathcal{F} $ and $ \mathcal{S} $ defined on $ \mathfrak{g}^* \oplus \mathfrak{h}^* $ is
\begin{equation}\label{matchdouble}
\begin{split}
(\mathcal{F},\mathcal{S})^{(mD)}(\mu , \nu)=&(\sum_{b}\Lambda_{\alpha b} \Lambda_{ \beta b }+\sum_{\gamma}\Lambda_{\alpha \gamma} \Lambda_{ \beta \gamma })\frac{\partial \mathcal{F}}{\partial \mu_\alpha}\frac{\partial \mathcal{S}}{\partial \mu_\beta} 
+
(\sum_{b}\Lambda_{a b} \Lambda_{ \alpha b }+\sum_{\gamma}\Lambda_{a \gamma} \Lambda_{ \alpha  \gamma })\frac{\partial \mathcal{F}}{\partial \nu_a}\frac{\partial \mathcal{S}}{\partial \mu_\alpha }
\\
&  
+(\sum_{a}\Lambda_{b a} \Lambda_{ c a }+\sum_{\beta}\Lambda_{b \beta} \Lambda_{c \beta  })\frac{\partial \mathcal{F}}{\partial \nu_b}\frac{\partial \mathcal{S}}{\partial \nu_c}
+
(\sum_{\beta}\Lambda_{\alpha \beta} \Lambda_{ b \beta  } +
\sum_{a}\Lambda_{\alpha a} \Lambda_{ b a })\frac{\partial \mathcal{F}}{\partial \mu_\alpha}\frac{\partial \mathcal{S}}{\partial \nu_b}. \\
\end{split}
\end{equation}
Referring to this bracket, the dissipative dynamics \eqref{diss-dyn-eq} for $a=1$, generated by a functional $\mathcal{S}$, is computed to be
\begin{equation}\label{eqofmoofmatchdouble}
\begin{split}
&\dot{\mu}_\beta=(\sum_{b}\Lambda_{\beta b}\Lambda_{\alpha b}+\sum_{\gamma}\Lambda_{\beta \gamma}\Lambda_{ \alpha \gamma })\frac{\partial \mathcal{S}}{\partial \mu_\alpha} +(\sum_{\alpha}\Lambda_{\beta \alpha} \Lambda_{ a \alpha }+\sum_{b}\Lambda_{\beta b} \Lambda_{ a b })\frac{\partial \mathcal{S}}{\partial \nu_a} \\
&\dot{\nu}_d=(\sum_{b}\Lambda_{d b} \Lambda_{ \alpha b }+
\sum_{\gamma}\Lambda_{d \gamma} \Lambda_{ \alpha \gamma })\frac{\partial \mathcal{S}}{\partial \mu_\alpha}
+
(\sum_{a}\Lambda_{d a} \Lambda_{ n a }
+\sum_{\alpha}\Lambda_{d \alpha} \Lambda_{n \alpha  })\frac{\partial \mathcal{S}}{\partial \nu_n}.
\end{split}
\end{equation}
In order to arrive at the explicit expression of the symmetric bracket \eqref{matchdouble} and the disspative dynamics in \eqref{eqofmoofmatchdouble} in terms of the local characterizations of left and the right actions and the structure constants of the constitutive  subalgebras, one needs to substitute the calculations \eqref{cosym2} into \eqref{matchdouble} and \eqref{eqofmoofmatchdouble}.

Now, we add the matched Lie-Poisson bracket $\{\bullet,\bullet\}_{\bowtie}$, given in 
\eqref{Lie-pois-double-non-bracket-mp}, and the matched double bracket $(\bullet,\bullet)^{(mD)}$ in \eqref{matchdouble}. This reads the matched metriplectic bracket. The matched metriplectic dynamics generated by 
a Hamiltonian function $\C{H}$ and an entropy type function $\C{S}$ is computed to be 
\begin{equation}\label{mMD-DB}
\dot{\mathbf{z}}=[ \left\vert \mathbf{z},\mathcal{H}\right\vert ]_{\bowtie,D}= \{\mathbf{z},\mathcal{H}\}_{\bowtie}+a
(z,\mathcal{S})^{(mD)}
\end{equation}
In order to arrive at the explicit expression of the equations of motion in \eqref{mMD-DB}, it is enough to add the reversible matched pair dynamics in \eqref{Lie-pois-nonbracket-dyn-mp} and the irreversible matched pair dynamics in \eqref{eqofmoofmatchdouble}. By taking one of the actions trivial, one arrives semidirect product metriplectic bracket and semidirect product dynamical equation. If, both of the actions are trivial, then the coupling turns out to be a simple addition.

\subsection{Matched Cartan-Killing Bracket} \label{cartan}
Once more, we recall the structure constants \eqref{sc-nonbracket-mp} of the matched pair Lie algebra. Referring to \eqref{CK}, we first compute the matched pair Cartan metric as  
$\bar{\mathcal{G}}_{a \beta}$ and $\bar{\mathcal{G}}_{\alpha b}$ 
\begin{equation}\label{CM-MP}
\begin{split}
\bar{\mathcal{G}}_{\alpha b}&=-R_{a \alpha }^{d}D_{bd}^{a}
+L_{a \alpha}^{\beta}R_{b \beta}^{a}
+C_{\alpha \beta}^{\gamma}L_{b \gamma}^{\beta}, \qquad 
\bar{\mathcal{G}}_{ab}=L_{a \alpha  }^{\beta}L_{b \beta }^{\alpha}
+D_{ad}^{k}D_{b k}^{d}
\\
\bar{\mathcal{G}}_{a \beta}&=L_{a \epsilon  }^{\gamma}C_{ \beta \gamma}^{\epsilon}
-D_{ ab}^{d}R_{ d \beta}^{b}
+R_{a \gamma }^{d}L_{d \beta}^{\gamma}
, \qquad
\bar{\mathcal{G}}_{\alpha \beta}
=R_{a \alpha }^{b}R_{ b \beta }^{a}+C_{\alpha \gamma}^{\epsilon}C_{\beta \epsilon}^{\gamma}
\end{split}
\end{equation}
on the matched pair Lie algebra $\mathfrak{g} \oplus \mathfrak{h }$. We write an element of $ \mathfrak{g}^* \oplus \mathfrak{h}^* $ as $
(\mu,\nu)=\mu_{\alpha}\bar{\mathbf{e}}^{\alpha}+\nu_{a}\bar{\mathbf{e}}^a$.
We compute the matched pair Cartan-Killing bracket, defined in \eqref{loc-car}, as 
\begin{equation} \label{Gen-Killing}
\begin{split}
(\C{F},\C{H})^{(mCK)}
=
&
\frac{\partial \C{F}}{\partial \mu_{\alpha}}\bar{\mathcal{G}}_{\alpha \beta} \frac{\partial \C{H}}{\partial \mu_{\beta}}
+ 
\frac{\partial \C{F}}{\partial \mu_{\alpha}}\bar{\mathcal{G}}_{\alpha b} \frac{\partial \C{H}}{\partial \nu_{b}}
+
\frac{\partial \C{F}}{\partial \nu_{a}}\bar{\mathcal{G}}_{a \beta} \frac{\partial \C{H}}{\partial \mu_{\beta}}
+
\frac{\partial \C{F}}{\partial \nu_{a}}\bar{\mathcal{G}}_{a b} \frac{\partial \C{H}}{\partial \nu_{b}}
\\
=&
(R_{\alpha a}^{b}R_{ \beta b}^{a}+C_{\alpha \gamma}^{\epsilon}C_{\beta \epsilon}^{\gamma})
\frac{\partial \C{F}}{\partial \mu_{\alpha}}
\frac{\partial \C{H}}{\partial \mu_{\beta}}
+ 
(-R_{a \alpha }^{d}D_{bd}^{a}
+L_{a \alpha}^{\beta}R_{b \beta}^{a}
+C_{\alpha \beta}^{\gamma}L_{\gamma b}^{\beta})
\frac{\partial \C{F}}{\partial \mu_{\alpha}}
\frac{\partial \C{H}}{\partial \nu_{b}}
\\
&
+
(R_{\gamma a}^{d}L_{d \beta}^{\gamma}+L_{\alpha a }^{\gamma}C_{ \beta \gamma}^{\alpha}
-D_{ ab}^{d}R_{ d \beta}^{b})
\frac{\partial \C{F}}{\partial \nu_{a}}
\frac{\partial \C{H}}{\partial \mu_{\beta}}
+
(L_{\alpha a }^{\beta}L_{\beta b}^{\alpha}
+D_{ad}^{k}D_{b k}^{d})
\frac{\partial \C{F}}{\partial \nu_{a}}
\frac{\partial \C{H}}{\partial \nu_{b}}.
\end{split}
\end{equation}

According to this formulation (\ref{Gen-Killing}), for a functional $\mathcal{S}$ the equation of motion will be
\begin{equation} \label{eqofmoofcart}
\dot{\mu}_\beta = \bar{\mathcal{G}}_{\beta a}\frac{\partial S}{\partial \nu_a}+\bar{\mathcal{G}}_{\beta \alpha}\frac{\partial S}{\partial \mu_\alpha}, \qquad 
\dot{\nu}_d = \bar{\mathcal{G}}_{d \beta}\frac{\partial S}{\partial \mu_\beta}+\bar{\mathcal{G}}_{d a}\frac{\partial S}{\partial \nu_a}.
\end{equation}
We substitute the explicit represetations of the metric \eqref{CM-MP} into the system \eqref{eqofmoofcart} and arrive at that
\begin{equation} \label{eqofmoofcart-gen}
\begin{split}
\dot{\mu}_\beta &= (-R_{a \beta }^{d}D_{bd}^{a}
+L_{a \beta}^{\alpha}R_{b \alpha}^{a}
+C_{\beta \alpha}^{\gamma}L_{\gamma b}^{\alpha}) \frac{\partial S}{\partial \nu_b}
+
(R_{\alpha a}^{b}R_{ \beta b}^{a}+C_{\alpha \gamma}^{\epsilon}C_{\beta \epsilon}^{\gamma})
\frac{\partial S}{\partial \mu_\beta},
\\
\dot{\nu}_d &= 
(R_{\gamma d}^{a}L_{a \beta}^{\gamma}+L_{\alpha d }^{\gamma}C_{ \beta \gamma}^{\alpha}
-D_{ db}^{a}R_{ a \beta}^{b}
)
\frac{\partial S}{\partial \mu_\beta}
+(
L_{\alpha a }^{\beta}L_{\beta d}^{\alpha}
+D_{ab}^{k}D_{d k}^{b})
\frac{\partial S}{\partial \nu_a}.
\end{split}
\end{equation}

\subsection{Matched Casimir Dissipation Bracket}

Recall the Casimir dissipation bracket given in \eqref{ccs}. In order to carry this discussion to coupled systems on $\mathfrak{g}^*\times \mathfrak{h}^*$, as it can deduced from that equation, we first need to determine a real valued bilinear operator on the matched pair Lie algebra $\mathfrak{g}\bowtie \mathfrak{h}$ and then employ the pairing equipped with a Casimir function(al). Let us first determine dissipations individually on $\mathfrak{g}^*$  and $\mathfrak{h}^*$ then couple them. 

Consider symmetric bilinear operators on $\mathfrak{g}$ ve $\mathfrak{h}$, denoted by 
$\psi $ and $\vartheta $, respectively. Assume that  $\C{C}$ is a Casimir function $%
\mathfrak{g}^{\ast }$ and $\mathcal{D}$ is a Casimir function on 
$\mathfrak{h}^{\ast }$. Then, Casimir dissipation brackets are
\begin{equation}\label{casimir}
\begin{split}
(\mathcal{F},\mathcal{H})^{(CD)}_{\mathfrak{g}^*}\left( \mu  \right) & =-\psi \left( \left[
\frac{\delta \mathcal{F}}{\delta \mu },\frac{\delta \mathcal{H}}{\delta \mu }\right] _{\mathfrak{g}},\left[ \frac{\delta \mathcal{C}}{\delta \mu },\frac{\delta \mathcal{H}}{\delta \mu }\right]
_{\mathfrak{g}}\right) \\
(\mathcal{F},\mathcal{H})^{(CD)}_{\mathfrak{h}^*}\left( \nu \right) &=-\vartheta \left( %
\left[ \frac{\delta \mathcal{F}}{\delta \nu },\frac{\delta \mathcal{H}}{\delta \nu }\right] _{ 
	\mathfrak{h}},\left[ \frac{\delta \mathcal{D}}{\delta \nu },\frac{\delta \mathcal{H}}{\delta \nu 
}\right] _{\mathfrak{h}}\right).
\end{split}
\end{equation}
Define a real valued symmetric bilinear operator on $\mathfrak{g}\oplus \mathfrak{h}$, using $\psi $ and $\vartheta $, as follows
\begin{equation}\label{sym-opp}
\left( \psi ,\vartheta \right) :(\mathfrak{g}\bowtie \mathfrak{h})\times 
(\mathfrak{g}\bowtie \mathfrak{h})\longrightarrow \mathbb{R}, \qquad (  \xi \oplus \eta ,\xi'\oplus \eta' ) \mapsto
\psi \left( \xi ,\xi' \right) +\vartheta \left( \eta ,\eta' \right) .
\end{equation}
In terms of the Casimir functions $\C{C}$ and $\C{D} $ on $\mathfrak{g}^*$ and $\mathfrak{h}^*$, respectively, we define a Casimir function $( \C{C},\C{D} )$ on $\mathfrak{g}^*\times \mathfrak{h}^*$, for example, as follows
\begin{equation}
( \C{C},\C{D} )(\mu,\nu)=\C{C}(\mu)+ \C{D}(\nu). 
\end{equation} 
So that, matched Casimir dissipation bracket is defined to be 
\begin{equation}\label{mCD-Bra}
\begin{split}
(\mathcal{F},\mathcal{H})^{(mCD)}(\mu\oplus\nu)&=-\left( \psi
,\vartheta \right) \left( \left[ \left( \frac{\delta \mathcal{F}}{\delta \mu }\oplus\frac{
	\delta \mathcal{F}}{\delta \nu }\right) ,\left( \frac{\delta \mathcal{H}}{\delta \mu }\oplus\frac{
	\delta \mathcal{H}}{\delta \nu }\right) \right]_{\bowtie} 
	,
	\left[ \left( \frac{\delta \mathcal{C}}{\delta
	\mu }\oplus \frac{\delta \mathcal{D}}{\delta \nu }\right) ,\left( \frac{\delta \mathcal{H}}{\delta \mu 
}\oplus \frac{\delta \mathcal{H}}{\delta \nu }\right) \right]_{\bowtie} \right)
\end{split}
\end{equation}
where $[\bullet,\bullet]_{\bowtie}$ is the matched Lie algebra bracket in \eqref{mpla}. Referring to this bracket, the dissipative dynamics \eqref{diss-dyn-eq} for $a=1$, generated by a functional $\mathcal{H}$, is a system of equations. The dynamics on $\G{g}^*$ is 
\begin{equation}\label{mCD-eq-1}
\begin{split}
	\dot{\mu} &=-ad_{\frac{\delta \mathcal{H}}{\delta \mu }}^{\ast }\big[ \frac{\delta \mathcal{C}}{\delta
		\mu },\frac{\delta \mathcal{H}}{\delta \mu }\big] ^{\flat}-ad_{\frac{\delta \mathcal{H}}{\delta
			\mu }}^{\ast }\big( \frac{\delta \mathcal{D}}{\delta \nu }\vartriangleright \frac{
		\delta \mathcal{H}}{\delta \mu }\big) ^{\flat}+ad_{\frac{\delta \mathcal{H}}{\delta \mu }}^{\ast
	}\big( \frac{\delta \mathcal{H}}{\delta \nu }\vartriangleright \frac{\delta \mathcal{C}}{\delta
		\mu }\big) ^{\flat} -\big[\frac{\delta \mathcal{C}}{\delta \mu },\frac{\delta \mathcal{H}}{\delta \mu }\big]^{\flat}\overset{
		\ast }{\vartriangleleft }\frac{\delta \mathcal{H}}{\delta \nu }\\
	&-\big[\frac{\delta \mathcal{D}}{
		\delta \nu }\vartriangleright \frac{\delta \mathcal{H}}{\delta \mu }\big]^{\flat}\overset{\ast 
	}{\vartriangleleft }\frac{\delta \mathcal{H}}{\delta \nu }+\big[\frac{\delta \mathcal{H}}{\delta \nu 
	}\vartriangleright \frac{\delta \mathcal{C}}{\delta \mu }\big]^{\flat}\overset{\ast }{
		\vartriangleleft }\frac{\delta \mathcal{H}}{\delta \nu } -\mathfrak{a}_{\frac{\delta \mathcal{H}}{\delta \nu }}^{\ast }\big[\frac{\delta \mathcal{D}}{
		\delta \nu },\frac{\delta \mathcal{H}}{\delta \nu }\big]^{\flat}-\mathfrak{a}_{\frac{\delta \mathcal{H}}{
			\delta \nu }}^{\ast } \big( \frac{\delta \mathcal{D}}{\delta \nu }\vartriangleleft 
	\frac{\delta \mathcal{H}}{\delta \mu }\big) ^{\flat}-\mathfrak{a}_{\frac{\delta \mathcal{H}}{
			\delta \nu }}^{\ast }\big( \frac{\delta \mathcal{H}}{\delta \nu }\vartriangleleft 
	\frac{\delta \mathcal{C}}{\delta \mu }\big) ^{\flat}
\end{split}
\end{equation}
whereas the dynamics on $\G{h}^*$ is 
\begin{equation}\label{mCD-eq-2}
\begin{split}
	\dot{\nu} &=-ad_{\frac{\delta \mathcal{H}}{\delta \nu }}^{\ast }\big[\frac{\delta \mathcal{D}}{\delta \nu }, 
	\frac{\delta \mathcal{H}}{\delta \nu }\big]^{\flat}-ad_{\frac{\delta \mathcal{H}}{\delta \nu }}^{\ast
	}\big( \frac{\delta \mathcal{D}}{\delta \nu }\vartriangleleft \frac{\delta \mathcal{H}}{\delta
		\mu }\big) ^{\flat}+ad_{\frac{\delta \mathcal{H}}{\delta \nu }}^{\ast }\big( \frac{
		\delta \mathcal{H}}{\delta \nu }\vartriangleleft \frac{\delta \mathcal{C}}{\delta \mu }\big)
	^{\flat} -\frac{\delta \mathcal{H}}{\delta \mu }\overset{\ast }{\vartriangleright }[\frac{
		\delta \mathcal{D}}{\delta \nu },\frac{\delta \mathcal{H}}{\delta \nu }]^{\flat}
		\\
	&-\frac{\delta \mathcal{H}}{
		\delta \mu }\overset{\ast }{\vartriangleright }\big[\frac{\delta \mathcal{D}}{\delta \nu }
	\vartriangleleft \frac{\delta \mathcal{H}}{\delta \mu }\big]^{\flat}+\frac{\delta \mathcal{H}}{\delta
		\mu }\overset{\ast }{\vartriangleright }\big[\frac{\delta \mathcal{H}}{\delta \nu }
	\vartriangleleft \frac{\delta \mathcal{C}}{\delta \mu }\big]^{\flat} -\mathfrak{b}_{\frac{\delta \mathcal{H}}{\delta \mu }}^{\ast }\big[\frac{\delta \mathcal{C}}{
		\delta \mu },\frac{\delta \mathcal{H}}{\delta \mu }\big]^{\flat}-\mathfrak{b}_{\frac{\delta \mathcal{H}}{
			\delta \mu }}^{\ast }\big( \frac{\delta \mathcal{D}}{\delta \nu }\vartriangleright 
	\frac{\delta \mathcal{H}}{\delta \mu }\big) ^{\flat}+\mathfrak{b}_{\frac{\delta \mathcal{H}}{
			\delta \mu }}^{\ast }\big( \frac{\delta \mathcal{H}}{\delta \nu }\vartriangleright 
	\frac{\delta \mathcal{C}}{\delta \mu }\big) ^{\flat}.
\end{split}
\end{equation}
Here, the right action $\vartriangleright$ and the left action $\vartriangleleft$ are those in \eqref{matched-pair-mutual-actions} whereas $\overset{\ast }{	\vartriangleright}$ and $\overset{\ast }{\vartriangleleft }$ are the dual actions in \eqref{eta-star} and \eqref{xi-star}, respectively. Notice that dual operators $\mathfrak{a}^*$ is in \eqref{a*}
, and $\mathfrak{b}^*$ is in \eqref{b*}. Here, superscript $\flat$ denotes the dualization obtained through for the symmetric operators $\psi$ and $\vartheta$ given by 
\begin{equation}
\langle \xi ^{\flat },\xi'
\rangle =\psi ( \xi ,\xi'), \qquad 
\langle \eta ^{\flat },\eta'
\rangle =\vartheta ( \eta ,\eta').
\end{equation}
In order to prevent to notation inflation, we denote these two mappings by the same notation as we did while denoting the Lie algebra brackets on $\G{g}$ and $\G{h}$.  

We can couple the matched irreversible motion, that is matched Casimir dissipation motion, in \eqref{mCD-eq-1} and \eqref{mCD-eq-2} with matched reversible motion, that is matched Lie-Poisson dynamics, in \eqref{LPEgh}. This results with matched pair of metriplectic system involving Casimir dissipation terms, and simple achieved by adding the right hand sides of the systems obeying the order. this collective motion, can be determined by a single matched metriplectic bracket 
\begin{equation*}
\lbrack \left\vert \mathcal{F},\mathcal{H}\right\vert ]_{\bowtie,CD}= \{\mathcal{F},\mathcal{H}\}_{\bowtie}+a(\mathcal{F},\mathcal{H})^{(mCD)},
\end{equation*}
where $\{\bullet,\bullet\}_{\bowtie}$ is the matched Lie-Poisson bracket in \eqref{LiePoissonongh} and $(\bullet,\bullet)^{(mCD)}$ is the matched Casimir dissipation bracket in \eqref{mCD-Bra}. 
In this case, the dynamics governed by a Hamiltonian function(al) $\C{H}$, is implicitly written by 
\begin{equation*}
(\dot{\mu}\oplus \dot{\nu})=\lbrack \left\vert (\mu\oplus\nu),\mathcal{H}\right\vert ]_{\bowtie,CD}.
\end{equation*}
\subsection{Matched Hamilton Dissipation Bracket} 

At first recall the Hamilton dissipation bracket given in \eqref{HD-sym} and the pure irreversible motion in \eqref{HDB-Dyn}. In this subsection, we couple (match) two Hamilton dissipation bracket in form \eqref{HD-sym} and two pure irreversible motion in \eqref{HDB-Dyn}. Accordingly, obeying the notation presenting in the presvious subsection we introduce the following Hamilton dissipations brackets on the constitutive spaces $\mathfrak{g}^*$ and $\mathfrak{h}^*$, for two bilinear operators $\psi$ and $\vartheta$, as follows
\begin{equation}
\begin{split}
(\mathcal{F},\mathcal{H})^{HD}_{\mathfrak{g}}\left( \mu  \right) & =-\psi \left( \left[
\frac{\delta \mathcal{F}}{\delta \mu },\frac{\delta \mathcal{C}}{\delta \mu }\right]_{\mathfrak{g}},\left[ \frac{\delta \mathcal{H}}{\delta \mu },\frac{\delta \mathcal{C}}{\delta \mu }\right]
_{\mathfrak{g}}\right) \\ 
(\mathcal{F},\mathcal{H})^{HD}_{\mathfrak{h}}\left( \nu \right) & =-\vartheta \left( %
\left[ \frac{\delta \mathcal{F}}{\delta \nu },\frac{\delta \mathcal{D}}{\delta \nu }\right] _{ 
	\mathfrak{h}},\left[ \frac{\delta \mathcal{H}}{\delta \nu },\frac{\delta \mathcal{D}}{\delta \nu 
}\right] _{\mathfrak{h}}\right).
\end{split}
\end{equation}
where $\mathcal{C}$ and $\mathcal{D}$ are Casimir functions on $\mathfrak{g}^*$ and $\mathfrak{h}^*$, respectively. In order to match these symmetric brackets, recall the real valued bilinear map \eqref{sym-opp} defined on matched pair Lie algebra $\mathfrak{g} \bowtie \mathfrak{h}$. Then, introduce matched Hamilton dissipation bracket
\begin{equation} \label{macthedhamilton}
(\mathcal{F},\mathcal{H})^{(mHD)}(\mu\oplus \nu)=-\left( \psi
,\vartheta \right) \left( \left[ \left( \frac{\delta \mathcal{F}}{\delta \mu }\oplus\frac{
	\delta \mathcal{F}}{\delta \nu }\right) ,\left( \frac{\delta \mathcal{C}}{\delta \mu }\oplus\frac{
	\delta \mathcal{D}}{\delta \nu }\right) \right] ,\left[ \left( \frac{\delta \mathcal{H}}{\delta
	\mu }\oplus\frac{\delta \mathcal{H}}{\delta \nu }\right) ,\left( \frac{\delta \mathcal{C}}{\delta \mu 
}\oplus\frac{\delta \mathcal{D}}{\delta \nu }\right) \right] \right),
\end{equation}
where the brackets inside the paring are the matched Lie bracket in \eqref{mpla}. 
Irreversible the dynamics on $\mathfrak{g}^*\times \mathfrak{h}^*$ can be obtained as 
\begin{equation}
\begin{split}
	\dot{\mu} &=-ad_{\frac{\delta \mathcal{C}}{\delta \mu }}^{\ast }\left[ \frac{\delta \mathcal{H}}{\delta
		\mu },\frac{\delta \mathcal{C}}{\delta \mu }\right] ^{\flat}-ad_{\frac{\delta \mathcal{C}}{\delta
			\mu }}^{\ast }\left( \frac{\delta \mathcal{H}}{\delta \nu }\vartriangleright \frac{
		\delta \mathcal{C}}{\delta \mu }\right) ^{\flat}+ad_{\frac{\delta \mathcal{C}}{\delta \mu }}^{\ast
	}\left( \frac{\delta \mathcal{D}}{\delta \nu }\vartriangleright \frac{\delta \mathcal{H}}{\delta
		\mu }\right) ^{\flat} -[\frac{\delta \mathcal{H}}{\delta \mu },\frac{\delta \mathcal{C}}{\delta \mu }]^{\flat}\overset{
		\ast }{\vartriangleleft }\frac{\delta \mathcal{D}}{\delta \nu }
		\\
	&-[\frac{\delta \mathcal{H}}{
		\delta \nu }\vartriangleright \frac{\delta \mathcal{C}}{\delta \mu }]^{\flat}\overset{\ast 
	}{\vartriangleleft }\frac{\delta \mathcal{D}}{\delta \nu }+[\frac{\delta \mathcal{D}}{\delta \nu 
	}\vartriangleright \frac{\delta \mathcal{H}}{\delta \mu }]^{\flat}\overset{\ast }{
		\vartriangleleft }\frac{\delta \mathcal{D}}{\delta \nu } -\mathfrak{a}_{\frac{\delta \mathcal{D}}{\delta \nu }}^{\ast }[\frac{\delta \mathcal{H}}{
		\delta \nu },\frac{\delta \mathcal{D}}{\delta \nu }]^{\flat}-\mathfrak{a}_{\frac{\delta \mathcal{D}}{
			\delta \nu }}^{\ast }\left( \frac{\delta \mathcal{H}}{\delta \nu }\vartriangleleft 
	\frac{\delta \mathcal{C}}{\delta \mu }\right) ^{\flat}-\mathfrak{a}_{\frac{\delta \mathcal{D}}{
			\delta \nu }}^{\ast }\left( \frac{\delta \mathcal{D}}{\delta \nu }\vartriangleleft 
	\frac{\delta \mathcal{H}}{\delta \mu }\right) ^{\flat}
\\
	\dot{\nu} &=-ad_{\frac{\delta \mathcal{D}}{\delta \nu }}^{\ast }[\frac{\delta \mathcal{H}}{\delta \nu }, 
	\frac{\delta \mathcal{D}}{\delta \nu }]^{\flat}-ad_{\frac{\delta \mathcal{D}}{\delta \nu }}^{\ast
	}\left( \frac{\delta \mathcal{H}}{\delta \nu }\vartriangleleft \frac{\delta \mathcal{C}}{\delta
		\mu }\right) ^{\flat}+ad_{\frac{\delta \mathcal{D}}{\delta \nu }}^{\ast }\left( \frac{
		\delta \mathcal{D}}{\delta \nu }\vartriangleleft \frac{\delta \mathcal{H}}{\delta \mu }\right)
	^{\flat} -\frac{\delta \mathcal{C}}{\delta \mu }\overset{\ast }{\vartriangleright }[\frac{
		\delta \mathcal{H}}{\delta \nu },\frac{\delta \mathcal{D}}{\delta \nu }]^{\flat}\\
	&-\frac{\delta \mathcal{C}}{
		\delta \mu }\overset{\ast }{\vartriangleright }[\frac{\delta \mathcal{H}}{\delta \nu }
	\vartriangleleft \frac{\delta \mathcal{C}}{\delta \mu }]^{\flat}+\frac{\delta \mathcal{C}}{\delta
		\mu }\overset{\ast }{\vartriangleright }[\frac{\delta \mathcal{D}}{\delta \nu }
	\vartriangleleft \frac{\delta \mathcal{H}}{\delta \mu }]^{\flat} -\mathfrak{b}_{\frac{\delta \mathcal{C}}{\delta \mu }}^{\ast }[\frac{\delta \mathcal{H}}{
		\delta \mu },\frac{\delta \mathcal{C}}{\delta \mu }]^{\flat}-\mathfrak{b}_{\frac{\delta \mathcal{C}}{
			\delta \mu }}^{\ast }\left( \frac{\delta \mathcal{H}}{\delta \nu }\vartriangleright 
	\frac{\delta \mathcal{C}}{\delta \mu }\right) ^{\flat}+\mathfrak{b}_{\frac{\delta \mathcal{C}}{
			\delta \mu }}^{\ast }\left( \frac{\delta \mathcal{D}}{\delta \nu }\vartriangleright 
	\frac{\delta \mathcal{H}}{\delta \mu }\right) ^{\flat}.
\end{split}
\end{equation}
Here, the right action $\vartriangleright$ and the left action $\vartriangleleft$ are those in \eqref{matched-pair-mutual-actions} whereas $\overset{\ast }{	\vartriangleright}$ and $\overset{\ast }{\vartriangleleft }$ are the dual actions in \eqref{eta-star} and \eqref{xi-star}, respectively and $\mathfrak{a}$, $\mathfrak{a}^*$ defined as in \eqref{a}, \eqref{a*} and $\mathfrak{b}$, $\mathfrak{b}^*$ defined as in \eqref{b}, \eqref{b*}.

\section{Illustration: Heisenberg Algebras in Mutual Actions}\label{Sec-3D}
\subsection{Heisenberg Algebra and Lie-Poisson Dynamics}
We start with a $3$ dimensional Heisenberg algebra which we denote by  $\mathfrak{g}$, see \cite{Ma20}. We assume a basis $\{\mathbf{e}_1,\mathbf{e}_2,\mathbf{e}_3\}$ for $\G{g}$, and define the Lie algebra operation as 
\begin{equation}\label{Hei-brack}
[\mathbf{e}_1,\mathbf{e}_3]=0, \qquad [\mathbf{e}_1,\mathbf{e}_2]=\mathbf{e}_3, \qquad [\mathbf{e}_2,\mathbf{e}_3]=0 .
\end{equation} 
These read that the structure constants as $C_{12}^3=-C_{21}^3=1$ while the rest is zero. 
Heisenberg algebra can be written as a 2-cocycle extension Lie algebra. So that we can examine it in the realms of the discussions done in Subsection \ref{2coc-Sec}. To see this, referring to the basis of the algebra  $\mathfrak{g}$, we define two linear spaces $V=\langle\mathbf{e}_3\rangle$ and $\mathfrak{l}=\langle\mathbf{e}_1,\mathbf{e}_2\rangle$. Accordingly, we introduce a $V$-valued skew-symmetric bilinear mapping on $\mathfrak{l}$ as  
\begin{equation}
\varphi: \mathfrak{l} \times \mathfrak{l} \longrightarrow V, \qquad \varphi(\mathbf{e}_1,\mathbf{e}_1)=0, \quad \varphi(\mathbf{e}_2,\mathbf{e}_2)=0, \quad \varphi(\mathbf{e}_1,\mathbf{e}_2)=\mathbf{e}_3.
\end{equation}
It is straightforward to verify that $\varphi$ is a 2-cocycle. Taking the left action of $\G{l}$ on $V$ (see the first action in the list \eqref{actionsoflk}) is zero, we arrive easily that the bracket \eqref{Hei-brack} is indeed in a 2-cocycle extension form \eqref{AAAAA}. 

\textbf{Coadjoint flow.}
The dual space is denoted by $\mathfrak{g}^*$ with the dual basis $\{ \textbf{e}^1,\textbf{e}^2,\textbf{e}^3\}$. Assume the following coordinates $\xi=(\xi_1,\xi_2,\xi_3)$ in $\G{g}$, and  $\mu=(\mu_1,\mu_2,\mu_3)$ in $\G{g}^*$. Then, the coadjoint action of a Lie algebra element $\xi$ in $\mathfrak{g}$ to a dual element $\mu$ in $\mathfrak{g}^*$ is computed to be  
\begin{equation}
ad^*:\mathfrak{g}\times \mathfrak{g}^*\longrightarrow \mathfrak{g}^*, \qquad  ad^*_\xi \mu=(\mu_3\xi^2,-\mu_3\xi^1,0).
\end{equation}
Referring to this calculation, we write the Lie-Poisson dynamics \eqref{loc-LP-Eqn} generated by a Hamiltonian function $\mathcal{H}$ as follows
\begin{equation} \label{Ray-mot}
\dot{\mu}_1=\mu_3 \frac{\partial \mathcal{H}}{\partial \mu_2},  \qquad \dot{\mu}_2=-\mu_3 \frac{\partial \mathcal{H}}{\partial \mu_1} \qquad \dot{\mu}_3=0.
\end{equation} 
Here, the latter gives that $\mu_3$ is a constant. In equation (\ref{Ray-mot}), if we choose $\mu_1=q$, $\mu_2=p$, and $\mu_3=1$, then we arrive at the Hamilton's equations in its very classical form
\begin{equation}
	\dot{q}= \frac{\partial \mathcal{H}}{\partial p}, \qquad \dot{p}=-\frac{\partial\mathcal{H}}{\partial q}.
\end{equation}
Therefore we claim that, in the present geometry, the Hamiltonian dynamics can be realized as a coadjoint flow.

\textbf{Double Bracket dissipation.} For the present case, the symmetric double bracket in \eqref{doubledissi} is computed to be
\begin{equation}
(\C{F},\C{S})^{(D)}(\mu)=\mu_3^2\big( \frac{\partial \mathcal{S}}{\partial \mu_1}\frac{\partial \mathcal{F}}{\partial \mu_1}
+
 \frac{\partial \mathcal{S}}{\partial \mu_2}\frac{\partial \mathcal{F}}{\partial \mu_2}\big). 
\end{equation}
Therefore for a function $\mathcal{S}$, the irreversible dynamics due to the symmetric bracket is computed to be
 \begin{equation}
  \dot{\mu}_1=\mu_3^2 \frac{\partial \mathcal{S}}{\partial \mu_1},  \qquad  \dot{\mu}_2=\mu_3^2\frac{\partial \mathcal{S}}{\partial\mu_2},  \qquad \dot{\mu}_3=0.	
 \end{equation}
 Then the metriplectic equations of motion (\ref{aaaa}) are computed to be
 \begin{equation} \label{MD-Ex-1}
 \dot{\mu}_1=\mu_3 \frac{\partial \mathcal{H}}{\partial \mu_2}+\mu_3^2\frac{\partial\mathcal{S}}{\partial \mu_1}, 
 \qquad \dot{\mu}_2=-\mu_3 \frac{\partial \mathcal{H}}{\partial \mu_1}+\mu_3^2\frac{\partial \mathcal{S}}{\partial\mu_2},\qquad 
 \dot{\mu}_3=0.	
 \end{equation}
If we choose $\mu_1=q$, $\mu_2=p$, and $\mu_3=1$, then  the metriplectic dynamics \eqref{MD-Ex-1} turns out to be
  \begin{equation} \label{qdot-1}
\dot{q}=\frac{\partial \mathcal{H}}{\partial p}+ \frac{\partial \mathcal{S}}{\delta q} , \qquad  \dot{p}=-\frac{\partial \mathcal{H}}{\partial q}+ \frac{\partial \mathcal{S}}{\partial p}.
  \end{equation}
Two interesting oarticular instances of the present dynamics:

\textbf{(1)} Let us take the Hamiltonian function $\mathcal{H}=p^2+V(q)$ to be the total energy of the system and $S=S(q)$ then the system \eqref{qdot-1} reduces to 
  \begin{equation}\label{2ndODE-1}
  	\ddot{q}-S_{qq} \dot{q}-V_q=0.
  	 \end{equation}
We cite \cite{Mielke2011} for a more elegant geometrization of the second order ODE \eqref{2ndODE-1}  in terms of the GENERIC framework. \textbf{(2)}  As another naive application of the dissipative system \eqref{qdot-1}, we consider a general Hamiltonian function $H$ but choose $S=ap^2/2$ for a scalar $a$, then a fairly straight-forward calculation gives \eqref{qdot-1} that 
   \begin{equation}\label{dyn-exp}
\dot{q}=\frac{\partial \mathcal{H}}{\partial p}, \qquad \dot{p}=-ap+\frac{\partial \mathcal{H}}{\partial q}.
  \end{equation}
This is the conformal Hamiltonian dynamics as described in \cite{Perl}. To see the geometry behind this dynamics, at first consider the vector field $X=(\partial \mathcal{H}/\partial p) \partial_q+(-ap+{\partial \mathcal{H}}/\partial q )\partial_ p$ generating \eqref{dyn-exp} and then define the symplectic two-form $\Omega=dq \wedge dp$. A Hamiltonin vector field preserves the symplectic two-forms but the vector field $X$ satisfies 
  \begin{equation}
  \mathfrak{L}_X \Omega= d\big(d\mathcal{H}-apdq\big) = adq \wedge dp=a \Omega
  \end{equation}	 
  where $\mathfrak{L}$ denotes the Lie derivative.
This reads that $X$ preserved the symplectic two-form up to some conformal factor $a$. 

\subsection{Coupling of Two Heisenberg Algebras and Matched Lie-Poisson Dynamics}

Consider 
two $3D$ Heisenberg algebras $\mathfrak{g}$ and $\mathfrak{h}$. We choose a basis 
\begin{equation}\label{Hei-basis}
\{ \mathbf{e}_1,\mathbf{e}_2,\mathbf{e}_3\}\in \mathfrak{g}, \qquad \{ \mathbf{f}_1,\mathbf{f}_2,\mathbf{f}_3\}\in \mathfrak{h}.
\end{equation}
To fix the notation, we record here the coordinate realizations of two arbitrary elements $\xi$ and $\xi'$ in $\mathfrak{g}$, and two arbitrary elements $\eta$ and $\eta'$ in $\mathfrak{h}$ as follows
\begin{equation} \label{elementsV+W}
\begin{split}
&\xi=\xi^1 \mathbf{e}_1+ \xi^2 \mathbf{e}_2+ \xi^3 \mathbf{e}_3, \qquad \xi'=\xi^{1'} \mathbf{e}_1+ \xi^{2'} \mathbf{e}_2+ \xi^{3'} \mathbf{e}_3, \\
&\eta=\eta^1 \mathbf{f}_1+ \eta^2 \mathbf{f}_2+ \eta^3 \mathbf{f}_3, \qquad \eta'=\eta^{1 '} \mathbf{f}_1+\eta^{2'} \mathbf{f}_2+ \eta^{3'} \mathbf{f}_3.
\end{split}
\end{equation} 
In accordance with these choices, the Lie algebra brackets on $\mathfrak{g}$ and $\mathfrak{h}$ can be exhibited in the form
\begin{equation}\label{bracketsonVW}
	[\xi,\xi']=(\xi^1 \xi^{2'} -\xi^{1 '} \xi^2 )\mathbf{e}_3, \qquad {[\eta,\eta']=(\eta^1 \eta^{2'}-\eta^{1'} \eta^2)\mathbf{f}_3}
\end{equation} 
respectively. 
Obeying the notation in \eqref{Lieact}, and referring to the coordinate realizations \eqref{elementsV+W}, we introduce a right action of $\mathfrak{g}$ on $\mathfrak{h}$, and a left action $\mathfrak{h}$ on $\mathfrak{g}$ as \cite{Ma20}
\begin{equation} \label{actionofHeisenberg}
\begin{split}
\vartriangleright &:\mathfrak{h}\otimes \mathfrak{g}\rightarrow \mathfrak{g}, \qquad \eta \vartriangleright \xi=-\eta^1 \xi^2 \mathbf{e}_3 	,\\
\vartriangleleft &:\mathfrak{h}\otimes \mathfrak{g}\rightarrow \mathfrak{g}, \qquad \eta \vartriangleleft \xi=-\xi^{1}\eta^2\mathbf{f}_3	
	.
	\end{split}
\end{equation}
This reads that the action constants $R_{12}^3=-1$,  $L_{12}^3=-1$ while the rest are zero. It is straight forward to prove that these actions are indeed satisfying the matched pair Lie algebra conditions in \eqref{comp-mpa}. 
Accordingly, for the present coupling, the matched pair Lie algebra bracket \eqref{mpla} on $\G{g}\bowtie \G{h}$ is computed to be
\begin{equation} \label{heisliebrac}
	[(\xi,\eta),(\xi',\eta')]_{\bowtie}=(\xi^1\xi^{2'}-\xi^2 \xi^{1'}-\eta^1\xi^{2'}+\eta^{1'} \xi^2)\mathbf{e}_3 \oplus (\eta^1\eta^{2'}-\eta^2 \eta^{1'}-\eta^2\xi^{1'}+\eta^{2'} \xi^1)\mathbf{f}_3.
\end{equation}

\textbf{As coupling of two cocycle extensions.} Referring to the basis \eqref{Hei-basis},  Heisenberg algebras $\mathfrak{g}$ and $\mathfrak{h}$ can be realized as 2-cocycle extensions 
\begin{equation} \label{decomp-Hei}
V{_\varphi\rtimes}\mathfrak{l} :=\langle \mathbf{e_3} \rangle {_\varphi\rtimes} \langle \mathbf{e_1},\mathbf{e_2} \rangle, \qquad
W {_\phi\rtimes} \mathfrak{k}:=\langle \mathbf{f_3} \rangle{_\phi\rtimes} \langle \mathbf{f_1},\mathbf{f_2} \rangle,
\end{equation} 
in the light of the following 2-cocycles 
\begin{equation}  \label{eps-iota}
\begin{split}
\varphi&:\mathfrak{l} \times  \mathfrak{l} \longrightarrow V, \qquad 
\varphi(\mathbf{e}_1,\mathbf{e}_1)=0, \quad \varphi(\mathbf{e}_2,\mathbf{e}_2)=0, \quad \varphi(\mathbf{e}_1,\mathbf{e}_2)=\textbf{e}_3,
\\
\phi&:\mathfrak{k} \times  \mathfrak{k} \longrightarrow W, \qquad
\phi(\mathbf{f}_1,\mathbf{f}_1)=0, \quad \phi(\mathbf{f}_2,\mathbf{f}_2)=0, \quad \phi(\mathbf{f}_1,\mathbf{f}_2)=\mathbf{f}_3
\end{split}
\end{equation}
while the left actions in \eqref{actionsoflk} are zero. 
We now examine Proposition \ref{Prop-Gokhan} for the matched pair of two Heisenberg algebras. That is, we show that $\G{g}\bowtie \G{h}$ is a 2-cocycle extension by itself. For this, following the notation in Subsection \ref{centrallyextend}, we take that both  \textbf{(1)} mutual actions $\blacktriangleleft$ and $\blacktriangleright$ of $\mathfrak{l}$ and $\mathfrak{k}$ on each other exhibited in \eqref{Lieact-k-l}, \textbf{(2)} cross actions $\curvearrowleft$ and $\curvearrowright$ in  \eqref{actionsofkl} are all zero mappings whereas \textbf{(3)} mappings $\epsilon$ and $\iota$ given in \eqref{epsiot} are computed to be 
\begin{equation}
\epsilon(\mathbf{f_1},\mathbf{e_2})=-\mathbf{e_3} \qquad	\iota(\mathbf{f_2},\mathbf{e_1})=-\mathbf{f_3},
\end{equation}
while the rest are zero. Referring to these choices, it is now straight forward to observe that both the left action $\vartriangleright$ and the right action $\vartriangleleft$ of the Heisenberg algebras given in the display \eqref{actionofHeisenberg} can be recasted in the form of 
\eqref{left-comp}. This reads that the matched Lie algebra bracket \eqref{heisliebrac} on Heisenberg algebras is indeed fitting \eqref{bracketlk-1}. 

On the other, now we show that the matched pair $\G{g}\bowtie \G{h}$ is a 2-cocycle extension of $\G{l}\bowtie \G{k}=\langle \mathbf{e}_1, \mathbf{e}_2\rangle \bowtie \langle \mathbf{f}_1, \mathbf{f}_2\rangle $ over its representation space $V\oplus W:=\langle \mathbf{e}_3 \rangle \oplus \langle \mathbf{f}_3 \rangle $ according to the decomposition in \eqref{decomp-Hei}. We remark here that, the mutual actions in $\G{l}\bowtie \G{k}$ are trivial since we take that $\blacktriangleleft$ and $\blacktriangleright$ of $\mathfrak{l}$ and $\mathfrak{k}$ as zero mappings. Nevertheless, we determine the left action $\rrbiprod$ of the Lie algebra  $\G{l}\bowtie \G{k}$ onto $V\oplus W$ as trivial. We introduce a cocycle
\begin{equation}
\begin{split}
\Theta&: (\mathfrak{l} \bowtie \mathfrak{k}) \times (\mathfrak{l} \bowtie \mathfrak{k}) \longrightarrow V \oplus W
\\&:\big (\langle \mathbf{e}_1, \mathbf{e}_2\rangle \bowtie \langle \mathbf{f}_1, \mathbf{f}_2\rangle \big ) \times \big (\langle \mathbf{e}_1, \mathbf{e}_2\rangle \bowtie \langle \mathbf{f}_1, \mathbf{f}_2\rangle \big )  \longrightarrow \langle \mathbf{e}_3 \rangle \oplus \langle \mathbf{f}_3 \rangle,
\end{split}
\end{equation}
which, referring to the formula in \eqref{theta}, can be exploited as
\begin{equation} \label{Hei-coco}
\begin{split}
\Theta(\xi^1\mathbf{e}_1+\xi^2\mathbf{e}_2+\eta^1\mathbf{f}_1+\eta^2\mathbf{f}_2,\xi^{1'}\mathbf{e}_1+\xi^{2'}\mathbf{e}_2+\eta^{1'}\mathbf{f}_1+\eta^{2'}\mathbf{f}_2)&=(\xi^1\xi^{2'}-\xi^{1'}\xi^2-\eta^1\xi^{2'}+\eta^{1'}\xi^2)\mathbf{e}_3 \\&\qquad  \oplus (\eta^1\eta^{2'}-\eta^2 \eta^{1'}-\eta^2\xi^{1'}+\eta^{2'} \xi^1)\mathbf{f}_3.
\end{split}
\end{equation}
So, the trivial action $\rrbiprod$ of $\G{l}\bowtie \G{k}$ on $V\oplus W$ with the cocycle \eqref{Hei-coco}, it is immediate to see the matched Lie algebra bracket \eqref{heisliebrac} is also a 2-cocycle extension bracket by obeying \eqref{AAAAA}. To sum up, we can argue that coupling of two Heisenberg algebras is a nontrivial example of Proposition \ref{Prop-Gokhan}.

\textbf{Matched Lie-Poisson equations.} Let us first fix the notation for the dual elements as 
\begin{equation} \label{daul-elements-exp-1}
\mu=\mu_1 \mathbf{e}^1+\mu_2 \mathbf{e}^2+\mu_3 \mathbf{e}^3 \in \mathfrak{g}^*, \qquad \nu=\nu_1 \mathbf{f}^1+\nu_2 \mathbf{f}^2+\nu_3 \mathbf{f}^3 \in \mathfrak{h}^*.
\end{equation}
The dual actions $\overset{\ast}{\vartriangleleft}$ in    (\ref{xi-star}) and $\overset{\ast}{\vartriangleright}$ in ( \ref{eta-star}), the cross actions $\mathfrak{a}^*$ in (\ref{a*}) and $\mathfrak{b}^*$ in (\ref{b*}) are computed to be 
 \begin{equation}\label{dualactionsofheis}
 \mu \overset{\ast}{\vartriangleleft} \eta=-\mu_3\eta^1 \mathbf{e}^2, \qquad \xi \overset{\ast}{\vartriangleright} \nu=-\nu_3 \xi^1 \mathbf{f}^2
, \qquad 
 	\mathfrak{b}^*_\xi \mu=-\mu_3 \xi^2 \mathbf{f}^1, \qquad \mathfrak{a}^*_{\eta} \nu=-\nu_3 \eta^2 \mathbf{e}^1.
 \end{equation}
Then the matched Lie-Poisson equations (\ref{LPEgh}) generated by a Hamiltonian function 
$\mathcal{H}=\mathcal{H}(\mu,\nu)$ on $\mathfrak{g}^* \oplus \mathfrak{h}^*$ is computed to be 
\begin{equation} \label{Ray-mot-1}
\begin{split}
\dot{\mu}_1&=-\nu_3 \frac{\partial \mathcal{H}}{\partial \nu_2}+\mu_3 \frac{\partial \mathcal{H}}{\partial \mu_2},  \hspace{1.2cm} 
\dot{\mu}_2=\mu_3 \frac{\partial \mathcal{H}}{\partial \nu_1}-\mu_3 \frac{\partial \mathcal{H}}{\partial \mu_1}, 
\hspace{1cm}
\dot{\mu}_3=0
\\ 
 \dot{\nu}_1&=-\mu_3 \frac{\partial \mathcal{H}}{\partial \mu_2}+\nu_3 \frac{\partial \mathcal{H}}{\partial \nu_2}, 
 \hspace{1cm} 
 \dot{\nu}_2=-\nu_3 \frac{\partial \mathcal{H}}{\partial \nu_1}+\nu_3 \frac{\partial \mathcal{H}}{\partial \mu_1},
 \hspace{1.1cm} 
 \dot{\nu}_3=0.
\end{split}
\end{equation} 
Let us study some particular instances on this system of equations. In equation (\ref{Ray-mot-1}), if we choose $\mu_1=q$, $\mu_2=p$, $\nu_1=u$, $\nu_2=w$, and $\mu_3=\nu_3=1$ then  we arrive at the Hamilton's equations in a coupled form
\begin{equation}\label{Lie-poissonqp}
\begin{split}
\dot{q}&= \frac{\partial \mathcal{H}}{\partial p}-\frac{\partial \mathcal{H}}{\partial w}, 
\hspace{1.2cm}
\dot{p}=\frac{\partial \mathcal{H}}{\partial u}-\frac{\partial \mathcal{H}}{\partial q},
\\
 \dot{u}&= -\frac{\partial \mathcal{H}}{\partial p}+\frac{\partial \mathcal{H}}{\partial w}, 
\hspace{1cm}
\dot{w}=-\frac{\partial \mathcal{H}}{\partial u}+\frac{\partial \mathcal{H}}{\partial q}.
\end{split}
\end{equation}

\textbf{Rayleigh type dissipation.} For the present case, in order to add  Rayleigh type dissipation term to the Lie-Poisson dynamics on $\mathfrak{g}^* \oplus \mathfrak{h}^*$ as
\begin{equation}
\mathfrak{g}^* \oplus \mathfrak{h}^*  \longrightarrow \mathfrak{g} \bowtie  \mathfrak{h}, \qquad \mu\oplus\nu \quad \mapsto ( \Upsilon^{\G{g}^1}\mathbf{e}_1+\Upsilon^{\G{g}^2}\mathbf{e}_2+\Upsilon^{\G{g}^3}\mathbf{e}_3  ,  \Upsilon^{\G{h}^1}\mathbf{f}_1+\Upsilon^{\G{h}^2}\mathbf{f}_2+ \Upsilon^{\G{h}^3}\mathbf{f}_3 ).
\end{equation}
Referring to the matched Lie-Poisson equation of motions (\ref{eqofmoofRay}) for the dual spaces $\mathfrak{g}^*\oplus \mathfrak{h}^*$  can be obtained as follows
\begin{equation}
\begin{split}
&\dot{\mu}_1-\mu_3\frac{\partial \mathcal{H}}{\partial \mu_2} +\nu_3 \frac{\partial \mathcal{H}}{\partial \nu_2}=\nu_3\Upsilon^{\G{h}^2}-\mu_3\Upsilon^{\G{g}^2}, \qquad \dot{\mu}_2-\mu_3 \frac{\partial \mathcal{H}}{\partial \nu_1}+\mu_3 \frac{\partial \mathcal{H}}{\partial \mu_1}=\mu_3\Upsilon^{\G{g}^1}-\mu_3 \Upsilon^{\G{h}^1}, \qquad \dot{\mu}_3=0,
\\
&\dot{\nu}_1 + \mu_3 \frac{\partial \mathcal{H}}{\partial \mu_2}-\nu_3 \frac{\partial \mathcal{H}}{\partial \nu_2}=\mu_3\Upsilon^{\G{g}^2}-\nu_3\Upsilon^{\G{h}^2},\qquad 
\dot{\nu}_2 +\nu_3 \frac{\partial \mathcal{H}}{\partial \nu_1}-\nu_3 \frac{\partial \mathcal{H}}{\partial \mu_1}=\nu_3\Upsilon^{\G{h}^1}-\nu_3 \Upsilon^{\G{g}^1}, \qquad \dot{\nu}_3=0.
\end{split}
\end{equation}
\textbf{Double Bracket dissipation.} For the present discussion, 
matched double bracket 
\eqref{matchdouble} takes the particular form 
\begin{equation}\label{matchdoublebrofheis}
\begin{split}
(\mathcal{F},\mathcal{H})^{(mD)}(\mu\oplus \nu)=&\mu_3^2 \frac{\partial \mathcal{F}}{\partial \mu_1} \frac{\partial \mathcal{H}}{\partial \mu_1}+(2\mu^2_3+2\nu_3\mu_3+\nu^2_3)\frac{\partial \mathcal{F}}{\partial \mu_2} \frac{\partial \mathcal{H}}{\partial \mu_2}-(\mu_3+\nu_3)^2\frac{\partial \mathcal{F}}{\partial \nu_2} \frac{\partial \mathcal{H}}{\partial \mu_2}\\
&+(2\nu^2_3+2\mu_3\nu_3+\mu^2_3)\frac{\partial \mathcal{F}}{\partial \nu_2} \frac{\partial \mathcal{H}}{\partial \nu_2} 
-(\mu_3+\nu_3)^2 \frac{\partial \mathcal{F}}{\partial \mu_2} \frac{\partial \mathcal{H}}{\partial \nu_2}+\nu^2_3 \frac{\partial \mathcal{F}}{\partial \nu_1} \frac{\partial \mathcal{H}}{\partial \nu_1}
.
\end{split}
\end{equation}
Then the irreversible dynamics  \eqref{eqofmoofmatchdouble} is computed to be
\begin{equation}
\begin{split}
\dot{\mu}_1&=\frac{\partial \mathcal{S}}{\partial \mu_1}\mu_3^2, \hspace{1cm} 
\dot{\mu_2}=\frac{\partial \mathcal{S}}{\partial \mu_2}\mu_3^2, \hspace{1cm}  \dot{\mu}_3=0   	
\\
 \dot{\nu}_1&=\frac{\partial \mathcal{S}}{\partial \nu_1} \nu_3^2,  \hspace{1cm}  \dot{\nu_2}=\frac{\partial \mathcal{S}}{\partial \nu_2} \nu_3^2, \hspace{1.2cm} \dot{\nu}_3=0.
\end{split}
\end{equation}
Let us collect the reversible Lie-Poisson dynamics in \eqref{Ray-mot-1} and the irreversible dynamics generated by $\mathcal{S}$ under the realm of the symmetric (double) bracket in \eqref{matchdoublebrofheis}.
Then the metriplectic equations of motion  are computed to be 
\begin{equation} \label{MD-Ex}
\begin{split}
 \dot{\mu}_1&=\mu_3 \frac{\partial \mathcal{H}}{\partial \mu_2}-\nu_3 \frac{\partial \mathcal{H}}{\partial \nu_2}+\frac{\partial \mathcal{S}}{\partial \mu_1}\mu_3^2 \hspace{1.3cm}
  \dot{\mu_2}=\mu_3 \frac{\partial \mathcal{H}}{\partial \nu_1}-\mu_3 \frac{\partial \mathcal{H}}{\partial \mu_1}+\frac{\partial \mathcal{S}}{\partial \mu_2}\mu_3^2,	
\\
 \dot{\nu}_1&=-\mu_3 \frac{\partial \mathcal{H}}{\partial \mu_2}-\nu_3 \frac{\partial \mathcal{H}}{\partial \nu_2}+\frac{\partial \mathcal{S}}{\partial \nu_1} \nu_3^2,
\hspace{1cm}
\dot{\nu}_2=\nu_3 \frac{\partial \mathcal{H}}{\partial \nu_1}+\nu_3 \frac{\partial \mathcal{H}}{\partial \mu_1}+\frac{\partial \mathcal{S}}{\partial \nu_2} \nu_3^2.
\end{split}
\end{equation}
After the choice $\mu_1=q$, $\mu_2=p$, $\nu_1=u$, $\nu_2=w$, and $\mu_3=\nu_3=1$ read the following system 
\begin{equation} \label{qdot}
\begin{split}
\dot{q}&=\frac{\partial \mathcal{H}}{\partial p}-\frac{\partial \mathcal{H}}{\partial w}+ \frac{\partial \mathcal{S}}{\partial q} , 
\hspace{1.3cm}  
\dot{p}=\frac{\partial \mathcal{H}}{\partial u}-\frac{\partial \mathcal{H}}{\partial q}+\frac{\partial \mathcal{S}}{\partial p}
\\
\dot{u}&=-\frac{\partial \mathcal{H}}{\partial p}-\frac{\partial \mathcal{H}}{\partial w}+ \frac{\partial \mathcal{S}}{\partial u} , 
\hspace{1cm}
  \dot{w}=\frac{\partial \mathcal{H}}{\partial u}+\frac{\partial \mathcal{H}}{\partial q}+ \frac{\partial \mathcal{S}}{\partial w}.
\end{split}
\end{equation}
Let us more particularly take the Hamiltonian function $\mathcal{H}=1/2(p^2+w^2)+V(q,u)$ to be the total energy of the system then from (\ref{qdot}) we obtain
\begin{equation}\label{2ndODE}
\begin{split}
\ddot{q}-\mathcal{S}_{qq}\dot{q}+2V_q+\mathcal{S}_w-\mathcal{S}_p=0
\\
\ddot{u}-\mathcal{S}_{uu}\dot{u}+2V_u+\mathcal{S}_w+\mathcal{S}_p=0
\end{split}
\end{equation}

  \section{Illustration: Rigid Bodies} \label{Diss-Gen-Exp}

We consider two identical three dimensional Euclidean spaces denoted by $\mathfrak{g}=\mathbb{R}^{3}$ and $\mathfrak{h}=\mathbb{R}^{3}_\mathbf{k}$. These are  Lie algebras equipped with the following Lie algebra structures
\begin{equation}\label{bracketsonsu2}
\left [ \xi,\xi' \right ]_\mathfrak{\mathbb{R}^{3}}= \xi \times \xi', \qquad \left [ \eta, \eta' \right ]_{\mathbb{R}^{3}_\mathbf{k}}=\mathbf{k} \times (\eta \times \eta'),
\end{equation}
respectively, \cite{Ma90}. 
Here, $\times$ denotes cross product on $ \mathbb{R}^{3}$, and $\textbf{k}$ is the unit vector $(0,0,1)$ available in the standard basis on $ \mathbb{R}^{3}$. Notice that, the subscript $\mathbb{R}^{3}_\mathbf{k}$ is just to remind that the bracket is not classical cross product on $\mathbb{R}^3$ instead its the second bracket in \eqref{bracketsonsu2}. We obey the notations presented in Subsection \ref{doublecross}. The coadjoint action of $\mathfrak{g}$ on it dual $\mathfrak{g}^*\simeq \mathbb{R}^{3}$, and the coadjoint action of $\mathfrak{h}$ on it dual $\mathfrak{h}^*\simeq \mathbb{R}^{3}$ are computed to be
\begin{equation}
ad^*_\xi \mu =\mu \times \xi, \qquad ad^*_{\eta} \nu =(\nu \cdot \eta)\mathbf{k}-\nu(\mathbf{k}\cdot \eta).
\end{equation}

In order to employ the notation in \eqref{nonbracket2-mp}, we introduce the following basis  $(\mathbf{e}_{1},\mathbf{e}_{2},\mathbf{e}_{3})$ on $\mathbb{R}^{3}$ and $(\mathbf{f}_{1},\mathbf{f}_{2},\mathbf{f}_{3})$ on $\mathbb{R}^{3}_\mathbf{k}$ 
\begin{equation}
\mathbf{e}_{1}=\mathbf{f}_1=\left ( 1,0,0 \right ), \quad 
\mathbf{e}_{2}=\mathbf{f}_2=\left ( 0,1,0 \right ),\quad  
\mathbf{e}_{3}=\mathbf{f}_3=\textbf{k}=\left ( 0,0,1 \right ).
\end{equation}
So that, the strcuture constants in \eqref{nonbracket2-mp} turn out to be 
\begin{equation}\label{C}
C_{12}^3=C_{31}^2=C_{23}^1=1, \qquad D_{13}^1=D_{23}^2=1,
\end{equation}
where the rest is zero. 

\textbf{Mutual actions.} 
Left action of $\mathbb{R}^{3}_\textbf{k}$ on $\mathbb{R}^{3}$, and right action $\mathbb{R}^{3}$ on $\mathbb{R}^{3}_\textbf{k}$ are defined through 
\begin{equation} \label{actionsofsu2}
\begin{split}
\vartriangleleft&: \mathbb{R}^{3}_\textbf{k} \times  \mathbb{R}^{3}  \rightarrow \mathbb{R}^{3}_\textbf{k}, \qquad  \eta \vartriangleleft \xi  :=\eta \times \xi \\
\vartriangleright&: \mathbb{R}^{3}_\textbf{k} \times  \mathbb{R}^{3}   \rightarrow \mathbb{R}^{3}, \qquad  \eta \vartriangleright \xi  :=
 \eta \times (\xi \times \textbf{k}).  
\end{split}
\end{equation}
It is straight forward to prove that these actions are indeed satisfying the matched pair Lie algebra conditions in \eqref{comp-mpa}. Referring to the notations fixed in \eqref{local-act}, the constants determining the actions (\ref{actionsofsu2}), can be computed to be 
\begin{equation}\label{L}
L_{11}^3=L_{22}^3=-1 ,\qquad L_{32}^2=L_{31}^1=1, \qquad R_{12}^3=R_{31}^2=R_{23}^1=-1,\qquad R_{13}^2=R_{21}^3=R_{32}^1=1,
\end{equation}
with the rest is zero. According to Theorem \ref{mp-prop}, one can define the matched pair Lie algebra $\mathbb{R}^{3}\bowtie \mathbb{R}^{3}_\textbf{k} $ equipped with the matched Lie bracket \eqref{mpla} calculated as
\begin{equation} \label{cerc}
\left [  \xi\oplus \eta , \xi'\oplus \eta'  \right ]_{\bowtie}
=\big(\xi \times\xi'+\eta \times(\xi'\times \mathbf{k})-\eta'\times  (\xi \times \mathbf{k})\big) \oplus\big(
 \mathbf{k}\times(\eta\times\eta')+\eta\times \xi' - \eta'\times \xi  \big).
\end{equation}

\textbf{Matched Lie-Poisson dynamics.} 
The Lie-Poisson dynamics on the dual space $\G{g}^* \times  \G{h}^* \simeq \mathbb{R}^{3}\times \mathbb{R}^{3} $ can be obtained by employing Proposition \ref{ad-*-prop}. To have this, first we compute the duals of the actions \eqref{actionsofsu2}, in a respective order, as 
\begin{equation}\label{dualactionsofsu}
\mu \overset{\ast}{\vartriangleleft} \eta=\mu(\eta\cdot \mathbf{k})-(\mu\cdot \mathbf{k})\eta, 
\qquad 
 \xi \overset{\ast}{\vartriangleright} \nu =\xi \times \nu
\end{equation}
whereas the mappings \eqref{a*} and \eqref{b*}
are
\begin{equation}\label{dualactionsofsu-2}
\mathfrak{a}^*_{\eta} \nu=\nu \times \eta, \qquad \mathfrak{b}^*_\xi \mu=(\mu \cdot \xi)\mathbf{k}-(\mu \cdot \mathbf{k})\xi.
\end{equation}
Notice that, the Lie-Poisson formulation on the dual of $\G{g}^*$ corresponds to the the rigid body dynamics in $3D$, \cite{Ho08,MaRa13}. So we can consider the matched pair dynamics in this setting as the coupling of two bodies in $3D$. Since the dynamics for the rigid bodies are given by minus Lie-Poisson equation, we refer minus Lie-Poisson bracket and minus Lie-Poisson equations. 
Assume the following coordinates $\mu=(\mu_1,\mu_2,\mu_3)$ and  $\nu=(\nu_1,\nu_2,\nu_3)$. 
Recall that, in \eqref{cosym2}, the explicit realization  of the matched Lie-Poisson bivector, defined for the matched pair Lie-Poisson bracket \eqref{Lie-pois-double-non-bracket-tc}, has already been given. In the following table, we exhibit the coefficients of the matched Lie-Poisson bivector for the present case. 
 \begin{center}
 \begin{tabular}{SSSSS} 
 	\toprule[1.8pt]
	{$\Lambda$} & {$\Lambda_{\alpha \beta}$} & {$\Lambda_{ \alpha b }$} & {$\Lambda_{ a \beta  }$} & {$\Lambda_{ a b }$}    \\ \midrule[1.5pt] 
	{$\Lambda_{11}$}  & 0 & {$\mu_3$} &{$-\mu_3$} &0  \\
	{$\Lambda_{12}$}  &{-$\mu_3$}   & {-$\nu_3$} &{$\nu_3$}   & 0    \\
	 {$\Lambda_{13}$} & {$\mu_2$}  & {$\nu_2-\mu_1$} & {$-\nu_2+\mu_1$}  & {$\nu_1$}   \\ \midrule
	 {$\Lambda_{21}$} &{$\mu_3$}   & {$\nu_3$} & {$-\nu_3$}  & 0    \\ 
	 {$\Lambda_{22}$} & 0  & {$\mu_3$} & {$-\mu_3$}   & 0   \\
	 $\Lambda_{23}${} & {$\mu_1$}  & {-$\nu_1-\mu_2$} & {$\nu_1+\mu_2$}  & {$\nu_2$}  \\ \midrule
	 {$\Lambda_{31}$} & {-$\mu_2$}  & {-$\nu_2$} & {$\nu_2$}  & {$-\nu_1$}    \\
	 {$\Lambda_{32}$} &  {-$\mu_1$} & ${\nu_1}$ & ${-\nu_1}$  & {$-\nu_2$}      \\ 
	 {$\Lambda_{33}$} & 0  & 0 & 0  & 0  \\
	 \bottomrule
\end{tabular}
\end{center}	
We remark that the first column determined the Lie-Poisson bivector on $\G{g}^*$ whereas the last one is the Lie-Poisson bivector on $\G{h}^*$ whereas the second and third columns are manifesting the mutual actions \eqref{actionsofsu2}.   Collecting all these results, we compute the matched Lie-Poisson equations (\ref{LPEgh}) generated by a Hamiltonian function $\mathcal{H}$ on the dual space as 
\begin{equation}
\begin{split} 
\dot{\mu}&=\frac{\partial \mathcal{H}}{\partial\mu}\times \mu+(\frac{\partial \mathcal{H}}{\partial \nu}\cdot \mathbf{k})\mu-(\mu\cdot \mathbf{k})\frac{\partial \mathcal{H}}{\partial \nu}-\nu \times \frac{\partial \mathcal{H}}{\partial \nu}, \\
\dot{\nu}&=(\mathbf{k} \cdot \frac{\partial \mathcal{H}}{\partial \nu})\nu-(\nu \cdot \frac{\partial \mathcal{H}}{\partial \nu})\mathbf{k}+\frac{\partial \mathcal{H}}{\partial \mu}\times \nu+(\mu \cdot \mathbf{k})\frac{\partial \mathcal{H}}{\partial \mu}-(\mu \cdot \frac{\partial \mathcal{H}}{\partial \mu})\mathbf{k}.
\end{split}
\end{equation}

In Section \ref{Sec-C-DS}, couplings of various ways of  dissipations are listed. 
We examine now these couplings for the concrete example given in previous subsection.  

\textbf{Rayleigh dissipation.}
In this $3D$ framework, and for a linear operator
\begin{equation}
(\mathbb{R}^{3})^*\oplus (\mathbb{R}^{3}_{\mathbf{k}})^* \longrightarrow \mathbb{R}^{3} \bowtie \mathbb{R}^{3}_{\mathbf{k}}, \qquad \mu\oplus\nu \mapsto (\Upsilon^{\G{g}}(\mu)\oplus\Upsilon^{\G{h}}(\nu)),
\end{equation}
the matched Lie-Poisson systems with Rayleigh type dissipations, that is the system \eqref{eqofmoofRay}, takes the particular form
\begin{equation}
\begin{split}
&\dot{\mu}-\frac{\partial \mathcal{H}}{\partial \mu}\times \mu-\mu(\frac{\partial \mathcal{H}}{\partial \nu}\cdot \mathbf{k})+(\mu\cdot \mathbf{k})\frac{\partial \mathcal{H}}{\partial \nu}+\nu \times \frac{\partial \mathcal{H}}{\partial \nu}=\mu \times \Upsilon^{\G{g}}(\mu)-\mu(\Upsilon^{\G{h}}(\nu)\cdot \mathbf{k})+(\mu\cdot \mathbf{k})\Upsilon^{\G{h}}(\nu)+\nu \times \Upsilon^{\G{h}}(\nu)
\\
&\dot{\nu}-\nu(\mathbf{k}\cdot \frac{\partial \mathcal{H}}{\partial \nu})+(\nu \cdot \frac{\partial \mathcal{H}}{\partial \nu})\mathbf{k}-\frac{\partial \mathcal{H}}{\partial \mu}\times \nu -(\mu \cdot \mathbf{k})\frac{\partial \mathcal{H}}{\partial \mu}+(\mu \cdot \frac{\partial \mathcal{H}}{\partial \mu})\mathbf{k}  = (\nu \cdot \Upsilon^{\G{h}}(\nu))\mathbf{k}-\nu(\mathbf{k}\cdot \Upsilon^{\G{h}}(\nu))\\
& \hspace{9cm}- \Upsilon^{\G{h}}(\nu) \times \nu +(\mu \cdot \Upsilon^{\G{g}}(\mu))\mathbf{k}-(\mu \cdot \mathbf{k})\Upsilon^{\G{g}}(\mu).
\end{split}
\end{equation}		

\textbf{Cartan-Killing dissipation.} Determine 
the matched Cartan-Killing metric \eqref{CM-MP} as 
\begin{equation}
[\C{G}_{ij}]=\begin{bmatrix}
-4 & 0 & 0 &  0 & -5 & 0\\
0 & -4 & 0 & 5 & 0 & 0\\
0 & 0 & -4 & 0 & 0 & 0\\
0 & 5 & 0 & -2 & 0 & 0\\
-5 & 0 & 0 & 0 & -2 & 0\\
0 & 0 & 0 &  0 & 0 & 1 \\
\end{bmatrix}.
\end{equation}
Then, irreversible dynamics generated by a function $\C{S}$, presented in (\ref{eqofmoofcart}), is computed to be 
\begin{equation}
\begin{split}
&\dot{\mu}_1=-4\frac{\partial \mathcal{S}}{\partial \mu_1}+5\frac{\partial \mathcal{S}}{\partial \mu_2},
 \hspace{1cm} 
 \dot{\mu}_2=-4\frac{\partial \mathcal{S}}{\partial \mu_2}-5\frac{\partial \mathcal{S}}{\partial \mu_1},
  \hspace{1cm}
 \dot{\mu}_3=-4\frac{\partial \mathcal{S}}{\partial \mu_3},\\
&\dot{\nu}_1=-5\frac{\partial \mathcal{S}}{\partial \nu_2}-2\frac{\partial \mathcal{S}}{\partial \nu_1}, 
\hspace{1.1cm}
\dot{\nu}_2=5\frac{\partial \mathcal{S}}{\partial \nu_1}-2\frac{\partial \mathcal{S}}{\partial \nu_2},
 \hspace{1.38cm}
\dot{\nu}_3=\frac{\partial \mathcal{S}}{\partial \nu_3}.
\end{split}
\end{equation} 

\section{Acknowledgments}

This paper is a part of the project "Matched pairs of Lagrangian and Hamiltonian Systems" supported by T\"UB\.ITAK (the Scientific and Technological Research Council of Turkey) with the project number 117F426. 
All of the authors are grateful to T\"UB\.ITAK for the support.  

The first named author (OE) is grateful to Prof. Miroslav Grmela, Prof. Parta Guha, Prof. Michal Pavelka, and  Prof. Petr Vágner for enlightening discussions on GENERIC. We are grateful to Begüm Ateşli for discussions done on extension theories. 

\bibliographystyle{plain}

\end{document}